\documentclass[draft]{article}
\usepackage{amsmath,amsfonts,latexsym, amssymb,  mathrsfs, eucal}

\usepackage{color}

\newcommand{\wt}{\widetilde}
\newcommand{\wh}{\widehat}

\newcommand{\beqa}{\begin{eqnarray}}
\newcommand{\eeqa}{\end{eqnarray}}
\newcommand{\e}{\varepsilon}
\newcommand{\eps}{\varepsilon}
\newcommand{\pt}{\partial}
\newcommand{\rd}{{\rm d}}
\newcommand{\bR}{{\mathbb R}}
\newcommand{\bC}{{\mathbb C}}

\newcommand{\non}{\nonumber}
\newcommand{\wH}{{K}}

\newcommand{\tr}{\mbox{Tr\,}}

\newcommand{\ba}{{\bf{a}}}

\newcommand{\bx}{{\bf{x}}}

\newcommand{\bT}{{\T}}

\newcommand{\mg}{{m_N}}

\newcommand{\al}{\alpha}

\newcommand{\be}{\begin{equation}}
\newcommand{\ee}{\end{equation}}

\newcommand{\ga}{{\gamma}}

\newcommand{\la}{\lambda}

\newcommand{\om}{{\omega}}

\newcommand{\cL}{{\mathscr L}}

\newcommand{\cX}{{\mathcal X}}
\newcommand{\cY}{{\mathcal Y}}

\newcommand{\cN}{{\mathcal N}}
\newcommand{\cH}{{\mathcal H}}

\newcommand{\ov}{\overline}

\newcommand{\re}{{\mathfrak{Re} \, }}
\newcommand{\im}{{\mathfrak{Im} \, }}
\newcommand{\E}{{\mathbb E }}
\newcommand{\R}{{\mathbb R }}
\newcommand{\N}{{\mathbb N}}

\newcommand{\Ci}{{ C_{inf}}}
\newcommand{\Cs}{{ C_{sup}}}

\renewcommand{\P}{{\mathbb P}}
\newcommand{\C}{{\mathbb C}}

\renewcommand{\S}{\mathbb S}
\newcommand{\T}{\mathbb T}
\newcommand{\U}{\mathbb U}

\newcommand{\bS}{\bf  S}

\newtheorem{theorem}{Theorem}

\newtheorem{lemma}[theorem]{Lemma}

\newtheorem{definition}{Definition}
\newcommand{\qed}{\hfill\fbox{}\par\vspace{0.3mm}}
\newenvironment{proof}{{\bf Proof.}} {\hfill\qed}

\numberwithin{equation}{section}
\numberwithin{theorem}{section}
\numberwithin{definition}{section}
\numberwithin{remark}{section}

\title{Universality for generalized Wigner matrices \\ with Bernoulli 
distribution}

\author{
L\'aszl\'o Erd\H os${}^1$\thanks{Partially supported
by SFB-TR 12 Grant of the German Research Council}, 
Horng-Tzer Yau${}^2$\thanks{Partially supported
by NSF grants DMS-0757425, 0804279}  \; and Jun Yin${}^2$ \\ \\
Institute of Mathematics, University of Munich, \\
Theresienstr. 39, D-80333 Munich, Germany \\ lerdos@math.lmu.de ${}^1$ \\ \\
Department of Mathematics, Harvard University\\
Cambridge MA 02138, USA \\  htyau@math.harvard.edu,  jyin@math.harvard.edu ${}^2$ \\ \\
\\}

\begin{document}

\date{Aug 11, 2010}

\maketitle

\begin{abstract}

The universality for  the eigenvalue spacing 
statistics of  generalized  Wigner matrices  was established in our previous work \cite{EYY} 
under  certain  conditions on the 
probability distributions of the matrix elements. 
A major class of probability measures  excluded in \cite{EYY}
 are the Bernoulli  measures.  
In this paper, we extend the universality result of  \cite{EYY} 
 to include the Bernoulli  measures  so that the only 
restrictions on the probability distributions of the matrix elements 
are  the subexponential decay and the normalization condition
that the variances in each row sum up to one.
 The new ingredient is a strong 
local semicircle law which improves the error estimate on  the
 Stieltjes  transform of the empirical 
measure of the eigenvalues from the  order
  $(N \eta)^{-1/2}$ to  $(N \eta)^{-1}$. 
Here $\eta$ is the imaginary part  of the spectral parameter
in the definition of the Stieltjes
 transform and $N$ is the size of the matrix.

\end{abstract}

{\bf AMS Subject Classification:} 15A52, 82B44

\medskip


\medskip

{\it Keywords:}  Random band  matrix, Local semicircle law,
sine kernel.

\medskip


\setcounter{tocdepth}{2}

\newpage
\section{Introduction}

The universality of  local eigenvalue statistics  in the bulk  of the spectrum
of random matrices has been 
traditionally considered  only for invariant ensembles 
\cite{BI, DKMVZ1, DKMVZ2, PS}. For non-invariant ensembles, a new approach to prove 
the bulk universality 
was developed in  \cite{EPRSY, ESY4, ESYY, EYY}. It consists of  the following three  steps:

\begin{enumerate}  
\item Local semicircle law.

\item  Universality for Gaussian divisible ensembles.

\item  Approximation by Gaussian divisible ensembles. 
\end{enumerate}

\noindent 
In Step 2,  the universality of the 
local eigenvalue statistics for a large class of matrices, i.e., Gaussian
 divisible matrices,  was established.  Thus in order to prove the universality
 of a given ensemble,  it remains to 
approximate the matrix elements in  this  ensemble  by  Gaussian divisible distribution in such 
a way  that the local eigenvalue statistics are unchanged. This approximation is 
intrinsically a density theorem and it
can be achieved by perturbative expansions in several different ways.
 In  the most recent approach \cite{ESYY,EYY}, the universality 
for Gaussian divisible ensembles was proved via the Dyson Brownian motion and
the stability of eigenvalues in Step 3 was provided by the Green function comparison theorem. 
In  Step 2  a technical tool, the logarithmic Sobolev inequality  (LSI), was needed 
to estimate the fluctuations of  eigenvalue distribution.  This restriction 
could not be completely removed in Step 3 and thus the Bernoulli measures were 
excluded in \cite{EYY}. 
 In this paper, we will improve the local semicircle law so that the LSI
is no longer needed.  This will enable us to prove the universality for
 generalized Wigner matrices 
with Bernoulli distributions. 
As a byproduct of the new stronger form of local semicircle law, we
 also obtain much stronger estimates 
on the  eigenvalue density and on the matrix elements of
the resolvent.

Recall 
the Stieltjes  
transform of the empirical 
measure of the eigenvalues $\{\lambda_j\}_{j=1}^N$ is defined by 
$$ 
   m_N(z) = \frac{1}{N}\sum_{j=1}^N\frac{1}{\lambda_j - z}.
$$
We have proved in \cite{EYY} that the difference between 
 $ m_N(z)$ and $ m_{sc}(z)$,  the  Stieltjes  
transform of the semicircle law \eqref{temp2.8}, is  bounded  by
 $(N \eta)^{-1/2}$ where $\eta = \im z$. The main result of this paper 
states that the error can be improved to  $(N \eta)^{-1}$. 
 The improvement of  a factor $ (N \eta)^{-1/2}$
resembles the
usual $N^{-1/2}$ factor in the central limit theorem and it results from 
a new estimate  on the correlations of error terms. 
This estimate also implies that  the error between 
the  normalized empirical counting function  of the eigenvalues and the 
one given by the semicircle law is less than $N^{-1+ \e}$ in the bulk of the spectrum 
for any $\e> 0$.  This new input is sufficiently strong to replace the usage of 
the (LSI) in \cite{EYY}, see 
the discussion after Theorem \ref{mainsk} for more details.

Notice that this improvement of a factor $ (N \eta)^{-1/2}$ and the removal of the LSI need
 a  substantial amount of work.   
Our motivations to take on this endeavor are for the following two reasons: (1) The 
 distributions of the Bernoulli random matrices are 
very singular while  the Gaussian measures in GOE are very smooth.
 It is not a priori  clear that the universality holds 
for such singular distributions. (2) The adjacency
 matrices for random graphs are natural  examples of 
 symmetric random matrices. 
The matrix elements of these matrices take the values $0$ or $1$ and thus they form  Bernoulli 
random matrices.  Our  current results
 do not cover this case since we require the mean zero condition, but 
they represent  the first step toward the universality of  the  adjacency matrices of random graphs.

\section{Main results}

We now state the main results of this paper. Since all our results hold for both hermitian
 and symmetric ensembles, 
we will state the results for the hermitian case only. The modifications to the 
symmetric case are straightforward and they will be omitted. 
Let $H=(h_{ij})_{i,j=1}^N$  be an $N\times N$  hermitian matrix where the
 matrix elements $h_{ij}=\ov{h}_{ji}$, $ i \le j$, are independent 
random variables given by a probability measure $\nu_{ij}$ 
with mean zero and variance $\sigma_{ij}^2$. 
The variance of $h_{ij}$ for $i>j$ is $\sigma_{ij}^2 =\E\,  |h_{ij}|^2 = \sigma_{ji}^2$.
For simplicity of the presentation, 
we assume that for any fixed $1\leq i<j\leq N$, ${\rm Re}\,h_{ij}$ 
and ${\rm Im}\,h_{ij}$ are i.i.d.
with distribution $\om_{ij}$,  i.e., $\nu_{ij} = \om_{ij}\otimes\om_{ij}$ 
in the sense that $\nu_{ij}(\rd h) = \om_{ij}(\rd {\rm Re}\,h)
\om_{ij}(\rd {\rm Im}\,h)$, but this assumption is
not essential for the result.
The distribution $\nu_{ij}$ and its variance $\sigma_{ij}^2$ may depend on $N$,
 but we omit this fact in the notation.  We assume that for any $j$ fixed
\be
   \sum_{i} \sigma^2_{ij} = 1 \, .
\label{sum}
\ee
Matrices with independent, zero mean  entries and with the normalization  condition
\eqref{sum} will be called {\it universal  Wigner matrices.}
The basic parameter of such matrices is the quantity
\be\label{defM}
M:= \frac{1}{\max_{ij} \sigma_{ij}^2}.
\ee

Define $\Ci$ and $\Cs$ by
\be\label{defCiCs}
      \Ci:= \inf_{N, i,j}\{N\sigma^2_{ij}\}\leq \sup_{N, i,j}\{N\sigma^2_{ij}\}=:\Cs.
\ee
Note that $\Ci=\Cs(=1)$ corresponds to the standard Wigner matrices and
the conditions $0< \Ci \le \Cs <\infty$ define more general Wigner matrices
with comparable variances.

We will also consider an even more general case when $\sigma_{ij}$
for different $(i,j)$ indices are not comparable. 
A special case is the {\it  band matrix}, where  $\sigma_{ij}=0$ for $|i-j|>W$
with some parameter $W$.

Denote by $\Sigma:=\{ \sigma^2_{ij}\}_{i,j=1}^N$ the matrix of variances 
which is symmetric,  doubly stochastic by \eqref{sum}, and in particular  satisfies
$-1\leq \Sigma \leq 1$. 
Let the spectrum of $\Sigma$ be supported in 
\be\label{de-de+}
\mbox{Spec}(\Sigma)\subset [-1+\delta_-, 1-\delta_+]\cup\{1\}
\ee  
with some nonnegative constants $\delta_\pm$. 
We will always have the following spectral assumption
\be\label{speccond}
\mbox{\it  1 is a simple eigenvalue of $\Sigma$ and
$\delta_-$ is a positive constant,  independent of $N$.}
\ee

The local semicircle law will be proven under this general condition, but
the precision of the estimate near the spectral edge will also depend on $\delta_+$
in an explicit way. For the orientation of the reader, we 
mention two special cases that provided the main motivation
for our work.

\bigskip

One  important class of universal  Wigner matrices  is  the  {\it generalized  Wigner ensemble}  
which is defined by the extra condition that
\be\label{VV}
	0<\Ci\leq \Cs<\infty,
\ee
It is easy to check that
 \eqref{de-de+} holds with
\be\label{de-de+2}
	\delta_\pm \ge C_{inf}.
\ee 
Another example is the {\it band matrix ensemble} whose variances are given by 
\be\label{BM}
   \sigma^2_{ij} = W^{-1} f\Big(\frac{ [i-j]_N}{W}\Big),
\ee
where $W\ge 1$, $f:\bR\to \bR_+$ is a nonnegative symmetric function with 
$\int f =1$, $f \in  L^\infty (\bR)$,
and we defined $[i-j]_N\in \{1, 2,\ldots N\}$ by the property
that  $[i-j]_N\equiv i-j \; \mbox{mod}\; N$.
The bandwidth $M$ defined in \eqref{defM} satisfies $M\le  W/\|f\|_\infty$.
In Appendix A of \cite{EYY},  we have proved that \eqref{speccond} is
 satisfied for the choice of \eqref{BM}
if $W$ is large enough.

\bigskip

Define  the Stieltjes transform of the empirical 
eigenvalue distribution of  $H $ by 
$$
 m(z)=\mg (z): = \frac{1}{N} \tr\, \frac{1}{H-z}\,,\,\,\, z=E+i\eta.
$$
Define $m_{sc} (z)$ as the unique solution of
$$
m_{sc} (z)  + \frac{1}{z+m_{sc} (z)} = 0, 
$$
with positive imaginary part for all $z$ with $\text{Im } z > 0$, i.e.,
\be\label{temp2.8}
m_{sc}(z)=\frac{-z+\sqrt{z^2-4}}{2}.
\ee
Here the square root function is chosen with a branch cut in the segment
$[-2,2]$ so that asymptotically $\sqrt{z^2-4}\sim z$ at infinity.
This guarantees that the imaginary part of $m_{sc}$ is non
 negative for $\text{Im } z > 0$ and it is the 
Wigner semicircle distribution
\be
   \varrho_{sc}(E) : = \lim_{\eta\to 0+0}\frac{1}{\pi}\im \, m_{sc}(E+i\eta)
 = \frac{1}{2\pi}
  \sqrt{ (4-E^2)_+}.
\label{def:sc}
\ee
 The Wigner semicircle law \cite{W} states that  $m_N(z) \to m_{sc} (z) $
for any fixed $z$, i.e.,
 provided that $\eta$ is independent of $N$. 
 We have proved \cite{EYY} a  local version of this result for universal Wigner
matrices and the main result can be stated as the following probability estimate: 
$$
\P\left(|\mg(z)-m_{sc}(z)|\geq (\log N)^{C_2}
 \frac{1}{ \sqrt {M\eta}\,\kappa}\right)\leq CN^{-c(\log \log  N)}
$$
with some constant $C_2$.
The accuracy of this estimate can be improved from 
$ (M\eta)^{-1/2}\,\kappa^{-1}$ to   $(M\eta)^{-1}\,\kappa^{-1} $, 
which is the content of the next theorem. It summarizes
the results of Theorems~\ref{thm:detailed} and   \ref{prop:mmsc}.
Prior  to our result in \cite{EYY}, a central limit theorem  for
the semicircle law on macroscopic scale for band matrices was established
by Guionnet \cite{gui} and Anderson and Zeitouni \cite{AZ}; a
 semicircle law  for Gaussian band matrices  was proved 
by Disertori, Pinson and Spencer \cite{DPS}.
 For a review on band matrices, see  the recent  article \cite{Spe} by Spencer.

\begin{theorem}[Local semicircle law] \label{lsc}{}
Let $H$ be a hermitian  $N\times N$ random matrix
with $\E\, h_{ij}=0$, $1\leq i,j\leq N$,  and assume that the variances $\sigma_{ij}^2$ 
satisfy \eqref{sum} and \eqref{speccond}. 
 Suppose that the distributions of the matrix elements have a uniformly 
  subexponential decay
in the sense that  there exist constants $\al$, $\beta>0$, independent 
of $N$, such that for any $x> 0$  we have 
\be\label{subexp}
\P(|h_{ij}|\geq x^\al |\sigma_{ij}|)\leq \beta e^{- x}.
\ee
We consider universal Wigner matrices and its  special class, the generalized Wigner matrices 
in parallel. The parameter $A$ will distinguish between the two cases;
we set $A=2$ for universal Wigner matrices, and $A=1$ for generalized Wigner
matrices, where the results will be stronger.

Define the following domain in $\C$ 
\be\label{fakerelkaeta}
    D : = \Big\{ z=E+i\eta\in \C\; : \;  |E|\le 5, \;  0 < \eta\le 10, \; \; 
  \sqrt{M\eta}\ge (\log N)^{C_1} (\kappa +\eta)^{\frac{1}{4}-A}
\Big\}
\ee
where  $\kappa : = \big| \, |E|-2 \big|$.
Then there exist constants $C_1$, $C_2$, $C$ and $c>0$, 
depending only on $\al$, $\beta$ and $\delta_-$ in \eqref{speccond}, such that 
for any $\e>0$ and $K>0$  the Stieltjes transform of the empirical 
eigenvalue distribution of  $H $  satisfies 
\be\label{mainlsresultfake}
\P\left(\bigcup_{z\in D} \Big\{ |\mg(z)-m_{sc}(z)|\geq
 \frac{N^\e}{{M\eta}\,(\kappa+\eta)^A}\Big\}\right)\leq \frac{C(\e, K)}{N^K}
\ee
for sufficiently large $N$.
Furthermore,  the diagonal matrix elements of
the Green function $G_{ii}(z) = (H-z)^{-1}(i,i)$ satisfy that 
\be\label{Gii}
\P\left(\bigcup_{z\in D}\Big\{ \max_i | G_{ii}(z)-m_{sc}(z)|\geq 
 \frac{ (\log N)^{C_2}} {\sqrt{M\eta}} \, 
 (\kappa+\eta)^{\frac{1}{4}-\frac{A}{2}}\Big\}\right)\leq CN^{-c(\log \log  N)}
\ee
and for the off-diagonal elements we have
\be\label{Gij}
\P\left( \bigcup_{z\in D}\Big\{ \max_{i\ne j} | G_{ij}(z)|\geq
 \frac{ (\log N)^{C_2}} {\sqrt{M\eta}} \,  (\kappa+\eta)^{\frac{1}{4}} 
\Big\} \right)\leq CN^{-c(\log \log  N)}
\ee
for any sufficiently large $N$.
\end{theorem}

  The subexponential decay condition \eqref{subexp} can also be
 easily weakened  if we are not aiming 
at error estimates faster than any power law of $N$. 
This can be easily carried out and we will not pursue  it in this paper.

\bigskip

Denote the eigenvalues of $H$ by $\lambda_1,  \ldots , \lambda_N$
and let $p_N(\lambda_1,  \ldots , \lambda_N)$ be 
their (symmetric) probability  density. 
For any
$k=1,2,\ldots, N$ the $k$-point correlation function  of the eigenvalues is defined by 
\be
 p^{(k)}_N(x_1, x_2,\ldots x_k):=
\int_{\bR^{N-k}} p_N(x_1, x_2, \ldots , x_N)\rd x_{k+1}\ldots \rd
x_N.
\label{corrfn}
\ee
We now state our main result concerning these correlation functions. 
The same result  was proved in \cite{EYY} under the  additional assumption \eqref{no3no4}.

\begin{theorem}[Universality for generalized  Wigner matrices] \label{mainsk}
Consider a generalized \, hermitian \\  Wigner ensemble  such that \eqref{sum},
 \eqref{speccond}
and  \eqref{VV} hold.
 Suppose that the distributions $\nu_{ij}$  of the matrix elements 
 have a uniformly   subexponential decay
in the sense of \eqref{subexp}.  Suppose that the
real and imaginary parts of $h_{ij}$ are i.i.d., distributed
according to $\om_{ij}$, i.e.,
$\nu_{ij}(\rd h) = \om_{ij}(\rd \im h)\om_{ij}(\rd \re h)$.
  Then for any $k\ge 1$ and for any compactly
 supported continuous test function
$O:\bR^k\to \bR$ we have 
\be
\begin{split}
 \lim_{b\to0}\lim_{N\to \infty} \frac{1}{2b}
\int_{E-b}^{E+b}\rd E' \int_{\R^k} & \rd\alpha_1 
\ldots \rd\alpha_k \; O(\alpha_1,\ldots,\alpha_k)  \\
&\times
 \frac{1}{\varrho_{sc}(E)^k} \Big ( p_{N}^{(k)}  - p_{GU\! E, N} ^{(k)} \Big )
  \Big (E'+\frac{\alpha_1}{N\varrho_{sc}(E)}, 
\ldots, E'+\frac{\alpha_k}{N \varrho_{sc}(E)}\Big) =0,
\label{matrixthm}
\end{split}
\ee
where  $p_{GU\! E, N} ^{(k)}$ is the $k$-point correlation
function for the GUE ensemble. The same statement holds for
symmetric matrices, with GOE replacing the GUE ensemble.
\end{theorem}

\medskip

\noindent 
{\bf Remark.} We can take $b = N^{-c}$ for some small constant
 $c>0$ so that there is no double limit taken. 
This is because  all our bounds have an effective error estimate 
 $N^{-c}$.  In  case of hermitian matrices there is no need
for averaging in the energy parameter $E'$.
The limit \eqref{matrixthm} holds even for any
fixed energy $E'$, with $|E'|<2$, since, 
instead of relying on the local relaxation flow of \cite{ESY4, ESYY},
  we can use the result of  
\cite{EPRSY} for Gaussian divisible 
ensembles at a fixed energy. 
\medskip

It is well-known that the limiting correlation functions
 of the GUE ensemble are given
by the sine kernel
$$
     \frac{1}{\varrho_{sc}(E)^k} p_{GU\! E, N} ^{(k)}
 \Big (E+\frac{\alpha_1}{N\varrho_{sc}(E)}, 
\ldots, E+\frac{\alpha_k}{N \varrho_{sc}(E)}\Big)
 \to \det\{ K(\al_i-\al_j)\}_{i,j=1}^k, \qquad K(x) = \frac{\sin \pi x}{\pi x},
$$
and a similar universal formula is available for the limiting gap distribution. 
The formulas for the GOE cases 
are more complicated and we refer the reader to standard references such as 
 \cite{AGZ, De2, For, M}.

We will prove Theorem~\ref{mainsk} using the approach of \cite{ESYY, EYY}. 
The logarithmic Sobolev inequality was an important tool  in these papers and 
 it was the main obstacle why the case of Bernoulli random matrices were not covered.
We note that the Bernoulli distribution satisfies the discrete version of the LSI
but it would not be sufficient for our purposes.  To   explain the necessity of LSI, 
we  now review the three basic ingredients of the approach of \cite{ESYY, EYY}.

\begin{enumerate}

\item[{Step 1.}]  {\it Local semicircle law:} It states that the density of eigenvalues
is  given by
the semicircle law down to short scales containing 
only $N^\e$ eigenvalues for all $\e> 0$, where $N$ is the size of the matrix. 

\item[{Step 2.}]  {\it  Local ergodicity of the Dyson Brownian motion:} 
 The Dyson Brownian motion is given by  the flow 
\be\label{matrixdbm}
H_t = e^{-t/2} H_0 + (1-e^{-t})^{1/2}\, V,
\ee
where $H_0$  is the initial  Wigner matrix, 
$V$ is an independent standard GUE (or GOE) matrix
and $t\ge 0$ is the time.
Here we have used the version that the dynamics of the matrix element 
is given by an Ornstein-Uhlenbeck (OU) process  on $\bC$. 
More precisely, let 
\be\label{H}
\mu=\mu_N(\rd{\bf x}):=
\frac{e^{-\cH({\bf x})}}{Z_\beta}\rd{\bf x},\qquad \cH({\bf x}) = \cH_N(\bx):=
N \left [ \beta \sum_{i=1}^N \frac{x_{i}^{2}}{4} -  \frac{\beta}{N} \sum_{i< j}
\log |x_{j} - x_{i}| \right ]
\ee
be the probability measure of the eigenvalues 
$\bx = (x_1, x_2, \ldots, x_N)$ 
of the general $\beta$ ensemble, $\beta\ge 1$ 
 ($\beta=2$ for the hermitian case and $\beta=1$ for the symmetric case). 
 Denote the distribution of 
the eigenvalues  of $H_t$ at  time $t$
by $f_t ({\bf x})\mu(\rd {\bf x})$.
Then $f_t=f_{t,N}$ satisfies \cite{Dy}
\be\label{dy}
\partial_{t} f_t =  \cL f_t.
\ee
where 
\be
\cL=\cL_N:=   \sum_{i=1}^N \frac{1}{2N}\partial_{i}^{2}  +\sum_{i=1}^N
\Bigg(- \frac{\beta}{4} x_{i} +  \frac{\beta}{2N}\sum_{j\ne i}
\frac{1}{x_i - x_j}\Bigg) \partial_{i}.
\label{L}
\ee

We now recall the following  theorem concerning the universality of 
the Dyson Brownian motion.   Following the convention in \cite{ESYY}, we
 label the assumptions as 
Assumptions II--IV since the Assumption I, a convexity property of the 
Hamiltonian for the invariant measure 
of the Dyson Brownian motions,  is automatically satisfied for any $\beta$ ensembles.

{\bf Assumption II.} For any fixed $a,b\in \bR$, we have 
\be
\lim_{N\to\infty}   \sup_{t\ge 0}  \Bigg|
\int \frac{1}{N}\sum_{j=1}^N {\bf 1} ( x_j \in [a, b]) f_t(\bx)\rd\mu(\bx)
- \int_a^b \varrho_{sc}(x) \rd x \Bigg| =0.
\label{assum1}
\ee
where $\varrho_{sc}$ is the density of the semicircle law \eqref{def:sc}.

Let $\gamma_j =\gamma_{j,N}$ denote the location of the $j$-th point
under the semicircle law, i.e., $\gamma_j$ is defined by
\be\label{def:gamma}
 N \int_{-\infty}^{\gamma_j} \varrho_{sc}(x) \rd x = j, \qquad 1\leq j\le N. \quad
\ee
We will call $\gamma_j$ the {\it classical location} of the $j$-th point.

\bigskip
{\bf Assumption III.} There exists an $\e>0$ such that
\be
 \sup_{t\ge 0} \int  \frac{1}{N}\sum_{j=1}^N(x_j-\gamma_j)^2
 f_t(\rd \bx)\mu(\rd \bx) \le CN^{-1-2\e}
\label{assum3}
\ee
with a constant $C$ uniformly in $N$.

\bigskip

The final assumption is an upper bound on the local density. 
For any $I\in \R$, let
$$
    \cN_I: = \sum_{i=1}^N {\bf 1}( x_i \in I)
$$
denote the number of eigenvalues in $I$.

\bigskip

{\bf Assumption IV.} For any compact subinterval $I_0\subset (-2,2)
=\{ E\;: \; \varrho_{sc}(E)>0\}$,
and for  any $\delta>0$,  $\sigma>0$
there are constant $C_n$, $n\in \N$, depending on $I_0$,
$\delta$ and $\sigma$ such that for any interval $I\subset I_0$ with
$|I|\ge N^{-1+\sigma}$ and for any $K\ge 1$, we have
\be
   \sup_{\tau \ge N^{-2\e+\delta}}
    \int {\bf 1}\big\{ \cN_I \ge KN|I| \big\}f_\tau \rd\mu
\le C_{n} K^{-n}, \qquad n=1,2,\ldots,
\label{ass4}
\ee
where $\e$ is the exponent from Assumption III and $\sigma$ 
and $\delta$ are arbitrarily small numbers. 
\bigskip

We have proved \cite{EYY} that  Assumption IV follows  from the local semicircle law
and    Assumption III also follows from  the local semicircle law provided that a 
uniform LSI for the distributions of the matrix elements is assumed.

\item[{Step 3.}]
  {\it Green function comparison theorem:} It  asserts  that the correlation functions 
of the eigenvalues of two matrix ensembles are identical up to  the
 scale $1/N$  provided that the first four moments 
of the matrix elements  of these two ensembles are almost identical.
 Given this theorem and the universality for the Dyson Brownian 
motion for $t \sim N^{-\e}$,  the universality  for a matrix ensemble $H$ holds if 
we can find  another matrix ensemble $H_0$ such that the first four moments
 of the matrix elements of $H$ and $H_t$ (given by \eqref{matrixdbm})
are almost the same. Furthermore, $H_0$ is required to satisfy
 a uniform LSI so that the Assumption III can be verified. 
This is possible if the first four moments of $H_0$ satisfy 
 \be\label{no3no4}
\inf_{N} \min_{1\leq i,j\leq N}\left\{\frac{m_4(i,j)}{(m_2(i,j))^2}-
\frac{(m_3(i,j))^2}{(m_2(i,j))^3}\right\}>1,
 \ee
where $m_k(i,j)$ is the $k$-th moment of 
 the $i, j$ matrix element in the symmetric case.
In the hermitian case, the moments of the real and
imaginary parts have to satisfy \eqref{no3no4}.

\end{enumerate} 

Combining these ingredients,  the universality of local eigenvalue 
statistics in the bulk was proved for all 
generalized  Wigner ensembles (see \eqref{VV} for  the definition)
  satisfying  \eqref{no3no4} and a subexponential decay technical condition.
The restriction  \eqref{no3no4} was needed to guarantee the existence
 of a matching matrix ensemble whose matrix element
distributions satisfy the LSI so that 
 the Assumption III can be verified.  
The local semicircle estimates in Theorem~\ref{lsc} imply that the empirical counting function of
the eigenvalues is close to the semicircle counting function (Theorem~\ref{prop:count})
and that the location of the eigenvalues are close to their classical 
location in mean square deviation sense (Theorem \ref{prop:lambdagamma}).
This provides a direct proof to the  Assumption III \eqref{assum3} and thus removes the usage of the LSI.

Finally we summarize the recent results related to the  bulk universality 
of local eigenvalue statistics. 
The local semicircle law for   Step 1  was first established for Wigner
 matrices in a series of papers  \cite{ESY1, ESY2, ESY3}. 
The  method  was based on a self-consistent equation for 
the Stieltjes  transform of the eigenvalues and the continuity of the 
imaginary part of the spectral parameter  in 
the Stieltjes  transform.  As a by-product, an
 eigenvector delocalization estimate was  proved.

The universality for  Gaussian divisible ensembles  was  proved by Johansson \cite{J} 
for {\it hermitian} Wigner ensembles. It was extended to {\it complex} sample 
covariance matrices by 
Ben Arous and  P\'ech\'e \cite{BP}. There were two major restrictions of this
 method: 1. The Gaussian component 
was fairly large, it was required  to be of order one independent of $N$. 
2. It relies on explicit formulas for the correlation functions of 
eigenvalues  which are valid only for Gaussian divisible ensembles with   unitary  invariant Gaussian component.  The size of the 
Gaussian component was reduced to $N^{-1+ \e}$ in \cite{EPRSY} by using 
an improved formula for correlation functions and 
the local semicircle law from  \cite{ESY1, ESY2, ESY3}. The Gaussian component
was then removed by a perturbation argument using the reverse heat flow.
Thus the three step strategy to prove the universality was 
 introduced and  it led to the 
first proof of the bulk universality for hermitian Wigner ensembles.
 Due to the reverse heat flow 
argument used in  Step 3, the universality class established in 
\cite{EPRSY} was 
restricted to matrices with 
 smooth distributions for the matrix elements. 
Shortly after, Tao and Vu  \cite{TV} proved the four moment theorem which 
in particular 
removes the  smoothness restriction in Step 3. 
It thus proved the universality for hermitian Wigner matrices whose matrix
 element distributions were supported on at least  three points. The last 
condition was removed in \cite{ERSTVY} by combining the arguments 
 of \cite{EPRSY, TV}. The 
result of \cite{TV}  also implies that  the local statistics of symmetric 
Wigner matrices
and  GOE are the same,   but under the restriction that 
the first four moments of the matrix elements match those of GOE.  
 Thus the universality class for the local correlation functions
established via the approach of combining  \cite{TV} and \cite{J}   was broader
for  the hermitian  ensembles  than for the symmetric ones.   This improvement
was due to  Johansson's result \cite{J},  which  provided
 the universality for 
 Gaussian divisible ensembles in  Step 2,  
was available only for hermitian ensembles.

A more general and conceptually  very appealing approach for Step 2  
is via the local ergodicity of
 Dyson Brownian motion. This approach,  initiated  in \cite{ESY4},
 was applied to prove  the
universality  for symmetric Wigner matrices with the three
 point  support condition.  In \cite{ESYY}, we formulated a general 
theorem for the bulk universality which applies to  all classical 
ensembles, i.e.,  real and complex Wigner matrices,   real and complex 
 sample covariance  matrices
 and quaternion Wigner matrices. Later on, Tao and Vu \cite{TV3} also 
extended their results to the sample covariance matrices 
with the three point support condition for complex covariance matrices
 and four moment matching conditions 
for real ones.  Shortly after \cite{TV3}, P\'ech\'e \cite{P} 
also  extended the approach \cite{EPRSY} 
to the complex sample covariance matrices and proved the 
universality in the bulk.

Most recently,  we introduced  \cite{EYY}
 the Green function comparison theorem and extended the local semicircle law 
  to include the matrix elements of the Green functions. This allows us to 
 remove the smoothness restriction from  the reverse heat flow
 argument in  Step 3 of our approach. 
 We remark that the comparison theorems in \cite{TV} 
concern  individual  eigenvalues with a fixed index,
while the Green function comparison theorem is at a  fixed energy. 
On the other hand,   in \cite{EYY} the variances of the matrix elements 
were allowed to vary, i.e., 
 the matrices belonged to  generalized Wigner ensembles. 
 The three step strategy can thus be applied and the universality was proved for
  generalized Wigner ensembles  with essentially
 only one class of measures,
 the Bernoulli measures,   excluded due to  the LSI used in verifying 
 Assumption III  in Step 2.  Finally, in the current paper, 
 Assumption III will be  shown to be  a consequence of a strong local 
semicircle law,  which will be proved for all ensembles  with a subexponential 
 decay property. In particular,   Bernoulli measures are now included in the
 universality class (in the sense of \eqref{matrixthm})
  for both hermitian and symmetric 
 generalized Wigner ensembles.   We have thus removed all restrictions except
 the subexponential decay  in our approach. A clear picture 
 of the three step strategy emerges:  Step 2 and 3 hold under very general 
conditions and are model independent.
 The main task of  proving the universality is to establish  a strong version
 of the local semicircle law---which can be model dependent.   
 We believe that our method  applies to generalized sample covariance
 matrices as well, but we will not pursue this direction in this paper.

\section{Proof of Universality}

We now prove the main universality theorem, Theorem \ref{mainsk}.

\noindent 
{Step 1.  \it Universality for Dyson Brownian Motion:} Under the Assumptions II--IV in the introduction, 
the universality for the Dyson Brownian Motion was proved in \cite{ESYY}.  We recall the statement in the 
following Theorem.

\begin{theorem}\label{thm:main}[Theorem 2.1 of \cite{ESYY}]  
Let  $\e>0$ be the exponent from Assumption III. 
Suppose that  the 
Assumptions II, III and IV  hold for the solution $f_t$ of the forward equation
\eqref{dy} for all time $ t \ge N^{-2 \e}$.  Let $E\in \bR$ be a point
where $\varrho(E)>0$. Then 
for any $k\ge 1$ and for any compactly supported continuous test function
$O:\bR^k\to \bR$, we have
\be
\begin{split}
\lim_{b\to 0}\lim_{N\to \infty} \sup_{t\ge N^{-2\e+\delta}}  \;
\frac{1}{2 b }\int_{E-b}^{E+b}\rd E'
\int_{\R^k} &  \rd\alpha_1
\ldots \rd\alpha_k \; O(\alpha_1,\ldots,\alpha_k) \\
&\times \frac{1}{\varrho(E)^k} \Big ( p_{t,N}^{(k)}  - p_{\mu, N} ^{(k)} \Big )
\Big (E'+\frac{\alpha_1}{N\varrho(E)},
\ldots, E'+\frac{\alpha_k}{ N\varrho(E)}\Big) =0 \, .
\label{abstrthm}
\end{split}
\ee
\end{theorem}

Notice that the assumption on the initial entropy is not needed as was remarked in \cite{EYY}.

\medskip 
\noindent 
{Step 2 \it Universality for Gaussian divisible ensembles:}   The Dyson Brownian motion is generated by the matrix flow 
\eqref{matrixdbm}. Our task is to determine the initial ensemble $H_0$ so that the Assumptions II--IV of  Theorem \ref{thm:main} 
can be proved  for the flow. The Assumption IV is a direct consequence of the local semicircle law, i.e.,   Theorem \ref{thm:detailed}.
The Assumption III will be proved in  Proposition \ref{prop:lambdagamma}.  For the generalized Wigner matrices, 
the only assumption of  Theorem \ref{thm:detailed}
and Proposition \ref{prop:lambdagamma} is the subexponential decay property of  the distributions of the matrix elements. 
Since the evolution of the matrix element is given by an 
Ornstein-Uhlenbeck process, the subexponential property is preserved and we only have to check it for the initial data. 
We have thus proved the following theorem.

\begin{theorem}\label{gengenJT}
Suppose that the probability law for the initial matrix $H_0$ satisfies the
  assumptions of Theorem \ref{mainsk}. 
 Then there exists $\e_0>0$ such that for any  
$t\geq N^{-\e_0}$, the probability law for the eigenvalues of $H_t$ satisfies the universality equation  \eqref{matrixthm}.
\end{theorem}

\medskip 
\noindent 
{Step 3 \it Green function comparison theorem:} We have proved the 
universality for all ensembles with the matrix element at $(i,j)$  distributed 
by $\sigma_{ij} \xi_t^{ij}$ with 
\be\label{matrixdbmfake}
\xi_t^{ij} = e^{-t/2} \xi_0^{ij} + (1-e^{-t})^{1/2}\, \xi^{ij}_G,
\ee
where $\xi_G^{ij}$ are independent  Gaussian random variables  with mean 
$0$ and variance $1$ and $t\sim N^{-\e}$. 
In order to prove  Theorem \ref{mainsk}, 
it remains to approximate all random variables with 
the subexponential property by $\xi_t$. 
The only requirement of $\xi_0$ is the subexponential decay property
 and the mean zero and variance 
one normalization. Our tool is the following Green function comparison
 theorem from \cite{EYY}. 
It  implies   that the correlation functions 
of the eigenvalues of two matrix ensembles at a fixed energy are identical
 up to  the scale $1/N$  provided that the first four moments 
of the matrix elements  of these two ensembles are almost identical. 
 Prior to this theorem,  it was  \cite{TV} 
proved that  the joint distribution of individual eigenvalues for Wigner
 ensembles is the same under the four moment assumption. 
Tao-Vu's theorem addresses the  distribution of individual 
eigenvalues\footnote{In a recent preprint \cite{TV5} (appeared after the current
preprint was first posted), it was pointed out that if the four moment condition is  violated, then the
differences between individual
eigenvalues of the two ensembles are bigger than the eigenvalue spacing.
Thus the four moment condition is also necessary for locating the individual eigenvalues.
This is in contrast with the main theme of this paper that   gap distribution and
correlation functions are even independent of  the second moments as long as they are nonzero.}
 while Theorem \ref{comparison} 
compares  Green functions (and thus eigenvalues) at a fixed energy.

\begin{theorem}\label{comparison}   Suppose that we have  
two generalized $N\times N$ Wigner matrices, $H^{(v)}$ and $H^{(w)}$, with matrix elements $h_{ij}$
given by the random variables $N^{-1/2} v_{ij}$ and 
$N^{-1/2} w_{ij}$, respectively, with $v_{ij}$ and $w_{ij}$ satisfying
the uniform subexponential decay condition \eqref{subexp}.
  Fix a bijective ordering map on the index set of
the independent matrix elements,
\[
\phi: \{(i, j): 1\le i\le  j \le N \} \to \Big\{1, \ldots, \gamma(N)\Big\} , 
\qquad \gamma(N): =\frac{N(N+1)}{2},
\] 
and denote by  $H_\gamma$  the generalized Wigner matrix whose matrix elements $h_{ij}$ follow
the $v$-distribution if $\phi(i,j)\le \gamma$ and they follow the $w$-distribution
otherwise; in particular $H^{(v)}= H_0$ and $H^{(w)}= H_{\gamma(N)}$. Let $\kappa>0$ be arbitrary
 and suppose that for any  small parameter $\tau>0$ 
and for any  $y \ge N^{-1 + \tau}$  we have 
the following estimate on the diagonal elements of the resolvent:
\be\label{basic}
\P\left(\max_{0 \le \gamma \le \gamma(N)} \max_{1 \le k \le N}  
 \max_{|E|\le 2-\kappa}\left |  \left (\frac 1 {  H_{\gamma}-E- i y} \right )_{k k } 
\right |\le N^{2\tau} \right)\geq 1-CN^{-c\log\log N}
\ee
with some constants $C, c$ depending only on $\tau, \kappa$.
 Moreover, we assume that the first three moments of
 $v_{ij}$ and $w_{ij}$ are the same, i.e.
$$
    \E \bar v_{ij}^s v_{ij}^{u} =  \E \bar w_{ij}^s w_{ij}^{u},
  \qquad 0\le s+u\le 3,
$$
 and the difference between the  fourth moments of 
 $v_{ij}$ and $w_{ij}$ is much less than 1, say
\be\label{4match}
\left|\E \bar v_{ij}^s v_{ij}^{4-s}- \E \bar w_{ij}^s w_{ij}^{4-s}
\right|\leq N^{-\delta}, \qquad s=0,1,2,3,4,
\ee
for some given $\delta>0$. 
 Let $\e>0$ be arbitrary and choose an 
$\eta$ with $N^{-1-\e}\le \eta\le N^{-1}$.
For any  sequence of positive integers $k_1, \ldots, k_n$, set  complex parameters
$z^m_j = E^m_j \pm i \eta$,   $j = 1, \ldots k_i$, $m = 1, \ldots, n$ with $ |E^m_j| \le2-2\kappa $
and with an arbitrary choice of the $\pm$ signs. 
Let  $G^{(v)}(z) =  ( H^{(v)}-z)^{-1}$ be the resolvent
and let    $F(x_1, \ldots, x_n)$ be a function such that for
any multi-index $\al=(\al_1, \ldots ,\al_n)$ with $1\le |\al|\le 5$
and for any $\e'>0$  sufficiently small, we have
\be\label{lowder}
\max \left\{|\partial^{\al}F(x_1, \ldots, x_n)|: 
\max_j|x_j|\leq N^{\e'}\right\}\leq N^{C_0\e'}
\ee
and
\be\label{highder}
\max\left\{|\partial^\al F(x_1, \ldots, x_n)|:
 \max_j|x_j|\leq N^2\right\}\leq N^{C_0}
\ee
for some constant $C_0$.

Then, there is a constant $C_1$,
depending on $\alpha, \beta$, $\sum_i k_i$ and $C_0$ such that for any $\eta$ with
$N^{-1-\e}\le \eta\le N^{-1}$
and  for any choices of the signs 
in  the imaginary part of $z^m_j$ 
\begin{align}\label{maincomp}
\Bigg|\E F  \left (  \frac{1}{N^{k_1}}\tr  
\left[\prod_{j=1}^{k_1} G^{(v)}(z^1_{j})\right ]  , \ldots, 
 \frac{1}{N^{k_n}} \tr \left [  \prod_{j=1}^{k_n} G^{(v)}(z^n_{j}) \right ]  \right )  &  -
  \E F\left ( G^{(v)} \to  G^{(w)}\right )  \Bigg| \non\\
\le & C_1 N^{-1/2 + C_1 \e}+C_1 N^{-\delta+ C_1 \e},
\end{align}
where in the second term the arguments of $F$ are changed from the 
Green functions of $H^{(v)}$
to $H^{(w)}$ and all other parameters remain unchanged.
\end{theorem}

Given this theorem, for any matrix ensemble $H$ 
whose  matrix element at $(i,j)$ are   distributed 
according to $\sigma_{ij} \zeta^{ij}$, 
we need to find $\xi_0^{ij}$  such that the first four moments of 
 $\zeta^{ij}$ and $\xi^{ij}_t$ 
are almost the same and $\xi_0^{ij}$  has a subexponential decay.
Since the real and imaginary parts are i.i.d., it is sufficient
to match them individually.
  This is the content of the following lemma which is
stated for  real random variables normalized to variance
one. With this lemma,
 we have proved  Theorem \ref{mainsk}. This lemma is essentially the same 
 as  Lemma 28 in \cite{TV}.

\begin{lemma}\label{fmam} Let  $m_3$ and $m_4$ be two real numbers  such that 
$$
m_4-m_3^2-1\ge 0,\,\,\, m_4\leq C_2
$$
for some positive constant  $C_2$. Let $\xi^G$ be a real Gaussian random variable 
 with mean $0$ and variance $1$. 
Then for any sufficient small $\gamma>0$ (depending on  $C_2$),
 there exists a real  random 
variable $\xi_\gamma$ with subexponential decay 
and independent of $\xi^G$, such that the first four moments of 
$$
\xi'=(1-\gamma)^{1/2}\xi_\gamma+\gamma^{1/2}\xi^G
$$ 
are $m_1(\xi')=0$, $m_2(\xi')=1$, $m_3(\xi')=m_3$ and $m_4(\xi')$, and 
\be\label{m4m4}
|m_4(\xi')-m_4|\leq C\gamma
\ee 
for some positive constant $C$ depending on $C_2$.
\end{lemma}

\noindent 
{\it Proof.} It is easy to see by an
explicit construction that the following holds: 
\begin{align}\label{a1}
& \text{For  any given numbers $m_3, m_4$, with $m_4-m_3^2-1\ge 0$
there is a random}  \nonumber  \\ & 
  \text{ variable $X$ with  first four moments
$0, 1, m_3, m_4$ and with subexponential decay.}
\end{align}
For any real random variable
$\zeta$, independent of $\xi^G$, and with the first 4 moments being
 $0$, $1$, $m_3(\zeta)$
 and $m_4(\zeta)<\infty$, the first 4 moments of 
$$
\zeta'=(1-\gamma)^{1/2}\zeta+\gamma^{1/2}\xi^G
$$
are $0$, $1$, 
\be\label{relm3}
m_3(\zeta')=(1-\gamma)^{3/2}m_3(\zeta)
\ee
 and 
 \be\label{relm4}
m_4(\zeta')=(1-\gamma)^{2}m_4(\zeta) +6\gamma-3\gamma^2.
 \ee

Using \eqref{a1}, we obtain that for any $\gamma>0$
there exists a real random variable $\xi_\gamma$ such that the 
first four moments are $0$, $1$, 
$$
m_3(\xi_\gamma)=(1-\gamma)^{-3/2} m_3
$$
 and
 $$
 m_4(\xi_\gamma)=m_3(\xi_\gamma)^2+(m_4-m^2_3).
 $$ 
With $m_4\leq C_2$, we have $m_3^2\leq C_2^{3/2}$, thus   
$$
 |m_4(\xi_\gamma)-m_4 |\leq C\gamma
$$
for some positive constant  $C$ depending on $C_2$. 
Hence with \eqref{relm3} and \eqref{relm4}, we obtain that $\xi'
=(1-\gamma)^{1/2}\xi_\gamma+\gamma^{1/2}\xi^G$
 satisfies $m_3(\xi')=m_3$ and \eqref{m4m4}.
This completes the proof of Lemma \ref{fmam}.
\qed

\section{Large Deviation of  Local Semicircle Law}\label{ld-semi}

We first reprove the large deviation of local semicircle law
 given in \cite{EYY}.  The result of  this section is relevant only for
$\eta \ge M^{-1}$.

\begin{theorem}\label{thm:detailed}
  Assume the $N\times N$ random matrix $H$ satisfies \eqref{sum}, \eqref{de-de+}, 
\eqref{speccond} and \eqref{subexp}, $\E\, h_{ij}=0$, for any $1\leq i,j\leq N$.
 Let $z=E+i\eta$ $(\eta>0)$ and let $\theta(z)$ be a non-negative function  defined  by
\be\label{q}
\theta =\theta(z):=\frac{1}{|1-m_{sc}(z)^2|}+
\frac{1}{\max\big\{\delta_+\,\,, \,\left|\re  m_{sc}^2(z)-1\right|\big\}}.
 \ee
Let  $\kappa\equiv ||E|-2|$. Then for all 
 $z=E+i\eta$ with 
\be\label{relkaeta}
|E|\leq 5,  \qquad  \frac{1}{N}< \eta \le 10,
\qquad  \sqrt{M\eta}\ge (\log N)^{12+3\al}\theta^2(z)(\kappa+\eta)^{1/4}
\ee 
we have 
\be\label{mainlsresult}
\P\left\{\max_{i}|G_{ii}(z)-m_{sc}(z)|\geq (\log N)^{6+2\al} 
\frac{(\kappa+\eta)^{1/4}}{\sqrt{M\eta}}\,\,\theta(z)\right\}\leq CN^{-c(\log \log  N)}
\ee
and
\be\label{mainlsresult2}
\P\left\{\max_{i\ne j}|G_{ij}(z)|\geq (\log N)^{6+2\al} 
\frac{(\kappa+\eta)^{1/4}}{\sqrt{M\eta}}\right\}\leq CN^{-c(\log \log  N)}
\ee
for sufficiently large N with positive some  constants $c$ and $C>0$ that depend
 only $\al$ and $\beta$ in 
\eqref{subexp} and $\delta_-$ in \eqref{de-de+} and \eqref{speccond}.
\end{theorem}

The theorem will be proved at the end of the section after collecting several lemmas.
The first lemma describes the
behavior of $m_{sc}$ in the various regimes, its proof is
elementary calculus. We use the notation $f\sim g$  for two positive functions
in some domain $D$  if
there is a positive universal constant $C$ 
such that $C^{-1}\le f(z)/g(z) \le C$ holds for all $z\in D$.
\begin{lemma}  We have for all $z$ with $\im z>0$ that
\be
   |m_{sc}(z)| = |m_{sc}(z)+z|^{-1}\le 1.
\label{zmsc2}
\ee

{F}rom  now on, let $z=E+i\eta$ with $|E|\le 5$ and $\eta>0$.
If $\eta\ge 10$, then we have
\be\label{largez}
   \im m_{sc}(z) \sim \eta^{-1}, \qquad 
  |m_{sc}(z)|\sim \eta^{-1},\qquad |1-m_{sc}^2(z)|\sim 1,\qquad |1-\re m_{sc}^2(z)|\sim 1.
\ee
If $\eta\le 10$, then we  have
\be
   |m_{sc}(z)|\sim 1, \qquad   |1-m_{sc}^2(z)|\sim \sqrt{\kappa+\eta}.
\label{smallz}
\ee
For the behavior of $|1-\re m_{sc}^2(z)|$ and $\im m_{sc} (z)$
we distinguish two cases.

{\it Case 1.} For $|E|\ge 2$ we have

\be\label{esmallfake}
\im m_{sc} (z)\sim\left\{\begin{array}{cc} 
 \frac{\eta}{\sqrt{\kappa+\eta}} & \mbox{if  $\kappa\ge\eta$} \\  & \\
\sqrt{\kappa+\eta} & \mbox{if $\kappa\le \eta$}
\end{array}
\right.
\ee
$$
   |1-\re m_{sc}^2(z)|\sim \sqrt{\kappa+\eta}.
$$

{\it Case 2.} For $|E|\le 2$ we have
$$
    \im m_{sc}(z)\sim \sqrt{\kappa+\eta},
$$
\be\label{esmall}
|1-\re m_{sc}^2 (z)|\sim\left\{\begin{array}{cc} 
 \kappa+\frac{\eta}{\sqrt{\kappa+\eta}} & \mbox{if  $\eta\le \kappa$} \\  & \\
\sqrt{\kappa+\eta} & \mbox{if $\kappa\le \eta$}
\end{array}
\right.
\ee
\qed
\end{lemma}

Thus the  control function $\theta(z)$ has the following behavior 
\be\label{defgz2}
\theta(z)\sim \left\{\begin{array}{cc} 
1 & \mbox{if} \;\; \eta\ge 10, \\  \\
\min\left\{\delta_+^{-1},\,\,\,\,\sqrt{\kappa}/\eta,\,\,\, 
\kappa^{-1}\right\} & \qquad \mbox{if} \;\; \eta\le 10, \quad
 |E|\leq 2 {\rm\,\,\, and \,\,\,}\kappa\geq \eta, 
\\
\\(\kappa+\eta)^{-1/2} &\mbox{if} \;\; \eta\le 10, \quad
\mbox{and}\;\; \big\{ 2\le |E|\le 10 {\rm\,\,\, or \,\,\,}\kappa\leq \eta\big\}.
\end{array}
\right.
\ee

Note that the precise formula \eqref{q} for $\theta(z)$ is not important, only
 its asymptotic behavior for small $\kappa$, $\eta$ and $\delta_+$ is relevant.
 The theorem remains valid if $\theta(z)$ is replaced by $\widetilde \theta(z)$
 with $\widetilde\theta(z)\leq C \theta(z)$. In particular, $\theta(z)$ 
can be chosen to be order one when $E$ is not near the edges of the 
spectrum. If we are only concerned with the 
generalized Wigner ensemble $\eqref{VV}$, then by \eqref{de-de+2}
we can choose
 $\theta(z)=(\kappa+\eta)^{-1/2}$ for any $z=E+i\eta$ ($\eta>0$).
For universal Wigner matrices we have $\theta(z)\le C(\kappa+\eta)^{-1}$
for $|z|\le 10$,
i.e., using the parameter $A$ introduced in  Theorem \ref{lsc}, we have
\be
  \theta(z) \le \frac{C}{(\kappa +\eta)^{A/2}}, \qquad A=1,2, \quad |E|, \eta\le 10.
\label{thetaA}
\ee
Based upon these formulas, we also have, for any $z=E+i\eta$ with  $\eta>0$,
\be
 \im m_{sc}(z) + \frac{1}{\theta(z)} \le C\min\{ 1, \sqrt{\kappa+\eta}\}.
\label{upper}
\ee
\bigskip

First, we introduce some notations. Recall that $G_{ij}=G_{ij}(z)$ 
denotes the  matrix element
$$
G_{ij}=\left(\frac1{H-z}\right)_{ij}
$$
and  
  $$m(z)=m_N(z)=\frac1N\sum_{i=1}^NG_{ii}(z).
$$

\begin{definition}\label{basicd}
 Let ${\T}=\{k_1$, $k_2$, $\ldots$, $k_t\}\subset \{1,2, \ldots ,N\}$
 be  an unordered set of $|\bT|=t$ elements and let 
 $H^{(\bT)}$ be the $N-t$ by $N-t$ minor of $H$ after removing the
 $k_i$-th $(1\leq i\leq t)$ rows and columns. For $\bT=\emptyset$, we have $H^{(\emptyset)}=H$.
Similarly, we define $\ba^{(\ell;\,\,{\T})}$ the $\ell$-th column with 
$k_i$-th $(1\leq i\leq t)$
 elements removed. Sometimes, we just use the short notation 
$\ba^{\ell}$=$\ba^{(\ell;\,\,{\T})}$.
 For any ${\T}\subset \{ 1, 2, \ldots , N\}$ we introduce the following notations:
 \begin{align}
 G^{({\T})}_{ij}:=&(H^{({\T})}-z)^{-1}(i,j) \non\\
 Z^{({\T})}_{ij}:=&\ba^{i}\cdot(H^{({\T})}-z)^{-1}\ba^{j}=\sum_{k,l\notin {\T}}
\overline{\ba^{\,i}_k} G^{({\T})}_{k\,l}\ba^{j}_{l\,} \non\\
\wH^{({\T})}_{ij}:= & h_{ij}-z\delta_{ij}-Z^{({\T})}_{ij}. \non
 \end{align}
These quantities depend on $z$, but we mostly neglect this dependence in the notation. 
\end{definition}
\bigskip

The following two results were proved in our previous work
(Lemma 4.2 and Corollary B.3 of \cite{EYY}) and they will 
be our key inputs. We start with the  self-consistent perturbation formulas.

\begin{lemma}\label{basicIG} [Self-consistent Perturbation Formulas] Let 
$\bT\subset \{ 1, 2, \ldots, N\}$. For  simplicity, we use the
 notation $(i \,{\T})$ for $(\{i\}\cup {\T})$ and $(i j \,{\T})$
 for $(\{i,j\}\cup {\T})$. 
 Then we have the following identities:
\par\begin{enumerate}    
\item For any  $i\notin {\T}$ 
\be\label{GiiHii} 
 G^{({\T})}_{ii}=(\wH^{(i\,{\T})}_{ii})^{-1}.
 \ee
 \item For $i\neq j$ and $i,j\notin {\T}$
\be\label{GijHij} 
 G^{({\T})}_{ij}=-G^{({\T})}_{jj}G_{ii}^{(j\,{\T})}\wH^{(ij\,\,{\T})}_{ij}=
-G^{({\T})}_{ii}G_{jj}^{(i\,{\T})}\wH^{(ij\,\,{\T})}_{ij}.
 \ee
 	\item  For $i\neq j$ and $i,j\notin {\T}$
  \be\label{GiiGjii}
  G^{({\T})}_{ii}-G^{(j\,\,{\T})}_{ii}=
G^{({\T})}_{ij}G^{({\T})}_{ji}(G^{({\T})}_{jj})^{-1}.
  \ee
  \item  For any indices  $i$, $j$ and $k$ that are different  and 
  $i,j,k \notin {\T}$
 \be\label{GijGkij}
G^{({\T})}_{ij}-G^{(k\,\,{\T})}_{ij}=G^{({\T})}_{ik}G^{({\T})}_{kj}
(G^{({\T})}_{kk})^{-1} . \ee 
 \end{enumerate}
 
\end{lemma}
\begin{lemma}\label{generalHWT}
Let $a_i$ ($1\leq i\leq N$) be $N$ independent random complex  variables with mean zero, 
variance $\sigma^2$  and having the uniform  subexponential decay \eqref{subexp}. 
Let $A_i$, $B_{ij}\in \C$ ($1\leq i,j\leq N$). 
 Then we have that 
 \begin{align}
\P\left\{\left|\sum_{i=1}^N a_iA_i\right|\geq 
(\log N)^{\frac32+\al} \sigma \,\Big(\sum_{i}|A_i|^2\Big)^{1/2}\right\}\leq & CN^{-\log\log N},
\label{resgenHWTD} \\
\P\left\{\left|\sum_{i=1}^N\overline a_iB_{ii}a_i-\sum_{i=1}^N\sigma^2 B_{ii}\right|\geq 
(\log N)^{\frac32+2\al} \sigma^2 \Big( \sum_{i=1}^N|B_{ii}|^2\Big)^{1/2}\right\}\leq &
 CN^{-\log\log N},\label{diaglde}\\
\P\left\{\left|\sum_{i\neq j}\overline a_iB_{ij}a_j\right|\geq (\log N)^{3+2\al} \sigma^2 
\Big(\sum_{i\ne j} |B_{ij}|^2 \Big)^{1/2}\right\}\leq & CN^{-\log\log N}, \label{resgenHWTO}
\end{align}
for some constants $C$ depending on $\al$ and $\beta$ in \eqref{subexp}.  
\end{lemma}

We start with determining  a system of self-consistent  equations for the
diagonal matrix elements of the resolvent.
We can write  $G_{ii}$  as follows, 
$$
   G_{ii} = (\wH^{(i)}_{ii})^{-1}=\frac{1}{\E_{\ba^i}\wH^{(i)}_{ii}+\wH^{(i)}_{ii}-
\E_{\ba^i}\wH^{(i)}_{ii}},
$$
where $\E_{\ba^i} = \E_i$ denotes the expectation with respect to 
the elements in the $i$-th column of the matrix $H$, i.e., w.r.t.
$\ba^i = (h_{1i}, h_{2i}, \ldots, h_{Ni})^t$. 
Introduce the notations 
\be
   A_i: =\sigma^2_{ii}G_{ii}+\sum_{j\neq i}\sigma^2_{ij}\frac{G_{ij}G_{ji}}{G_{ii}}
\label{defA}
\ee
and
$$
  Z_i:=  \sum_{k,l\ne i} 
\Big[ \ov{\ba^i_k} G^{(i)}_{k\,l}\ba^i_l-\E_{\ba^i} \ov{\ba^i_k} 
G^{(i)}_{k\,l}\ba^i_l\Big]=Z_{ii}^{(i)}-\E_i Z_{ii}^{(i)}.
$$

Using the fact that $G^{(i)}=(H^{(i)}-z)^{-1}$ is  independent of $\ba^i$ and $\E_{\ba^i}
 \overline {\ba^i_k}\ba^i_l=\delta_{k\,l}\sigma^2_{ik}$, we obtain
$$
\E_{\ba^i}\wH^{(i)}_{ii}=-z- \sum_{j\neq i}\sigma^2_{ij}G^{(i)}_{jj}
$$
and
$$
  \wH^{(i)}_{ii}-\E_{\ba^i}\wH^{(i)}_{ii}
  = h_{ii} - Z_i . \qquad
$$
Denote by 
\be\label{seeqerror}
\Upsilon_i=\Upsilon_i(z):=A_i +
\left(\wH^{(i)}_{ii}-\E_{\ba^i}\wH^{(i)}_{ii}\right) = A_i +h_{ii}-Z_i
\ee
and we  have the identity 
\be\label{mainseeq}
G_{ii} = \frac{1}{-z- \sum_{j}\sigma^2_{ij}G_{jj}+\Upsilon_i}.
\ee
Let
$$
  v_i := G_{ii}-m_{sc}, \qquad m:=\frac{1}{N}\sum_i G_{ii}, \qquad \bar v:=\frac{1}{N}
  \sum_i v_i = \frac{1}{N}\sum_i (G_{ii}-m_{sc}).
$$

\bigskip

We will estimate the following key quantities
\be\label{defLambda}
  \Lambda_d:=\max_k |v_k| = \max_k |G_{kk}-m_{sc}|, \qquad
 \Lambda_o:=\max_{k\ne \ell} |G_{k\ell}|,
\ee
where the subscripts refer to ``diagonal'' and ``offdiagonal'' matrix elements.
All the quantities defined so far depend on the spectral parameter $z=E+i\eta$, but
we will mostly omit this fact from the notation. The real part $E$ will always be kept
fixed. For the imaginary part we will use a continuity argument at the end of the proof
and then the dependence of $\Lambda_{d,o}$ on $\eta$ will be indicated.

\bigskip

Both quantities $\Lambda_d$ and $\Lambda_o$ will be typically small, 
eventually we will prove that their size is less than $(M\eta)^{-1/2}$,
modulo logarithmic  corrections and a factor involving the distance to the
edge. We thus define the exceptional event
\be
   \Omega_\Lambda =\Omega_\Lambda(z):= \Big\{ \Lambda_d(z) + \Lambda_o(z)
\ge \frac{(\log N)^{-3/2}}{\theta(z)} \Big\}.
\label{defOm}
\ee
We will always work in $\Omega_\Lambda^c$, and, in particular,  we will have
$$
   \Lambda_d(z) + \Lambda_o(z)
\le C(\log N)^{-3/2}
$$
since $1/\theta(z)\le C$ by \eqref{upper}.
Define the set
$$
 S:=\{ z=E+i\eta\; : \;
 |E|\leq 5,  \quad  N^{-1}< \eta \le 10\}.
$$
We thus have
\be
c\le |G_{ii}(z)|\le  C
\qquad \mbox{in $\Omega_\Lambda^c$}
\label{lowerbound}
\ee
for any $z\in S$
with some universal constant $c>0$.
 Here we estimated $\big| |G_{ii}|- |m_{sc}|\big|\le \Lambda_d$, and we  
used from \eqref{largez}--\eqref{smallz} that $m_{sc}(z)$ satisfies
$|m_{sc}(z)|\sim 1$  for $z\in S$.

Thus, a special case of \eqref{GijGkij} or \eqref{GiiGjii},
$$
   G_{k\,l}^{(i)}= G_{k\,l} - \frac{G_{ki}G_{il}}{G_{ii}}, \qquad  i\ne l,k,
$$
together with \eqref{lowerbound}
implies that for any $i$ and with a sufficiently large constant $C$
\be
    \max_{k\ne l} |G_{k\,l}^{(i)}|\le \Lambda_o + C\Lambda_o^2 \le C\Lambda_o
 \qquad \mbox{in $\Omega_\Lambda^c$},
\label{Gkl}
\ee
\be
   C^{-1}\le |G_{kk}^{(i)}|\le C,  \qquad \mbox{for all 
$k\ne i$ and in $\Omega_\Lambda^c$}
\label{Gkk}
\ee
\be
 |G_{kk}^{(i)}-m_{sc}|\le \Lambda_d +  C\Lambda_o^2
 \qquad \mbox{for all $k\ne i$ and in $\Omega_\Lambda^c$}
\label{Gkkm}
\ee
 and 
\be
   |A_i|\le \frac{C}{M} + C\Lambda^2_o \qquad \mbox{in $\Omega_\Lambda^c$}.
\label{Aest}
\ee
Here we have used that 
$$
\Big | \frac{G_{ki}G_{il}}{G_{ii}} \Big | \le c^{-1} \Lambda^2_o 
\qquad \mbox{in $\Omega_\Lambda^c$}
$$
with $c$ being the constant in  \eqref{lowerbound}
and we also used that $\sum_{j} \sigma_{ij}^2 =1$.
Similarly, with one more expansion step, we get
\be
  \max_{ij}\max_{k\ne l} |G_{k\,l}^{(ij)}|\le C\Lambda_o,  \qquad
   \max_{ij}\max_{k} |G_{kk}^{(ij)}|\le C \qquad \mbox{in $\Omega_\Lambda^c$}
\label{Gklij}
\ee
and
\be
 |G_{kk}^{(ij)}-m_{sc}|\le \Lambda_d +  C\Lambda_o^2
 \qquad \mbox{for all $k\ne i,j$ and in $\Omega_\Lambda^c$}.
\label{Gkkim}
\ee
\medskip

Using these estimates, 
the following lemma shows that $Z_i$ and  $Z_{ij}^{(ij)}$ are small
assuming  $\Lambda_d+\Lambda_o$ is small and the $h_{ij}$'s are not
too large. The control parameter for the $Z$'s is $\Phi=\Phi(z)$,
defined below \eqref{defphithe}.
These bounds hold uniformly in $S$.

\begin{lemma}\label{selfeq1} 
Denote by 
\be\label{defphithe}
 \Phi:= \Phi(z)= \frac{ \sqrt{\Lambda_d} +\Lambda_o+(\kappa+\eta)^{1/4}}{ \sqrt{M\eta} } ,
\ee
and define the exceptional events 
\begin{align} 
\Omega_1 &:= \left \{  \max_{1\leq i,j\leq N}|h_{ij}|\geq
 (\log N)^{2\al}|\sigma_{ij}| \right \} \non \\
\Omega_d(z) &:= \left \{   \max_{i}|Z_i(z)|\ge (\log N)^{5+2\al}\Phi (z)\right \} 
 \non \\
\Omega_o(z) &:= \left \{    \max_{i\ne j}|Z_{ij}^{(ij)}(z)| \ge (\log N)^{5+2\al}
 \Phi (z)\right \} \non  
\end{align}
and we let 
\be
  \Omega:= 
\Omega_1 \cup \bigcup_{z\in S}
\Big[\big( \Omega_d(z) \cup \Omega_o(z) \big)\cap \Omega_\Lambda^c(z) \Big]
\label{defOmega}
\ee 
to be the set of all exceptional events.
Then we have 
\be\label{B12}
\P (\Omega)    \le CN^{-c(\log \log  N)}.
\ee
\end{lemma}

\bigskip

{\it Proof}:
 Under the assumption of \eqref{subexp}, we have 
\be\label{resboundhij}
\P\left(\Omega_1 \right)\leq CN^{-c\log\log N},
\ee
therefore we can work on the complement set $\Omega_1^c$. Define
the event 
$$
 \wt\Omega_\Lambda(z):= \Big\{ \Lambda_d(z) + \Lambda_o(z)
\ge 2\frac{(\log N)^{-3/2}}{\theta(z)} \Big\}.
$$
Notice that the estimates \eqref{Gkl}--\eqref{Gkkim}
also hold on $\wt\Omega^c_\Lambda$, maybe with  different constants $C$.
We now prove that for any fixed $z\in S$, we have
\be\label{POdz}
\P\Big(\wt\Omega_\Lambda^c(z)\cap  \left \{   \max_{i}|Z_i(z)|\ge
  C(\log N)^{2+2\al}\Phi (z)\right \}  \Big)\leq CN^{-c\log\log N}
\ee
and
\be\label{Poz}
\P\Big(\wt\Omega_\Lambda^c(z)\cap 
 \left \{    \max_{i\ne j}|Z_{ij}^{(ij)}(z)| \ge (\log N)^{4+2\al}
 \Phi (z)\right \}\Big)\leq CN^{-c\log\log N}.
\ee

To see \eqref{POdz}, we
apply the estimate \eqref{diaglde} from the large 
deviation  Lemma \ref{generalHWT}, and we obtain that 
 \be\label{BI23}
|Z_{i}|
 \leq C(\log N)^{\frac{3}{2}+2\al}\sqrt{\sum_{k,l\not= i }\left|\sigma_{ik}
G^{(i)}_{k\,l}\sigma_{li}\right|^2}
\ee
holds with a probability larger than $1-CN^{-c(\log \log  N)}$ for sufficiently large $N$.

Denote by $u^{({i})}_\alpha$ and $\lambda_\al^{({i})}$ ($\al=1,2,\ldots ,N-1$)
 the eigenvectors and eigenvalues of $H^{({i})}$. Let $u^{({i})}_\alpha(l)$ denote the
 $l$-th coordinate of $u^{({i})}_\alpha$. 
 Then, using $\sigma_{il}^2 \le 1/ M$ and \eqref{Gkkm}, we have
\begin{align} \label{B14}
\sum_{k, l\not= i }\left|\sigma_{ik}
G^{(i)}_{k\,l}\sigma_{li}\right|^2 &\le \frac1M \sum_{k\not= i } \sigma_{ik}^2
 \left(|G^{(i\,)}|^2\right)_{kk}  \nonumber \\
&=
\frac1M \sum_{k\not= i } \sigma_{ik}^2 \sum_{\alpha}
\frac{|u^{(i\,)}_\alpha(k)|^2}{|\lambda_\alpha^{(i\,)}-z|^2}
\le 
\frac1M \sum_{k\not= i } \sigma_{ik}^2 \frac{\im G^{(i\,)}_{kk}(z)}{\eta}\non\\
& \le 
\frac{\Lambda_d + C\Lambda_o^2+  {\im m_{sc}(z) }}{M\eta} \non\\
&\le C \Phi^2 \qquad \mbox{in $\wt\Omega_\Lambda^c$}.
\end{align}
Here we defined $|A|^2:= A^* A$ for any matrix $A$ and we used  \eqref{upper}
to estimate $\im  m_{sc}(z) $.
Together with \eqref{BI23} we have proved \eqref{POdz} for a fixed $z$.

For the offdiagonal estimate \eqref{Poz}, for  $i\ne j$,
 we have from \eqref{resgenHWTO} that 
 \be\label{BI234}
|Z_{ij}^{(ij)}|
 \leq C(\log N)^{3+2\al}\sqrt{\sum_{k,l\not= i, j }\left|\sigma_{ik}
G^{(ij)}_{k\,l}\sigma_{lj}\right|^2}
\ee
holds with a probability larger than $1-CN^{-c(\log \log  N)}$ for sufficiently large $N$.
Similarly to \eqref{B14}, by using \eqref{Gkkim}, we get
$$
   \sum_{k,l\not= i, j }\left|\sigma_{ik}
G^{(ij)}_{k\,l}\sigma_{lj}\right|^2 \le C \Phi^2 \qquad \mbox{in $\wt\Omega_\Lambda^c$}.
$$
This proves \eqref{Poz}.

\bigskip

Now we start proving \eqref{B12}. First we
 choose an $N^{-10}$-net $\mathcal N$ in the set $S$,
i.e., a collection of points, 
$\{z_n\}_{n\in I}\subset S$, such that for any $z\in S$ there is $\widetilde z\in \mathcal N$
such that $|z-\widetilde z|\leq N^{-10}$. 
The net can be chosen such that $|I|\leq CN^{20}$.  
 Then \eqref{POdz} and \eqref{Poz} imply that
\be\label{znet}
\P \left( \exists \,\widetilde z\in \mathcal N \; , \; \mbox{s.t. 
$\wt\Omega_\Lambda^c(\widetilde z)$ holds and}
\max_{i}|Z_i(\widetilde z)|+\max_{i\ne j}|Z_{ij}^{(ij)}(\widetilde z)|\geq
 2(\log N)^{4+2\al} \Phi(\widetilde z) 
\right)\leq CN^{-c\log\log N} .
\ee
Now let $z\in S$ be arbitrary and choose
  $\widetilde z\in \mathcal N$ such that $|z-\widetilde z|\leq N^{-10}$. For any fixed $i\ne j$, 
we have
\be\label{Poooz}
\left||Z_{ij}^{(ij)}(z)|-|Z_{ij}^{(ij)}(\widetilde z)|\right|\leq
  |z-\widetilde z|\max_{\xi\in S}\left|\frac{\pt Z_{ij}^{(ij)}}{\partial z}(\xi)\right|.
\ee
By $\partial Z_{ij}^{(ij)}/\partial z= -\sum_{s,k,l\notin {(ij)}}
\overline{\ba^{\,i}_k} G^{(ij)}_{ks}G^{({ij})}_{sl}\ba^{j}_{l\,}$ and $\max_{ab}|G^{(ij)}_{ab}|
\leq \eta^{-1}$,  we have 
$$
\max_{\xi\in S}\left|\frac{\partial Z_{ij}^{(ij)}}{\partial z}(\xi)\right|\leq 
\frac{(\log N)^{6\al}}{M\eta^2}N^3\leq N^6, \qquad \mbox{in $\Omega_1^c$}.
$$
In the last inequality, we used 
 the assumption $\eta \geq N^{-1}$. Thus
$$
  \left||Z_{ij}^{(ij)}(z)|-|Z_{ij}^{(ij)}(\widetilde z)|\right|\leq N^{-4}  
\qquad \mbox{in $\Omega_1^c$}.
$$

Since $\Phi\ge M^{-1/2}\eta^{-1/4}\ge cN^{-1}$ for $z\in S$,
 we obtain 
$$
\left||Z_{ij}^{(ij)}(z)|-|Z_{ij}^{(ij)}(\widetilde z)|\right|\leq 
\Phi( z)  \quad \mbox{in $\Omega_1^c$},
$$ 
and exactly in the same way, we have
$$
\Big| \; |Z_i(z)|-|Z_i(\widetilde z)|\;\Big|\leq \Phi(\widetilde z)  \quad\mbox{in $\Omega_1^c$}.
$$
Moreover, by estimating $|\pt_z G|\le N^2$
in $S$, we see that $\Lambda_d(z)$, $\Lambda_o(z)$,
and $\Phi(z)$ are Lipschitz continuous functions in $S$ with a Lipschitz constant
bounded by $CN^3$. Therefore $\Phi (\widetilde z)$ can be replaced with $\Phi(z)$
in the lower bound on $|Z_{ij}^{(ij)}(\widetilde z)|$ and $|Z_i(\widetilde z)|$ 
obtained from \eqref{znet},
and, furthermore, $\Omega_\Lambda^c(z)\subset \wt\Omega_\Lambda^c(\widetilde z)$
using a trivial upper bound $\theta(z)\le N$.
Thus we get
$$
\P \left( \exists  z\in S  \; \mbox{s.t. 
$\Omega_\Lambda^c(z)$ and $\Omega_1^c$ hold and}
\max_{i}|Z_i( z)|+\max_{i\ne j}|Z_{ij}^{(ij)}( z)|\geq
 (\log N)^{5+2\al} \Phi(\widetilde z) 
\right)\leq CN^{-c\log\log N} .
$$
Combining  this with \eqref{resboundhij},
we obtain   \eqref{B12} and thus  Lemma \ref{selfeq1}.   
\qed

\bigskip 

Our goal is to  show that   $\Lambda_o(z) +\Lambda_d(z)$ is smaller
than $(M\eta)^{-1/2}$ (modulo edge and logarithmic corrections)
for any $z\in S$
in  the event $\Omega^c(z)$. We will use a continuity argument. In Lemma
\ref{thm: stepone} we show for any $z\in S$ that if $\Lambda_o(z) +\Lambda_d(z)$ is
smaller than $(\log N)^{-3/2}$, then it is actually also smaller than 
 $(M\eta)^{-1/2}$. In Lemma \ref{lm:initial} we show that this input condition
holds at least for $\im z = \eta=10$. Then reducing $\eta$, we show
by a continuity argument  that it holds for each $z\in S$.

\begin{lemma} [Bootstrap] \label{thm: stepone}
Let $z=E+i\eta$ and satisfy \eqref{relkaeta}, in particular $z\in S$.
Recall $\Lambda_d$, $\Lambda_o$ and $\Omega$ defined in  \eqref{defLambda} 
and \eqref{defOmega}. 
 Then we have that, in the event $\Omega^c$, if 
\be\label{alamb}
\Lambda_o(z) +\Lambda_d(z)\leq \frac{(\log N)^{-3/2}}{\theta(z)},
\ee
then we have
\be 
\label{alamb2}
\Lambda_o(z)+\Lambda_d(z)\leq (\log N)^{6+2\al} 
\frac{ (\kappa+\eta)^{1/4} }{\sqrt{M\eta}\,\,}\theta(z)
\ee
and we also have a stronger bound for the off-diagonal terms:
\be\label{alamb3}
 \Lambda_o(z) \leq (\log N)^{5+2\al} 
\frac{ (\kappa+\eta)^{1/4} }{\sqrt{M\eta}\,\,}.
\ee
\end{lemma}
\bigskip

\textit{Proof of Lemma \ref{thm: stepone}}. 
First note that condition \eqref{alamb} is equivalent
assuming the  event $\Omega_\Lambda^c(z)$ and we have
\be
    \Omega^c \cap \Omega_\Lambda^c(z) 
\subset  \Omega_d^c(z)\cup \Omega_o^c(z),
\label{OME}
\ee
so the event $\Omega_d^c(z)\cup \Omega_o^c(z)$ holds.
We recall from \eqref{upper}
that 
\be
   \frac{1}{\theta(z)}\le C\sqrt{\kappa+\eta} \le C, \qquad z\in S.
\label{thetaest}
\ee
With the assumption \eqref{alamb} we have (see \eqref{lowerbound}, \eqref{Gkk})
\be
   c\le |G_{ii}|\le C, \qquad c\le |G_{jj}^{(i)}|\le C
\label{Giibound}
\ee
and by \eqref{thetaest}
\be\label{lambdall1}
\Lambda_d(z)+\Lambda_o(z)\le \frac{C\sqrt{\kappa+\eta}}{(\log N)^{3/2}} \le 
 \frac{C}{(\log N)^{3/2}} 
\ee 
and thus, by \eqref{relkaeta} and \eqref{thetaest},
\be\label{Thetaimmsc}
  \frac{(\kappa+\eta)^{1/4}}{\sqrt{M\eta}\,\,}\le
  \Phi(z)\le C \frac{(\kappa+\eta)^{1/4}}{\sqrt{M\eta}\,\,}\le 
C(\log N)^{-12-3\al}.
\ee

We first estimate  the offdiagonal term $G_{ij}$. From \eqref{GijHij} we have
\be
  |G_{ij}|  = |G_{ii}| |G_{jj}^{(i)}| |K_{ij}^{(ij)}|
  \le C\left(|h_{ij}| + |Z_{ij}^{(ij)}|\right), \qquad i\ne j,
\label{estoff}
\ee
where we used \eqref{Giibound}.

By the remark after \eqref{OME}
we have 
$$
|G_{ij}|\le \frac{C(\log N)^{2\al}}{\sqrt{M}}+
C(\log N)^{5+2\al} \Phi\le C(\log N)^{5+2\al} \Phi,
$$
where we used
\eqref{Thetaimmsc} 
 to show that the
first term can be absorbed into the second.
 {F}rom the second       inequality in
\eqref{Thetaimmsc} we also have
\be
\Lambda_o=  \max_{i\ne j}|G_{ij}|\le 
 C(\log N)^{5+2\al}  \frac{(\kappa+\eta)^{1/4} }{\sqrt{M\eta}\,\,}.
\label{gij}
\ee
 This
proves the estimate \eqref{alamb3}. Using \eqref{thetaest}, we also see 
that \eqref{alamb2} holds for the summand $\Lambda_o$.

\bigskip

Now we estimate the diagonal terms. 
 Recalling $\Upsilon_i= A_i +h_{ii}-Z_i$ from \eqref{seeqerror}, 
with \eqref{Aest},
\eqref{Thetaimmsc}, \eqref{gij}
 we have,
\be
\Upsilon=\Upsilon(z):=\max_i   |\Upsilon_i(z)|\le  
 C\frac{(\log N)^{2\al}}{\sqrt{M}}+C(\log N)^{5+2\al}\Phi
\qquad \mbox{in  $\Omega^c\cap\Omega_\Lambda^c(z)$} . 
\label{defUpsilon}
\ee
Again, the first term can be absorbed into the second, 
 so we have proved
\be\label{41}
\Upsilon\le
(\log N)^{5+2\al}\Phi \le  C(\log N)^{-6}\qquad
 \mbox{in   $\Omega^c\cap\Omega_\Lambda^c(z)$}. 
\ee
 In the last step we used
\eqref{Thetaimmsc}.

{F}rom \eqref{mainseeq} we have the identity 
\be
  v_i = G_{ii} - m_{sc} 
 = \frac{1}{-z- m_{sc}- \Big(\sum_{j}\sigma^2_{ij}v_j-\Upsilon_i\Big)}
  - m_{sc} .
\label{idd}
\ee
Using $(m_{sc}+z) = -m_{sc}^{-1}$,  and the fact that $|m_{sc}+z|\ge 1$, so
with  $\Lambda_d +\Upsilon\le \frac{1}{10}|m_{sc}+z| $
(see in \eqref{lambdall1} and \eqref{41}),
 we can expand \eqref{idd} as
\be
   v_i = m_{sc}^2\cdot\Big(\sum_{j}\sigma^2_{ij}v_j-\Upsilon_i\Big)
  + O\Big(\sum_{j}\sigma^2_{ij}v_j-\Upsilon_i\Big)^2
  = m_{sc}^2\cdot\Big(\sum_{j}\sigma^2_{ij}v_j-\Upsilon_i\Big) + 
 O\Big((\Lambda_d+\Upsilon)^2 \Big).
\label{expand}
\ee
Summing up this formula for all $i$ and
recalling the definition
 $\bar v\equiv\frac1N\sum_{i}v_i=m-m_{sc}$ 
yield
$$
   \bar v = m_{sc}^2 \,\bar v - \frac{m_{sc}^2}{N}\sum_i \Upsilon_i 
  +  O\Big((\Lambda_d+\Upsilon)^2 \Big).
$$
Introducing the notations $\zeta: = m_{sc}^2(z)$, $\ov\Upsilon:= \frac{1}{N}\sum_i\Upsilon_i$
 for simplicity, we have
(using $\Lambda_d\le 1$)
\be\label{vi09}
\bar v=-\frac{\zeta}{1-\zeta}\ov\Upsilon
+O\Big(\frac{\zeta}{1-\zeta}(\Lambda_d+\Upsilon)^2 \Big)
=O\left( \left|\frac{\zeta}{1-\zeta}\right| \left(\Lambda^2_d+\Upsilon\right)\right).
\ee
Recall that $\Sigma$ denotes the matrix of covariances,
 $\Sigma_{ij}=\sigma^2_{ij}$, and we know that $1$ is 
a simple eigenvalue with the constant vector  ${\bf e}=N^{-1/2}(1,1,\ldots,1)$ as the eigenvector.
Let $Q:=I-|{\bf e}\rangle\langle {\bf e}|$
be the projection onto the orthogonal complement of ${\bf e}$, note that $\Sigma$ and $Q$ commute.
  Let  $\| \cdot \|_{\infty\to \infty}$ denote the $\ell^\infty\to \ell^\infty$ matrix norm.
With these notations, \eqref{expand}  can be written as
$$
   v_i-\bar v = \zeta \sum_{j}\Sigma_{ij} (v_j-\bar v) -
\zeta\Big( \Upsilon_i -\ov\Upsilon\Big)
+O\left( |\zeta| \left(\Lambda^2_d+\Upsilon\right)\right)
 + O\Big((\Lambda_d+\Upsilon)^2 \Big),
$$
and the error terms for each $i$ sums up to zero. Therefore, with $\Upsilon\le 1$, we have
\begin{align}
    v_i-\bar v = &  -\sum_j\Big(\frac{\zeta}{1-\zeta\Sigma}\Big)_{ij} 
\left(\Upsilon_j- \ov\Upsilon\right) + 
   O\Bigg( \Big\| \frac{\zeta Q}{1-\zeta \Sigma}\Big\|_{\infty\to \infty}
 (\Lambda_d^2+\Upsilon)\Bigg) \label{vi10} \\
= & \Big\| \frac{\zeta Q}{1-\zeta \Sigma}\Big\|_{\infty\to \infty}O\big( \Lambda^2_d+\Upsilon\big). \non
\end{align}

Combining \eqref{vi09} with \eqref{vi10}, we have 
\be\label{temp3.49}
\max_{i}|v_i|\leq C \left(\Big\| \frac{\zeta Q}{1-\zeta \Sigma}
\Big\|_{\infty\to \infty}+\left|\frac{\zeta}{1-\zeta}\right|\right)(\Lambda_d^2+\Upsilon).
\ee 
To estimate the norm of the resolvent, we recall  
the following elementary lemma (Lemma 5.3  in \cite{EYY}).

\begin{lemma} \label{msc}
Let $\delta_->0$ be a given constant. Then there exist small real numbers 
$\tau\geq 0$ and $c_1>0$, depending only on $\delta_-$, such that 
 for any positive number 
$\delta_+$, we have 
\be
    \max_{x\in [-1+\delta_-,1-\delta_+]}\left\{\Big| \tau + x\, m_{sc}^2(z)\Big|^2 
\right\}\le \left(1-c_1\, q(z) \right)(1+\tau)^2
\label{al}
\ee
with 
\be\label{defq}
q(z):=\max\{\delta_+,|1- \re \, m_{sc}^{2}(z)|\}. 
\ee
\qed
\end{lemma}

\begin{lemma} Suppose that $\Sigma$ satisfies \eqref{de-de+}, i.e.,
$\mbox{Spec}(Q\Sigma)\subset [-1+\delta_-, 1-\delta_+]$. Then we have
\be
\Big\| \frac{Q}{1- m_{sc}^2(z) \Sigma}\Big\|_{\infty\to\infty} 
 \le \frac{C(\delta_-) \log N}{ q(z)} 
\label{inftynorm}
\ee
with some constant $C(\delta_-)$ depending on $\delta_-$ and 
with $q$ defined in \eqref{defq} 
\end{lemma}

{\it Proof}: 
Let $\| \cdot \|$ denote the usual $\ell^2\to\ell^2$ matrix norm
and introduce $\zeta=m_{sc}^2(z)$. Rewrite 
$$
  \Big\| \frac{Q}{1-\zeta \Sigma}\Big\| 
  = \frac{1}{1+\tau} \Big\| \frac{Q}{1- \frac{\zeta \Sigma +\tau}{1+\tau}}\Big\|
$$
with $\tau$ given in \eqref{al}.  
By \eqref{al}, we have 
$$
   \Big\| \frac{\zeta \Sigma +\tau}{1+\tau}Q\Big\|  \le     \sup_{x\in[-1+\delta_-,
1-\delta_+]}
  \Big| \frac{\zeta x+\tau}{1+\tau} \Big|
\le  (1- c_1q(z))^{1/2} .
$$
To estimate the  $\ell^\infty\to\ell^\infty$ norm of this matrix,  recall that 
$|\zeta| =  |m_{sc}|^{2}\le 1$  and 
$\sum_j | \Sigma_{ij}| =\sum_j \Sigma_{ij}=\sum_{j}\sigma_{ij}^2=1$.
Thus we have 
$$
   \Big\|  \frac{\zeta \Sigma +\tau}{1+\tau} Q\Big\|_{\infty\to\infty}
    = \max_i \sum_j 
   \Big| \Big( \frac{\zeta \Sigma +\tau}{1+\tau} \Big)_{ij} \Big| 
   \le \frac{1}{1+\tau} \max_i \sum_j |\zeta \Sigma_{ij} + \tau \delta_{ij}|
  \le \frac{ |\zeta |+\tau}{1+\tau}\le 1.
$$

To see \eqref{inftynorm}, we can expand
\begin{align*}
\Big\| \frac{1}{1- \frac{\zeta \Sigma +\tau}{1+\tau}}Q\Big\|_{\infty\to
  \infty} 
  & \le \sum_{n< n_0} 
 \Big\| \frac{\zeta \Sigma +\tau}{1+\tau}Q\Big\|_{\infty\to\infty}^{n}
 + \sum_{n\ge n_0}  \Big\| \Big( \frac{\zeta \Sigma +\tau}{1+\tau}\Big)^nQ
 \Big\|_{\infty\to\infty} \\
 &  \le n_0 + \sqrt{N} \sum_{n\ge n_0}  \Big\| \Big( \frac{\zeta \Sigma +\tau}{1+\tau}\Big)^n
 Q\Big\|
  = n_0 + \sqrt{N} \sum_{n\ge n_0} (1-c_1 q(z))^{n/2}\\
&  = n_0 + C\sqrt{N} \frac{ (1-c_1q(z))^{n_0/2}}{q(z)} \le \frac{C\log N}{q(z)}.
\end{align*}
Choosing $n_0= C\log N/q(z)$ with a large $C$,
we have  proved the Lemma. 
\qed

\bigskip

We now return to the proof of Lemma \ref{thm: stepone},
 recall that we are in the set $\Omega^c\cap \Omega^c_\Lambda(z)$.
  First, inserting  \eqref{zmsc2}  and 
\eqref{inftynorm} into \eqref{temp3.49}, and using $1/q\le \theta$, we obtain
$$
  \Lambda_d= \max_i |v_i|
   \le   C\theta(z)  (\Lambda^2_d+\Upsilon)\log N .
$$
By the  assumption  \eqref{alamb}, we have $C\theta(z)\Lambda_d \log N\le 1/2$,
for large enough $N$,
therefore we get
$$
\Lambda_d   \le   C\theta(z)\Upsilon\log N.
$$
Using the bound on  $\Upsilon$ in \eqref{41}  and \eqref{Thetaimmsc}, we obtain 
$$
 \Lambda_d \leq C\theta(z)(\log N)^{6+2\al}
\frac{(\kappa+\eta)^{1/4}}{\sqrt{M\eta}},
$$
which, together with \eqref{gij},  completes the proof of \eqref{alamb2}.
\qed

\bigskip 

\begin{lemma}[Initial step]\label{lm:initial} 
Define
$$
\Omega_H  : =\left\{ \|H\|\ge 3 \right\},
$$
recall the  definitions of $\Omega_1$, $\Omega_d$ and $\Omega_o$
from \eqref{defOmega} and define
\be
   \wh\Omega : =\Omega_H\cup \Omega_1 \cup \bigcup
 \Big\{ \Omega_o(z) \cup \Omega_d(z)\; :\;
z=E+10i, |E|\le 5\Big\}.
\label{defwhOm}
\ee
Then we have
\be
  \P (\wh \Omega) \le CN^{-c\log\log N}.
\label{whomest}
\ee
Furthermore, in the set $\wh\Omega^c$ we have
\be\label{alambinit}
\Lambda_o(z) +\Lambda_d(z)\leq \frac{(\log N)^{-3/2}}{\theta(z)}
\ee
for $z= E+10i$, $|E|\le 5$.
\end{lemma}

\textit{Proof.} 
 The exceptional event $\Omega_H$
is controlled by  Lemma 7.2 of \cite{EYY}. For convenience,
we will recall this result in
Lemma~\ref{N-1/6}, Eq.  \eqref{resnlambda3}, and we note that the condition of this lemma,
$M\ge (\log N)^9$, is implied by \eqref{relkaeta}
and \eqref{upper}). Thus we have $\P (\Omega_H)\le  CN^{-c(\log \log  N)}$.

Denote by $u_\alpha$ and $\lambda_\al$ 
 the eigenvectors and eigenvalues of $H$. 
 On the set $\Omega_H^c$ all eigenvalues
are bounded, $|\lambda_\al|\le 3$.
In this set we have, with $|E|\le 5$, 
 \be
\im G_{kk} =   \eta   \sum_{\alpha}
\frac{|u_\alpha(k)|^2}{  (\lambda_\alpha - E )^2 + \eta^2} \ge  
 \frac c {  \eta }  \sum_{\alpha}
|u_\alpha(k)|^2  =  \frac c {  \eta }
\ee
with some positive constant $c>0$. We also have the upper
bound $|G_{kk}|\le\eta^{-1}$ and $\Lambda_o+\Lambda_d \le C/\eta$.
In particular, for $\eta=10$, we have
\be
   c\le |G_{kk}|\le C, \qquad  \mbox{in $\Omega_H^c$},
\label{lowerbound1}
\ee
with some positive constants.
Inspecting the proof of Lemma \ref{selfeq1}, 
notice that the restriction to the set $\Omega_\Lambda^c$
was used only to obtain the estimate
\eqref{lowerbound}. Once this estimate is obtained independently,
as in \eqref{lowerbound1} in the set  $\Omega_H^c$, 
all the estimates \eqref{Gkl}--\eqref{Gkkim} hold
and these are the necessary inputs for Lemma \ref{selfeq1}.
Thus, following  the proof of  \eqref{POdz}--\eqref{Poz},
and replacing $\Omega_\Lambda^c$ with $\Omega_H^c$,
 we obtain
that $\P\big\{\Omega_H^c\cap(\Omega_o(z)\cup \Omega_d(z))\big\}\le CN^{-c\log\log N}$
for each fixed $z=E+10i$, $|E|\le 5$. Finally, this estimate can be extended
 to hold simultaneously for all $z=E+10i$, $|E|\le5$
using an $N^{-10}$-net as for the proof of  \eqref{B12}.
This proves \eqref{whomest}.

Similarly, the argument \eqref{estoff}--\eqref{gij} shows that in the set $\wh \Omega^c$,
we have
\be
   \Lambda_o(z)\le \frac{C(\log N)^{5+2\al}}{\sqrt{M}}, \qquad z= E+10i,
\label{Lambdaoff}
\ee
and the argument
\eqref{defUpsilon}--\eqref{41} guarantees that 
\be
  \Upsilon(z)\le  \frac{C(\log N)^{5+2\al}}{\sqrt{M}}, \qquad z= E+10i,
\label{upsi}
\ee
in  $\wh \Omega^c$.
Finally, to control $\Lambda_d$, we use that from the
self consistent equation \eqref{idd} and the 
definition of $m_{sc}$, we have
\be\label{temp1.47}
v_n=\frac{\sum_{i}\sigma^2_{ni}v_i+O(\Upsilon)}{(z+m_{sc}+\sum_{i}\sigma^2_{ni}
v_i+O(\Upsilon))(z+m_{sc})}, \,\,\,1\leq n\leq N.
\ee
For $\eta=10$, with \eqref{temp2.8}, we have  $|z+m_{sc}(z)|>2$. Using
 $|G_{ii}|\leq \eta^{-1}=\frac{1}{10}$ and $|m_{sc}|\leq \eta^{-1}=\frac{1}{10}$,  we obtain 
 \be\label{vileq2}
 |v_i|\leq 2/\eta\leq 1/5, \qquad 1\leq i\leq N.
 \ee 
Using \eqref{upsi},
together with $|z+m_{sc}(z)|>2$
 and  \eqref{vileq2}, we obtain that   the absolute value of the r.h.s  
of \eqref{temp1.47} is less than 
\be
\frac{\sup_i|v_i|}{|z+m_{sc}(z)|-\sup_i|v_i|}+O(\Upsilon).
\ee
Taking the absolute value of \eqref{temp1.47} and maximizing over $n$, we have  
\be\label{temp1.49}
\Lambda_d = \sup_n |v_n|\leq \frac{\Lambda_d }{|z+m_{sc}|-\Lambda_d }+O(\Upsilon).
\ee
Since the  denominator satisfies $|z+m_{sc}(z)|-\sup_i|v_i|\geq 2-1/5$,
\be
    \Lambda_d\le C\Upsilon
\label{Lambdad}
\ee
 follows from the last equation. Combining it with \eqref{Lambdaoff} and \eqref{upsi},
we obtain  \eqref{alambinit}, and this completes the proof of Lemma
\ref{lm:initial}. \qed

\bigskip

\textit{Proof of Theorem \ref{thm:detailed}}. 
 Lemma \ref{thm: stepone} states that, in the event $\Omega^c$,
if $\Lambda_d(z)+\Lambda_o(z) \le R(z)$ then  $\Lambda_d(z)+\Lambda_o(z) \le S(z)$
with 
\[
R(z): =  (\log N)^{-3/2} (\theta(z))^{-1}, \qquad S(z) := 
 (\log N)^{6+2\al} 
\frac{ (\kappa+\eta)^{1/4} }{\sqrt{M\eta}\,\,}\theta(z).
\]
By   assumption \eqref{relkaeta} of Theorem \ref{thm:detailed}, 
 we have $S(z)< R(z)$ for any $z\in S$ and these functions are
continuous.
Lemma \ref{lm:initial} states that in the set $\wh \Omega^c$
the bound $\Lambda_d(z)+\Lambda_o(z)
 \le R(z)$  holds for $\eta=10$.

 Thus by a continuity argument,
 $\Lambda_d(z)+ \Lambda_o(z)\le S(z)$ in the set $\Omega^c\cap \wh \Omega^c$
 as long as the condition
 \eqref{relkaeta} is satisfied. Finally, once $ \Lambda_o(z)\le S(z)$ is proven,
we can use $S(z) \le R(z)$ (in the domain $D$) and Lemma \ref{thm: stepone} once more
to conclude the stronger bound on $\Lambda_o(z)$.
This proves  Theorem \ref{thm:detailed}.

We  record that combining the bound on $\Lambda_d, \Lambda_o$
with \eqref{41},
we also proved that under the assumption  \eqref{relkaeta} we have
\be
  \Lambda_d(z)+ \Lambda_o(z) + \Upsilon(z)\le C (\log N)^{16+4\al} 
\frac{ (\kappa+\eta)^{1/4} }{\sqrt{M\eta}\,\,}\theta(z)  
\qquad\mbox{in $\Omega^c\cap\wh\Omega^c$}.
\label{LLU}
\ee
\qed

\section{Local semicircle law}

In this section we strengthen the estimate of Theorem \ref{thm:detailed} for
the Stieltjes transform $m(z) = \frac{1}{N}\sum_i G_{ii}$. The key improvement is
that $|m-m_{sc}|$ will be estimated with a precision $(M\eta)^{-1}$
while the $|G_{ii}-m_{sc}|$ was controlled by a precision $(M\eta)^{-1/2}$ only
(modulo logarithmic terms and terms expressing the deterioration 
of the estimate near the edge).

\begin{theorem}\label{prop:mmsc} Assume the conditions of Theorem \ref{thm:detailed}
and recall the notations $\kappa= \kappa_E:= \big| |E|-2\big|$  and $\theta(z)$ from \eqref{q}.
 Define the domain
\be\label{domainD}
   D^*:= \Big\{ z=E+i\eta \in \C\; : \; |E|\le 5,  \;\;\frac{1}{N}\le \eta
\le 10\; , 
  \quad M\eta  \ge (\log N)^{24+6\al}\theta^4(z)(\kappa+\eta)^{1/2}\Big\}.
\ee
 Then for any $\e>0$  and $K>0$
there exists a constant $C=C(\e,K)$ such that
\be
   \P \Big( \bigcup_{z\in D^*}\Big\{ |m(z)-m_{sc}(z)|\ge 
 \frac{N^\e\theta^2(z) }{M\eta } \Big\}
\Big) \le \frac{C(\e, K)}{N^K}.
\label{eq:mmsc}
\ee
\end{theorem}

{\it Proof of Theorem \ref{prop:mmsc}.}
We will work in the set $\Omega^c\cap \wh\Omega^c$, which has almost
full probability by \eqref{B12}
and \eqref{whomest}. Note that the set
 $D^*$ is included in the domain defined by \eqref{relkaeta}, therefore
we can use the estimates from Section~\ref{ld-semi}.

As in \eqref{vi09}, where $\bar v=m(z)-m_{sc}(z)$, we have that 
$$
m-m_{sc}=-\frac{\zeta}{1-\zeta}\frac1N \sum_{i}\Upsilon_i
+O\Big(\frac{\zeta}{1-\zeta}(\Lambda_d+\Upsilon)^2 \Big)
$$
holds with a very high probability. Recall that $\zeta = m_{sc}^2(z)$
and we mostly omit the argument $z$ from the notations.
The quantities $\Lambda_d$, $\Upsilon_i$ and $\Upsilon$ 
were defined in \eqref{defLambda},  \eqref{seeqerror} and \eqref{defUpsilon}.
Then with \eqref{LLU} we have 
$$
m(z)-m_{sc}(z)=O\left(\frac{\zeta}{1-\zeta} \;\frac1N \sum_{j}\Upsilon_j\right)
+O\Big(\frac{N^\e}{|1-\zeta|}\frac{\theta^2(z)\sqrt{\kappa+\eta}}{M\eta} \Big)
$$
holds with a very high probability for any small $\e>0$. 
Recall that $\Upsilon_i = A_i + h_{ii} -Z_i$. 
We have, from \eqref{defA}, \eqref{lowerbound} and $\sigma_{ij}^2\le M^{-1}$,
$$
    A_j \le \frac{C}{M} + C\Lambda_o^2 \le CN^\e\frac{\theta^2\sqrt{\kappa+\eta}}{M\eta },
$$
where we used \eqref{LLU} to bound $\Lambda_o$ and \eqref{thetaest} to control the $C/M$ term.

We thus obtain that
\be\label{tempd28}
m-m_{sc}=O\left(\frac{\zeta}{1-\zeta}\left(
\frac1N\sum_{i}Z_i-
\frac1N\sum_{i}h_{ii}\right)\right)
+O\Big(\frac{N^\e}{|1-\zeta|}\frac{\theta^2\sqrt{\kappa+\eta}}{M\eta } \Big)
\ee
holds with a very high probability. 
Since $h_{ii}$'s are independent,  applying the 
first estimate in the large deviation Lemma \ref{generalHWT}, we have 
\be\label{sumhii}
\P\left(\Big|\frac1N\sum_{i}h_{ii}\Big|\geq (\log N)^{3/2+\al} \frac{1}{\sqrt{MN}} \right)\leq CN^{-c\log\log N}. 
\ee
On the complement event, the estimate $ (\log N)^{3/2+\al}(MN)^{-1/2}$ can be included in
the last error term in \eqref{tempd28}.
It only remains to bound 
$$
\frac1N\sum_{i=1}^NZ_i,
$$
whose moment is bounded in the next lemma which will be proved in Sections  \ref{mo-est} and \ref{sec:gen}.

\begin{lemma}\label{motN} 
For fixed $z$ in domain $D^*$ \eqref{domainD} and any even number $p$, we have
\be\label{52}
\E\left|\frac1N\sum_{i=1}^NZ_i\right|^{p}\leq C_{p}\left((\log N)^{3+2\al}X^2 \right)^p
\ee
for sufficiently large $N$, where 
\be\label{defX}
X=X(z):=(\log N)^{10+2\al} \frac{(\kappa+\eta)^{1/4}}{\sqrt{M\eta}}, \qquad z=E+i\eta, \quad
\kappa = \big| |E|-2\big|.
\ee
\end{lemma}

Using Lemma \ref{motN}, we have that for any $\e>0$ and $K>0$, 
$$
\P\left(\frac1N\left|\sum_{i=1}^NZ_i\right|\geq N^{\e}\frac{\sqrt{\kappa+\eta}}{M\eta }\right)\leq N^{-K}
$$
for sufficiently large $N$. Combining this with \eqref{sumhii} and \eqref{tempd28}
and noting that $|1-\zeta|\sim \sqrt{\kappa +\eta}$, see \eqref{smallz}, 
 we obtain \eqref{eq:mmsc} and complete the proof of Theorem~\ref{prop:mmsc}. 

\qed

\section{Empirical counting function}\label{sec:emp}

In this section we translate the information on the Stieltjes transform obtained
in Theorem~\ref{prop:mmsc} to an asymptotic on the empirical counting
function.  The main ingredient for the first step is
the following lemma based upon the Helffer-Sj\"ostrand formula.
We will formulate this lemma for general signed measures, but we will apply it
to the  Stieltjes  transform $m^\Delta=m-m_{sc}$ of the difference
between the empirical density and the semicircle law.
A similar statement was already proven in Lemma B.1 in \cite{ERSY}
and Lemma 7.7 in \cite{EYY}.

\begin{lemma}\label{lm:HS1} 
Let $\varrho^\Delta$ be a signed measure on the real line with
 $\mbox{supp} \;\varrho^\Delta \subset
[-K,K]$ for some fixed constant $K\geq 4$.
For any $E_1, E_2 \in [-3,3]$ and $\eta\in (0, 1/2]$ we
define $f(\lambda)=f_{E_1,E_2,\eta}(\lambda)$
to be a characteristic function of $[E_1, E_2]$ smoothed on scale $\eta$, i.e.,
$f\equiv 1$ on $[E_1, E_2]$, $f\equiv 0$ on $\R\setminus [E_1-\eta, E_2+\eta]$
and $|f'|\le C\eta^{-1}$, $|f''|\le C\eta^{-2}$. For any $x\in \R$, 
set $\kappa_x:= \big||x|-2\big|$.
Let $m^\Delta$ be the Stieltjes transform of $\varrho^\Delta$.
Suppose for some positive $U$, and  non-negative constant $A$ we have
\be
   |  m^\Delta (x+iy)|\le \frac{CU}{ y(\kappa_x+y)^A} \qquad \mbox{for} \qquad
1\geq y>0, \quad |x| \le K+1,
\label{trivv1}
\ee
and in case of $A>0$ we additionally assume $\eta \le 
 \frac{1}{2} \min\{\kappa_{E_1},\kappa_{E_2}\}$. 
Then
\be
   \left|\int f_{E_1,E_2,\eta}(\lambda)\varrho^\Delta(\lambda)\rd\lambda \right|  \le
  \frac{CU|\log \eta|}{\big[\min\{\kappa_{E_1},\kappa_{E_2}\}\big]^A}
\label{genHS1}
\ee
with some constant $C$ depending on $K$ and  $A$. 
\end{lemma}

\medskip

{\it Proof of Lemma \ref{lm:HS1}.} For simplicity, we drop the $\Delta$ superscript
in the proof.
Analogously to  (B.13), (B.14)  and (B.15) in \cite{ERSY}
we obtain that (with $f=f_{E_1,E_2,\eta}$)
\begin{eqnarray}\nonumber
\left|\int f(\lambda)\varrho(\lambda)\rd\lambda \right|
\leq && C\int_{\R^2} ( |f(x)| +|y| |f'(x)|) |\chi'(y)| 
| m(x+iy)| \rd x\rd y\non \\ 
&& +C\left|\int_{|y|\leq \eta}\int y f''(x) \chi(y)
\im m(x+iy)\rd x\rd y\right|\label{intr2fe1} \\
&&+C\left|\int_{|y|\geq \eta}\int_\R y f''(x)\chi(y) \im  m(x+iy)\rd x
\rd y\right|, \non
\end{eqnarray}
where $\chi(y)$ is  a smooth cutoff function with support in $[-1,1]$, with $\chi(y) = 1$ for
$|y|\leq  1/2$ and with bounded derivatives.
The first term is estimated by, with \eqref{trivv1}, 
\be\label{xxl}
 \int_{\R^2} ( |f(x)| +|y| |f'(x)|) |\chi'(y)| 
| m(x+iy)| \rd x\rd y
 \le CU. 
\ee

For the second term in r.h.s of \eqref{intr2fe1} we use 
that from  \eqref{trivv1}  it follows for any $1 \ge y>0$ that
\be
  y |\im  m(x+iy)|\le \frac{CU}{(\kappa_x+y)^A}.
\label{ym}
\ee
 With $|f''|\leq C\eta^{-2}$ and
\be
{\rm supp} f'(x)\subset\{|x-E_1|\leq \eta\}\cup\{|x-E_2|\leq \eta\} ,
\label{fpr}
\ee
we get
$$
  \mbox{second term in r.h.s of \eqref{intr2fe1}} 
\le \frac{CU}{ \big[\min\{\kappa_{E_1},\kappa_{E_2}\}\big]^A}.
$$

As in (B.17) and (B.19) in \cite{ERSY}, we integrate the third term in
 \eqref{intr2fe1} by parts first in $x$, then in $y$. Then we bound it with an absolute value by 
\be\label{temp7.501}
C\int_{|x|\leq K+1}\!\!\!\eta |f'(x)| |\re m(x+i\eta)|\rd x+ 
 C \int _{\R^2}  \!\!\!|f '(x)  \chi'(y)   \re m (x +iy) |
+
\frac C\eta \int_{\eta\le y\leq 1}\int_{ |x-E|\leq \eta} \!\!\!|\re m(x+iy)|\rd x\rd y.
\ee
The second term is bounded in \eqref{xxl}. By using \eqref{trivv1} and \eqref{fpr}
in the first term and  \eqref{trivv1}
in the third, we have
\begin{align}\nonumber
\eqref{temp7.501}\leq& \frac{CU}{ \big[\min\{\kappa_{E_1},\kappa_{E_2}\}\big]^A}+CU+
CU\eta^{-1}\sum_{k=1,2}\int_{|x-E_k|\leq \eta}\rd x
 \int_{\eta\le y\le 1}\frac{1}{y(\kappa_x+y)^A}\rd y  \non\\
\leq & \frac{CU|\log \eta|}{\big[\min\{\kappa_{E_1},\kappa_{E_2}\}\big]^A}.  
 \qquad \mbox{\qed}\non
\end{align}

Let $\la_1\le \la_2 \le \ldots \le \la_N$ be the ordered
eigenvalues of a universal Wigner matrix.
We define the {\it normalized empirical counting function}  by
\be
 {\mathfrak n}(E):= \frac{1}{N}\# \{ \lambda_j\le E\}
\label{deffn}
\ee
and the {\it  averaged counting function} by
\be\label{defnlambda}
n(E)=\frac1N\E \#[{\lambda_j\leq E}].
\ee
Finally, let
\be
 n_{sc}(E) :  = \int_{-\infty}^E \varrho_{sc}(x)\rd x
\label{nsc}
\ee
be the distribution function
of the semicircle law which is very close to
the counting function of $\gamma$'s,  $n^\gamma(E):=\frac{1}{N}\#[\ga_j\le E]$.

\bigskip

We will  need some control on the spectral edge, we recall the  Lemma 7.2 from \cite{EYY}.

\begin{lemma}\label{N-1/6}
(1)  Let the universal Wigner matrix
 $H$ satisfy \eqref{sum}, \eqref{defM} and  \eqref{subexp} with  
$M\geq (\log N)^9$. Then we have
\be\label{resnlambda3}
n(-3)\leq CN^{-c\log\log N}\,\,\,{\rm and }\,\,\,n(3)\geq 1-CN^{-c\log\log N}.
\ee

(2) Let $H$ be a generalized Wigner matrix with 
subexponential decay, i.e., \eqref{sum}, \eqref{defM}, \eqref{VV} and \eqref{subexp} hold.
 Then 
\be\label{resnlambda1}
n(-2-N^{-1/6+\e})\leq Ce^{-N^{\eps'}}\,\,\,{\rm and }
\,\,\,n(2+N^{-1/6+\e})\geq 1-Ce^{-N^{\eps'}},
\ee
for  any small $\eps>0$ with an $\e'>0$ depending on $\e$.
Furthermore, for $K\geq 3$, 
\be\label{resnlambda2}
n(-K)\leq e^{-N^\e\log K}\,\,\,{\rm and }\,\,\,n(K)\geq 1-e^{-N^\e\log K},
\ee
for some $\eps> 0$.
\end{lemma}
\bigskip

With these preliminary  lemmas, we have the following theorem that we
state for universal Wigner matrices and for their subclass,
the generalized Wigner matrices in parallel.

\begin{theorem}\label{prop:count}
Let $A=2$  for universal Wigner matrices and $A=1$ for generalized Wigner matrices.
Suppose that the universal Wigner matrix ensemble satisfies \eqref{sum}, 
\eqref{defM} and  \eqref{subexp} with  
$M\geq (\log N)^{24+6\al}$ and the
 generalized Wigner matrix ensemble
 satisfies \eqref{sum}, \eqref{defM}, \eqref{VV} and \eqref{subexp}.
We recall $M=N$ in the latter case.
Then for any $\e>0$ and $K\ge 1$
there exists a constant $C(\e,K)$ such that
$$
   \P\Big\{ \sup_{|E|\le 3} \big| {\mathfrak n} (E)-n_{sc}(E)\big| \, [\kappa_E]^A \le
 \frac{CN^{\e}}{M} \Big\}\ge 1- \frac{C(\e, K)}{N^K},
$$
where the ${\mathfrak n}(E)$ and $n_{sc}(E)$ were defined in \eqref{deffn} and \eqref{nsc}
and $\kappa_E=\big| |E|-2\big|$.
\end{theorem}

{\it Proof.} For definiteness, we will consider the case of generalized Wigner matrices, i.e., $A=1$.
In this case $M=N$,  $\delta_+\ge C_{inf}>0$ (see \eqref{de-de+2}) and thus 
$\theta(z) \le C(\kappa+\eta)^{-1/2}$ for $|z|\le 10$, see \eqref{defgz2}.
For simplicity of the presentation, we assume that $\theta(z) = (\kappa+\eta)^{-1/2}$ 
as overall constant factors do not matter (see the remark after \eqref{defgz2}).
We  set  $\eta=1/N$, $U= N^{\e-1}$ and
 apply  Lemma \ref{lm:HS1} to the difference $m^\Delta = m- m_{sc}$. Let 
$\varrho^\Delta = \varrho - \varrho_{sc}$, where $\varrho(x)= \frac{1}{N}\sum_j 
\delta(x-\lambda_j)$ is the 
normalized empirical counting measure of eigenvalues. First we check 
the conditions of   Lemma \ref{lm:HS1}.
%
To check that \eqref{trivv1} holds, set $L=(\log N)^{24+6\al}$ and
for a fixed $x$, let $y_x$ satisfy 
$  N y_x(\kappa_x+y_x)^{3/2} =L$, so that $x+iy_x\in D^*$.
Clearly \eqref{trivv1} holds for any $y\ge y_x$ with a very high probability by  \eqref{eq:mmsc}. In particular, we know that
\be\label{temp7.30}
|    m(x+iy_x)-   m_{sc}(x+iy_x)|\leq \frac{CU}{y_x(\kappa_x+y_x)}.
\ee
Consider $y< y_x$,
set $z=x+iy$, $z_x=x+iy_x$ and
estimate
\be\label{mmmsc}
   |   m(z)-m_{sc}(z)|
   \le 
     |   m(z_x)-m_{sc}(z_x)|
  + \int_y^{y_x} \big| \partial_\eta
\big( m(x+i\eta)-m_{sc}(x+i\eta)\big)\big|\rd \eta.
\ee
Note that
\begin{align}
   |\partial_\eta m(x+i\eta)| 
   = & \Big|\frac{1}{N}\sum_j \partial_\eta G_{jj}(x+i\eta)\Big|
 \\ \le & \frac{1}{N}\sum_{jk} |G_{jk}(x+i\eta)|^2 =
 \frac{1}{N\eta}\sum_j \im G_{jj}(x+i\eta) 
 = \frac{1}{\eta}\im m(x+i\eta),
\end{align}
and similarly
$$
  |\partial_\eta m_{sc}(x+i\eta)| = 
\Big|\int \frac{\varrho_{sc}(s)}{(s-x-i\eta)^2}\rd s\Big|
\le \int \frac{\varrho_{sc}(s)}{|s-x-i\eta|^2}\rd s
 = \frac{1}{\eta} \im m_{sc}(x+i\eta).
$$
Now we use the fact that the functions $y\to y\im m(x+iy)$ and
$y\to y\im m_{sc}(x+iy)$ are monotone increasing for any $y>0$
since both are Stieltjes transforms of a positive measure.
Therefore the integral in \eqref{mmmsc} can be bounded by
\be\label{intbb}
   \int_y^{y_x} \frac{\rd \eta}{\eta} \big[ \im   m(x+i\eta) +  
\im m_{sc}(x+i\eta)\big] \le y_x\big[ \im m(x+iy_x) +  
\im m_{sc} (x+iy_x)\big]  \int_y^{y_x} \frac{\rd \eta}{\eta^2}
\ee
By the choice of $y_x$ and using that
$\im \,   m_{sc}(z_x)\leq C\sqrt{\kappa_x+y_x} $, we have  
\be\label{temp5.10}
\im \,   m_{sc}(z_x )\leq \frac{ CU}{ y_x (\kappa_x+y_x)}.
\ee
and then $ \im \,  \E m(z_x)$ can be estimated from
 \eqref{temp7.30}.
Inserting these estimates into \eqref{mmmsc} and \eqref{intbb},
and using \eqref{temp7.30},
we get 
$$
 |  m(z)-m_{sc}(z)|\le   |   m(z_x)-m_{sc}(z_x)|
  +  \frac{ CU}{ y_x(\kappa_x+y_x)} \frac{y_x}{y}
 \le   \frac{ CU}{ y (\kappa_x+y )} 
$$
with a possible larger $C$ in the r.h.s.
 Thus \eqref{trivv1}   holds for the difference  $m^\Delta =m-m_{sc}$. 

The application of Lemma \ref{lm:HS1} shows that for  $\eta= 1/N$
\be
   \left|\int f_{E_1,E_2,\eta}(\lambda)\varrho(\la) \rd\lambda
 -\int f_{E_1,E_2,\eta}(\lambda) \varrho_{sc}(\lambda)\rd\lambda \right|  \le
  \frac{C N^{2\e}}{N \min\{\kappa_{E_1},\kappa_{E_2}\}+1}.
\label{genHS2}
\ee
Recall that $ f_{E_1,E_2,\eta}$ the characteristic function of the interval
$[E_1, E_2]$, smoothed on scale $\eta$ at the edges. 
The additional $1$ in the denominator in the r.h.s. of \eqref{genHS2}
comes from the case when $\kappa_{E_1}$, $\kappa_{E_2}$ are very small
and the trivial estimate $f\le 1$ with $\int \varrho=\int \varrho_{sc}=1$
 gives a better bound
than  Lemma \ref{lm:HS1}.

With the fact $y\to y\im m(x+iy)$ is  monotone increasing for any $y>0$, \eqref{temp5.10} implies a crude upper bound on the empirical density.
Indeed,  for any interval $I:=[x-\eta, x+\eta]$, with $\eta=1/N$, we have 
\be
   {\mathfrak n}(x+\eta)- {\mathfrak n}(x-\eta ) 
   \le C\eta\,\im \, m\big( x+ i\eta\big)
   \le Cy_x\,\im \, m\big( x+ iy_x\big)
  \le \frac{CN^{2\e} }{N\kappa_x+1}.
\label{density}
\ee
since $\eta=1/N\le y_x$ for any $x$.

Choose arbitrary  $E_1, E_2\in [-3,3]$, then  we have 
\begin{align}
  \Big|  {\mathfrak n}(E_1)-{\mathfrak n}(E_2) - \int f_{E_1, E_2, \eta}(\lambda)\varrho(\lambda)\rd\lambda\Big|
  \le &  C\sum_{j=1,2} \big[ {\mathfrak n}(E_j+\eta)-{\mathfrak n}(E_j-\eta)\big] \non \\
  \le &\sum_{j=1,2} \frac{C N^{2\e} }{N\kappa_{E_j}+1} \label{nne1}
\end{align}
from \eqref{density}.  Since $\varrho_{sc}$ is bounded, we also have
\be\label{nne2}
   \Big|  n_{sc}(E_1)-n_{sc}(E_2) - \int f_{E_1, E_2, \eta}(\lambda)\varrho_{sc}
  (\lambda)\rd\lambda\Big| \le C\eta = C/N.
\ee
Subtracting \eqref{nne1} and \eqref{nne2} and  using \eqref{genHS2},
we obtain that for any $E_1, E_2\in [-3,3]$ 
$$
  \Big|  \big[{\mathfrak n} (E_1)- {\mathfrak n} (E_2)\big] -  \big[n_{sc}(E_1)-n_{sc}(E_2)\big]\Big|
  \le \frac{C N^{2\e} }{N\min \{ \kappa_{E_1}, \kappa_{E_2}\}+1}
$$
with a very high  probability, i.e.,
apart from a set of probability smaller than $C(\e, K)N^{-K}$ for any $K$.
The estimate  \eqref{resnlambda2} from  Lemma \ref{N-1/6}
on the extreme eigenvalues shows that
 $\varrho$ is supported in $[-3,3]$ with very high 
probability, i.e., ${\mathfrak n} (-3)=n_{sc}(-3)=0$, ${\mathfrak n} (3)=n_{sc}(3)=1$. Thus we obtain that
\be
   \Big| {\mathfrak n}(E) - n_{sc}(E)\Big|
  \le \frac{C N^{2\e} }{N\kappa_E+1}
\label{fixE}
\ee
holds for any fixed $E\in [-3,3]$ with an overwhelming probability.

 We now choose a fine grid of equidistant points $E_j\in [-3,3]$
with $|E_j-E_{j+1}|\le N^{-1}$, then  \eqref{fixE} holds simultaneously
for every $E=E_j$ with an  overwhelming probability.
For any $E\in [-3,3]$ we can find an $E_j$ with $|E-E_j|\le N^{-1}$
and by \eqref{density} we obtain 
$$ 
   |{\mathfrak n} (E)- {\mathfrak n} (E_j)| \le {\mathfrak n} (E_j+1/N)- {\mathfrak n} (E_j-1/N)\le
\frac{C N^{2\e} }{N\kappa_{E_j}+1}.
$$
This guarantees that \eqref{fixE} holds simultaneously for all $E$. Since 
$\e>0$ was arbitrary, 
 this proves Theorem \ref{prop:count} for generalized Wigner matrices. 

The proof 
for universal Wigner matrices is very similar, just $M$ replaces $N$
in the estimates, $U=N^{\e}M^{-1}$ and instead of $\theta(z)\le C (\kappa+\eta)^{-1/2}$ one uses
$\theta(z) \le C(\kappa+\eta)^{-1}$ which follows from  \eqref{defgz2}. 
The main technical estimate \eqref{genHS2} is modified to
\be
   \left|\int f_{E_1,E_2,\eta}(\lambda)\varrho(\la) \rd\lambda
 -\int f_{E_1,E_2,\eta}(\lambda) \varrho_{sc}(\lambda)\rd\lambda \right|  \le
  \frac{C N^{2\e}}{M \big[\min\{\kappa_{E_1},\kappa_{E_2}\}\big]^2+1}
\label{genHS3}
\ee
and the rest of the proof is identical.
\qed
\medskip

\section{Location of eigenvalues}\label{sec:loc}

In this section we estimate the mean square 
  deviation of the eigenvalues from
their classical location. The main input is Theorem~\ref{prop:count},
the estimate on the counting function.
 For simplicity, we consider  only the case of generalized Wigner matrices.
Similar, but weaker results can be obtained along the same lines
for universal Wigner matrices.

\begin{theorem}\label{prop:lambdagamma} 
Let $H$ be a generalized Wigner matrix  with 
subexponential decay, i.e., assume that \eqref{sum}, \eqref{defM}, \eqref{VV} and \eqref{subexp} hold.
 Let $\lambda_j$ denote the eigenvalues of $H$ and 
$\gamma_j$ be their classical location, defined by \eqref{def:gamma}. 
Then for any  $\e_0<1/7$ and for any $K>1$ there exists a constant $C$, depending
on $K$ and $\e_0$, such that
\be
    \P \Big\{ \sum_{j=1}^N |\lambda_j-\gamma_j|^2 \le N^{-\e_0}\Big\}
  \ge 1-\frac{C}{N^K}.
\label{lambdaminusgamma}
\ee
and 
\be
    \sum_{j=1}^N \E |\lambda_j-\gamma_j|^2 \le CN^{-\e_0}.
\label{lambdaminusgammaE}
\ee
\end{theorem}

{\it Proof.} The proof of \eqref{lambdaminusgammaE}
directly follows from \eqref{lambdaminusgamma} by using the estimates
on the extreme eigenvalue \eqref{resnlambda2} from  Lemma \ref{N-1/6}.
For the proof of \eqref{lambdaminusgamma}, we
can assume that $\max_j |\lambda_j |\le 2+ N^{-1/7}$
since the  complement event has
a negligible probability by
 \eqref{resnlambda1} and \eqref{resnlambda2}
of Lemma \ref{N-1/6}. {F}rom Theorem~\ref{prop:count}
we can also assume that
\be
   |{\mathfrak n}(E)-n_{sc}(E)|\le \frac{CN^\e}{N\kappa_E}
\label{nnee}
\ee
holds for every $E\in \R$.

{F}rom the definition of $\gamma_{j}$ it follows that for $j\le N/2$, i.e., $\gamma_j\le 0$,
\be
 -2+ C_1\Big(\frac{j}{N}\Big)^{2/3} \le \gamma_{j} \le -2+ C_2\Big(\frac{j}{N}\Big)^{2/3}
\label{locgamma}
\ee
with some positive constants $C_1, C_2$.

Choose $\beta = \frac{2}{5}-\e$. Consider first those $j$-indices for which
 $C_0N^{1-3\beta/2}\le j \le N-C_0N^{1-3\beta/2}$ with a sufficiently large constant.
We choose $C_0$ so that \eqref{locgamma} would imply
$-2+2N^{-\beta}\le \gamma_j \le 2- 2N^{-\beta}$.
 We then  claim that
\be
 \lambda_j \in [-2+ N^{-\beta}, 2- N^{-\beta}]\qquad \mbox{for}\quad
  C_0N^{1-3\beta/2}\le j \le N-C_0N^{1-3\beta/2}.
\label{lamloc}
\ee
We will show that $\lambda_j\ge -2+ N^{-\beta}$, the upper bound is analogous.
Suppose that $\lambda_j$ were
smaller than $-2+N^{-\beta}$,
then ${\mathfrak n}(-2+N^{-\beta})\ge j$. On the other hand, $n_{sc}(-2+2N^{-\beta})\le j$
and thus
$$
   n_{sc}(-2+ N^{-\beta}) = n_{sc}(-2+2N^{-\beta}) -\int_{-2+N^{-\beta}}^{-2+2N^{-\beta}}
  \varrho_{sc}(x)\rd x \le j - cN^{-3\beta/2} 
$$
with some positive constant $c$.
Therefore
$$
   cN^{-3\beta/2} \le
  {\mathfrak n}(-2+ N^{-\beta}) - n_{sc}(-2+ N^{-\beta}) \le CN^{\beta+\e-1},
$$
where the second inequality follows from \eqref{nnee}, but this contradicts to the
choice $\beta = \frac{2}{5}-\e$.

\medskip
Let $j$ satisfy $C_0N^{1-3\beta/2}\le j \le N/2$;
the indices $ N/2\le j\le N-C_0N^{1-3\beta/2}$ can be treated analogously.
Note that $\lambda_{N/2}\le CN^{-1+\e}$ by \eqref{nnee}.
Define $c(j)$ to be index of the $\gamma$-point right below $\lambda_j$, i.e.,
$$
        \gamma_{c(j)}\le \lambda_j \le \gamma_{c(j)+1}.
$$
By \eqref{lamloc} we see that $-2+ \frac{1}{2}N^{-\beta}\le \gamma_{c(j)}\le CN^{-1+\e}$
and
from \eqref{nnee} and \eqref{locgamma} it follows that
\be
   |c(j)-j|\le \frac{CN^\e}{2+\gamma_{c(j)}} \le CN^{\e+\beta}.
\label{cjj1}
\ee
By the choice of $\beta$ we have
$\e+\beta < 1-\frac{3}{2}\beta$, i.e., \eqref{cjj1} implies $|c(j)-j|\ll  j$.
Using now \eqref{locgamma}, we have
\be
   |c(j)-j|\le \frac{CN^\e}{\gamma_{c(j)}+2} \le \frac{CN^{2/3+\e}}{c(j)^{2/3}}
  \le \frac{CN^{2/3+\e}}{j^{2/3}} .
\label{cjj}
\ee
Finally, we can estimate
$$
  |c(j)-j| = N \Big| \int_{\gamma_{c(j)}}^{\gamma_j} \varrho_{sc}(x)\rd x\Big| \ge 
  CN|\gamma_{c(j)}-\gamma_j| (2+\gamma_j)^{1/2} \ge  CN|\gamma_{c(j)}-\gamma_j| 
\Big(\frac{j}{N}\Big)^{1/3},
$$
using $|c(j)-j|\ll j$ and hence $(2+\gamma_j)$ and $(2+\gamma_{c(j)})$ are
comparable. In the last step we also used \eqref{locgamma}.
Combining this with \eqref{cjj}, we have
$$
   |\gamma_{c(j)}-\gamma_j| \le C\frac{|c(j)-j|}{N^{2/3}j^{1/3}} \le \frac{CN^\e}{j}.
$$
and the same estimate holds for $|\gamma_{c(j)+1}-\gamma_j|$ and thus
$$ 
|\lambda_j-\gamma_j| \le  \frac{CN^\e}{j}
$$
as well.
Therefore
\be
  \sum_{C_0N^{1-3\beta/2}\le j\le N/2} |\lambda_j -\gamma_j|^2 \le CN^{2\e-1+3\beta/2}
  \le CN^{-2/5 +\e/2}
\label{bulksum}
\ee
by the choice of $\beta$ and similar estimate holds 
for the sum over the indices $N/2\le j\le N-C_0N^{1-3\beta/2}$ as well.

\medskip

Now we consider the indices $j\le C_0N^{1-3\beta/2}$ and $\lambda_j \ge -2- N^{-\beta}$.
By a similar argument that proved \eqref{lamloc}, we can see
that there is a constant $C_3$ such that $\lambda_j\le -2+C_3N^{-\beta}$,
otherwise ${\mathfrak n}(-2+C_3N^{-\beta})\le j$, but $n_{sc}(-2+C_3N^{-\beta})\ge j + cN^{-3\beta/2}$,
which would contradict \eqref{nnee}. It is easy to see that $\gamma_j\le -2+ CN^{-\beta}$
for all $j\le C_0N^{1-3\beta/2}$, therefore in this regime we estimate $|\lambda_j -\gamma_j|
\le CN^{-\beta}$ and thus
\be
   \sum_{j=1}^{C_0N^{1-3\beta/2}} |\lambda_j-\gamma_j|^2 {\bf 1}(\lambda_j \ge -2-N^{-\beta})
  \le C_0N^{1-3\beta/2} (CN^{-\beta})^2 \le CN^{-2/5 + 7\e/2}.
\label{middlesum}
\ee
The indices $j\ge N-C_0N^{1-3\beta/2}$ and $\lambda_j \le 2+ N^{-\beta}$
can be treated similarly.

\medskip

Finally we deal with the extreme eigenvalues $\lambda_j\le -2- N^{-\beta}$ 
with index $j\le C_0N^{1-3\beta/2}$ and
we can assume that $\lambda_j\ge -2-N^{-1/7}$. 
For these indices $-2\le \gamma_j\le -2 + CN^{-\beta}$
and we can estimate
$$
   |\lambda_j-\gamma_j|\le C|\lambda_j+2|.
$$
For any $a$ with 
$N^{-\beta }\le a \le N^{-1/7}$, we have $n_{sc}(-2-a)=0$, thus we
obtain from \eqref{nnee} that
$$
    {\mathfrak n}(-2-a) \le \frac{CN^\e}{Na}.
$$
Therefore
\begin{align}\label{extre}
   \sum_j |\lambda_j -\gamma_j|^2 {\bf 1}(-2-N^{-1/7}\le \lambda_j \le -2-N^{-\beta})
   \le & C\sum_j |\lambda_j+2|^2 {\bf 1}(-N^{-1/7}\le \lambda_j+2 \le -N^{-\beta}) \non \\
 \le & C\int_0^{N^{-1/7}} a \cdot \frac{CN^\e}{Na}\;\rd a \non \\
\le & CN^{-1/7+\e}.
\end{align}
The other extreme eigenvalues, $\lambda_j\ge 2+N^{-\beta}$, are treated analogously.

Combining \eqref{bulksum}, \eqref{middlesum} and \eqref{extre} and choosing $\e$
sufficiently small in the definition of $\beta$, we proved 
\eqref{lambdaminusgamma} with any $\e_0 <1/7$. \qed

\section{Moment Estimates of Error Terms}\label{mo-est}

In this section we prove  the second and fourth moment estimates of Lemma \ref{motN}; 
the general cases will be proved in  Section \ref{sec:gen}.

\begin{definition}
Define the operator ${\mathbb {IE}_i}$ as 
\be\label{defIEi}
{\mathbb {IE}_i}\equiv \mathbb I-\E_{\ba^i},
\ee
where $\mathbb I$ is identity operator. 
\end{definition}

Recall the definition of $Z_i$, which we rewrite   as 
\be \label{y112}
Z_i= {\mathbb {IE}_i}   Z_{ii}^{(i)},     
\qquad Z_{ii}^{(i)}= \sum_{k,l\neq i}\overline {\ba^i_{k}}G_{k\,l}^{(i)}\ba^i_{l}
= {\ba^i} \cdot G^{(i)}\ba^i.
\ee
We first prove a bound on the Green function  $G_{k\,l}^{(i)}$.

\begin{lemma}\label{boundGij}
Recall the definition of $X$ in \eqref{defX}. 
Let $t$ be any fixed positive integer,  ${\T}=\{k_1,k_2\ldots k_t\}\in \N^t$,
 $1\leq k_i\leq N$ for any $1\leq i\leq t$. Then there exists
a constant $C_t$, depending only on $t$, such that for any $z\in D^*$ in \eqref{domainD} in the 
set $\Omega^c$ \eqref{defOmega}, we have 
\begin{eqnarray}\label{lkGTlk}
 \max_{k,l:l\neq k,\,\,\, l,k\notin{\T}}|G^{{(\T)}}_{lk}(z)|
&\leq& C_t X(z),
\\\label{kkGTkk}
 \max_{k: k\notin{\T}}|G^{{(\T)}}_{kk}(z)-m_{sc}(z)|
&\leq& C_t X(z) \theta(z)
\end{eqnarray}
and for some constant $c$, $C$ independent of $t$, 
\be\label{cCTGkk}
c\leq\min_{k: k\notin{\T}}|G^{{(\T)}}_{kk}(z)|\leq \max_{k: k\notin{\T}}|G^{{(\T)}}_{kk}(z)|\leq C,
\ee
for sufficiently large $N$.
\end{lemma}
{\it Proof }
Consider first the case $t=0$. Let $Y$ denote the event inside the probability in the equation 
  \eqref{mainlsresult}. 
The proof of Theorem \ref{thm:detailed} yields that $\Omega^c \subset Y^c$. It is clear that 
 \eqref{kkGTkk} holds 
in the  event  $Y^c$  and this proves  \eqref{kkGTkk} in $\Omega^c$ 
 in the case $t=0$. 
 Similarly, in the case of $t=0$, we can prove  \eqref{lkGTlk} using the 
event in the  equation \eqref{mainlsresult2}.  
By  definition of the domain $D^*$,  the right side of \eqref{kkGTkk} is
 $o(1)$ and this proves  \eqref{cCTGkk} in the case $t=0$.

For  the case $t=1$ and $i_1=i$,  using \eqref{GiiGjii} and \eqref{GijGkij}, we obtain that
\begin{eqnarray}\label{tempd352}
|G^{(i)}_{lk}|\!\!\!\!\!\!\!\!\!\!\!&&\leq |G_{lk}|+|G_{li}G_{ik}||G_{ii}|^{-1},\\\nonumber
|G^{(i)}_{kk}-m_{sc}|\!\!\!\!\!\!\!\!\!\!\!&&\leq |G_{kk}-m_{sc}|+|G_{ki}G_{ik}||G_{ii}|^{-1}.
\end{eqnarray}
Since $X^2\ll X$ in  $D^*$,   \eqref{kkGTkk} and \eqref{lkGTlk} in the case $t=1$ follows
 from  \eqref{tempd352}
and the case $t=0$.  Repeating this process,  we prove  \eqref{kkGTkk}
 and \eqref{lkGTlk} for any $t>1$ by induction on $t$.
\qed

\bigskip

Now we return to the second and fourth moment estimates of Lemma \ref{motN}.

\bigskip

\subsection{Proof of Lemma \ref{motN} for $p=2$. }
Now we prove the special case of Lemma \ref{motN} for $p=2$. 
The second moment  of $\sum_{i=1}^NZ_i$ is given by  
 \be\label{y2}
 \frac{1}{N^2} \E\left|\sum_{i=1}^N Z_i \right|^2
= \frac{1}{N^2} \E\sum_{\al\neq \beta} 
\overline{Z_\al}Z_\beta +\frac{1}{N^2} \E\sum_{\al} 
 \left | Z_\al\right | ^2.
 \ee 
We start with estimating  the first term of \eqref{y2} for 
 $\al=1$ and $\beta=2$. The basic idea is to rewrite  $G^{(1)}_{k\,l}$ as
\be\label{expandG1P}
G^{(1)}_{k\,l}=P^{(1),\emptyset}_{k\,l}+P^{(1),(2)}_{k\,l},\,\,\,\,\,\,\,\,\,\,\,\,k,l\neq 1,
\ee  
with $P^{(1),(2)}_{k\,l}$  independent of $\ba^1$, $\ba^2$ and 
$P^{(1),\emptyset}_{k\,l}$  independent of $\ba^1$. 
The $P$'s have two upper indices. The first one refers to the fact that it comes from the $H^{(1)}$
minor (i.e. follows the upper index of $G^{(1)}$) and the second
one indicates the additional independence.   

To construct this decomposition for $ k, l \notin \{1,2  \}$,  by \eqref{GiiGjii} or
 \eqref{GijGkij} we can  rewrite  $G^{(1)}_{k\,l}$  as 
\be\label{y13}
G^{(1)}_{k\,l}=G^{(12)}_{k\,l}+\frac{G^{(1)}_{k\,2}G^{(1)}_{2\,l}}{G^{(1)}_{2\,2}}, 
\qquad k, l \notin \{1,2  \}.
\ee
The first term on the r.h.s is independent of $\ba ^2$. With Lemma \ref{boundGij}, we have
that the bound
\be\label{77x2}
\left|\frac{G^{(1)}_{k\,2}G^{(1)}_{2\,l}}{G^{(1)}_{2\,2}}\right|\leq CX^2
\ee
holds with a very high probability.  

Next  we define $P^{(1)}$  for ($k,l \neq 1$).
\begin{enumerate}
	\item If $ k, l \neq 2$, 
	\be\label{defP1}
	P^{(1),(2)}_{k\,l}=G^{(12)}_{k\,l},\,\,\,P^{(1),\emptyset}_{k\,l}=
\frac{G^{(1)}_{k\,2}G^{(1)}_{2\,l}}{G^{(1)}_{2\,2}}=G^{(1)}_{k\,l}-G^{(12)}_{k\,l}.
	\ee
	\item If $k=2$ or $l=2$,
	\be\label{defP2}
	P^{(1),(2)}_{k\,l}=0,\,\,\,P^{(1),\emptyset}_{k\,l}=G^{(1)}_{k\,l}.
	\ee
\end{enumerate}
Hence \eqref{expandG1P} holds and $P^{(1),(2)}_{k\,l}$ is independent of  $\ba^2$.

With this convention, we have the following expansion of $Z_1$ 
\be\label{y14fake}
Z_1= \mathbb{IE}_1 \ba^{1}\cdot{P^{(1), (2)}} {\ba^{1}}  + \mathbb{IE}_1\ba^{1} \cdot  {P^{(1), \emptyset}}  {\ba^{1}}.
\ee

\begin{lemma}\label{y4}
For $N^{-1} \le \eta \le 10$ and fixed $p\in \N$, we have the following estimates 
\be\label{y7}
\E\left | \ba^{1}\cdot {P^{(1), \emptyset}} {\ba^{1}}  \right | ^p
\leq C_p \left((\log N)^{3+2\al}\right)^p X^{2p},
\ee
\be\label{y8}
\E\left |\ba^{1} \cdot{P^{(1), (2)}}  {\ba^{1}}  \right |^p 
\leq  C_p \left((\log N)^{3+2\al}\right)^p X^{p}.
\ee
\end{lemma}

\bigskip 

Since $X^2 \le X$ in $D^*$,   this lemma also implies that 
\be
\label{EZK}
\E\left |Z_i \right |^p 
\leq   C_k \left((\log N)^{3+2\al}\right)^p X^{p}, \qquad 1\leq i\leq N.
\ee

\noindent 
{\it Proof. } First we rewrite  $\ba^{1}\cdot {P^{(1), \emptyset}} {\ba^{1}} $ as follows
\be\label{y100}
\ba^{1}\cdot {P^{(1), \emptyset}} {\ba^{1}}
=\sum_{k,l\neq 2}\overline {\ba^{1}_k}
\left(\frac{G^{(1)}_{k\,2}G^{(1)}_{2\,l}}{G^{(1)}_{2\,2}}\right)\ba^{1}_l
+\sum_{k\neq 2}\overline {\ba^{1}_k}G^{(1)}_{k\,2}\ba^{1}_2
+\sum_{l\neq 2}\overline {\ba^{1}_2}G^{(1)}_{2\,l}\ba^{1}_l
+\overline {\ba^{1}_2}G^{(1)}_{2\,2}\ba^{1}_2.
\ee
By the  large deviation estimate  \eqref{resgenHWTO}, we have 
\be\label{y72}
\P \left ( \left |\sum_{k,l\neq 2}\overline {\ba^{1}_k}
\left(\frac{G^{(1)}_{k\,2}G^{(1)}_{2\,l}}{G^{(1)}_{2\,2}}\right)\ba^{1}_l   \right | 
\ge C (\log N)^{3+2\al} X^2 \right ) \le N^{-c \log \log N} .
\ee
Similarly,  from  \eqref{resgenHWTD}, using $a_i$ as $\ba^1_3$, $\ba^1_4,\ldots, \ba^{1}_N$
and keeping $\ba^{1}_2$ fixed, we have 
\be\label{y721}
\P \left ( \left | \sum_{k\neq 2}\overline {\ba^{1}_k}G^{(1)}_{k\,2}\ba^{1}_2  \right | 
\ge  C (\log N)^{3/2+\al} X |\ba^{1}_2|  \right ) \le N^{-c \log \log N} .
\ee
By \eqref{resboundhij},  $\|\ba^{1}\|_\infty \le (\log N)^{2 \alpha}M^{-1/2}$ holds with a
very high probability. We can thus replace 
$|\ba^{1}_2| $ by $(\log N)^{2 \alpha}M^{-1/2}$ in \eqref{y721}. The third term in
 \eqref{y100} can be estimated  in the same way, 
and the last term  can be bounded by  $(\log N)^{4\al}\frac{1}{M}$ with very high probability.

Since $\eta \le 10$,  by the definition of $X$ in \eqref{defX} we have 
\be\label{XX1N}
X^2\geq C(\log N)^2/M .
\ee
Thus 
\[
 (\log N)^{3/2+3\al} \frac{X}{\sqrt{M}}+(\log N)^{4\al}\frac{1}{M} \le C (\log N)^{3+2\al} X^2,
\] 
and we have proved that 
\be\label{y722}
\P \left ( \left | \ba^{1}\cdot {P^{(1), \emptyset}} {\ba^{1}}  \right | 
\le C (\log N)^{3+2\al} X^2\right ) \ge 1- N^{-c \log \log N} .
\ee
This  inequality implies  the desired inequality \eqref{y7} except for the contribution from the exceptional set where 
\eqref{y722} fails.
Since all Green functions are bounded by $\eta^{-1} \le N$, 
the contribution from the exceptional set is negligible and this proves   \eqref{y7}. 
Finally,  a similar proof yields \eqref{y8}. 
\qed

\bigskip 

Exchange the index $1$ and $2$, we can define $P^{(2), (1)}$ and
 $P^{(2), \emptyset}$ and expand $Z_2$ as 
\be\label{y142}
Z_2= \mathbb{IE}_2 \ba^{2}\cdot{P^{(2), (1)}} {\ba^{2}}  + 
\mathbb{IE}_2\ba^{2} \cdot  {P^{(2), \emptyset}}  {\ba^{2}}.
\ee
Here $P^{(2),(1)}_{k\,l}$ is independent of $\ba^2$ and $\ba^1$;  
 $P^{(2),\emptyset}_{k\,l}$ is independent of $\ba^2$.
Combining \eqref{y142} with  \eqref{y14fake}, we have 
\be\label{y77}
 \E
 \overline Z_1
 Z_2
=
 \E \left [ 
  \left(  \mathbb{IE}_1  \left \{ \overline{{\ba^{1}}\cdot {P^{(1), (2)}}  {\ba^{1}}  +{\ba^{1}}  \cdot {P^{(1), \emptyset}} {\ba^{1}}}\right \} \right) 
\left(   \mathbb{IE}_2  \left \{  {\ba^{2}}\cdot P^{(2), (1)} \ba^{2}  +{\ba^{2}}\cdot P^{(2),   \emptyset} \ba^{2} \right \}   \right) \right ].
\ee
The only non-vanishing term on the right-hand side is 
\be\label{tempd77}
 \E   \left ( \mathbb{IE}_1  \overline {{\ba^{1}} \cdot  {P^{(1), \emptyset}}  {\ba^{1}}}\right) 
  \left( \mathbb{IE}_2  {\ba^{2}} \cdot
P^{(2), \emptyset}  \ba^{2} \right) .
\ee
By the Cauchy-Schwarz inequality and  Lemma \ref{y4},  we  obtain
\be
|\E\overline Z_1Z_2|\leq  C \left((\log N)^{3+2\al}\right)^2 X^{4}.
\ee
Similarly, Lemma \ref{y4} and \eqref{XX1N}  imply that 
\[
\E
 \left | Z_1\right | ^2
 \le  C\left((\log N)^{3+2\al}\right)^2 X^{2}\leq   C M \left((\log N)^{3+2\al}\right)^2 X^{4}.
\] 
Since the indices $1$ and $2$ can be replaced by $\alpha \not = \beta$,  together with \eqref{y2}  we have thus proved Lemma \ref{motN} 
for $p=2$.

\bigskip 

\subsection{Proof of Lemma \ref{motN} for $p=4$}\label{sec:4}
Now we prove the special case of Lemma \ref{motN} for $p=4$: 
\begin{eqnarray}\label{y10}
N^{-4} \E\left|\sum_{i=1}^NZ_i\right|^4
& \le  & C N^{-4}   \sum_{1\leq \al<\beta<\chi< \gamma\leq N} 
\left | \E \; \overline Z_\al\overline Z_\beta  Z_\chi Z_\gamma  \right | \\\nonumber
&&+ C N^{-4}  \sum_{1\leq \al<\beta<\chi\leq N}  \left |   \E |Z_\al|^2\overline
 Z_\beta Z_\chi  \right | +\ldots\\\nonumber
&& +C N^{-4}   \sum_{1\leq \al<\beta\leq N} \left(\E |Z_\al|^2  |Z_\beta|^2+ 
 \left |  \E |Z_\al|^2\overline Z_\al Z_\beta   \right | \right) +\ldots\\\nonumber
&& +C N^{-4}   \sum_{1\leq \al\leq N} \E |Z_\al|^4.
\end{eqnarray}

Here $\ldots$ means the permutation
of the ordered indices and the complex conjugate operators. 
We are going to compute the first two terms in the
r.h.s of \eqref {y10}.  The other two terms
can be treated analogously.
By the permutation symmetry of the indices, we can assume that
 $\al=1$, $\beta =2$, $\chi=3$ and $\gamma=4$. 
As in the estimate for the 
second moment, the key idea is to decompose $Z^{(1)}_{11}$ in a suitable way:

\begin{lemma}\label{Q1E}
There exist two decompositions of $Z_{11}^{(1)}$ 
\be\label{KEP}
Z_{11}^{(1)}= \sum_{{\T} \subset \{2, 3\} } {\ba^{1}} \cdot Q^{ (1), (\T)} 
\ba^1,\,\,\,\,\,\,
Z_{11}^{(1)}= \sum_{{\T} \subset \{2, 3, 4\} }{\ba^{1}} \cdot R^{ (1), (\T)} 
\ba^1,
\ee
such  that   $Q^{(1), (\T)}$ and  $R^{(1), (\T)} $ are
 independent of the rows in $ {\T} \cup \{ 1 \} $, i.e., 
\be\label{PQ0}
\frac{\partial \left(\ba^{1} \cdot Q^{(1), ({\T})}  { \ba^{1}}\right)}{\partial\ba_j^i}=
\frac{\partial \left(\ba^{1} \cdot Q^{(1), ({\T})}  { \ba^{1}}\right)}{\partial\overline{\ba_j^i}}=0, 
\,\,\,\,i\in {\T}\subset\{2,3\},
 \,\,\,1\leq j\leq N. 
\ee
and \be\label{PR0}
\frac{\partial \left(\ba^{1} \cdot R^{(1), ({\T})}  { \ba^{1}}\right)}{\partial\ba_j^i}=
\frac{\partial \left(\ba^{1} \cdot R^{(1), ({\T})}  { \ba^{1}}\right)}{\partial\overline{\ba_j^i}}=0, \,\,\,\,i\in {\T}\subset\{2,3,4\}, \,\,\,1\leq j\leq N. 
\ee
Furthermore, the decompositions can be chosen in such a way that for all $N^{-1} \le \eta \le 10$ the following estimates hold:
\be\label{EQK}
\E\left | {\ba^{1}}\cdot  {Q^{(1), {(\T)}}} {\ba^{1}}  \right | ^ p \leq C_p \left((\log N)^{3+2\al}\right)^p (X^{3-|{\T}|})^{p}, \quad p \in \N
\ee
and 
\be\label{ERK}
\E\left | {\ba^{1}}\cdot  {R^{(1), {(\T)}}} {\ba^{1}}  \right | ^ p \leq C_p \left((\log N)^{3+2\al}\right)^p (X^{4-|{\T}|})^{p}, \quad p \in \N.
\ee
\end{lemma}

\bigskip

We postpone the proof of this lemma  and first finish the proof of  
Lemma \ref{motN} in the case of $p=4$. It is clear that Lemma \ref{Q1E} holds  for different
index combinations. E.g. $Z_{22}^{(2)}$ can be decomposed as 
\be
Z_{22}^{(2)}= \sum_{{\T} \subset \{1, 3, 4\} } {\ba^{2}} \cdot R^{ (2), (\T)} \ba^2 
\ee
and $R^{(2)}$'s have the same properties (except for the exchange of $1$ and $2$) 
as  $R^{(1)}$ in \eqref{PR0} and \eqref{ERK} . 
 By this property,  we can  estimate the first  term on the r.h.s. of \eqref{y10} by 
\begin{eqnarray} 
&& \E \; 
 \left ( \mathbb{IE}_1  \overline  { Z_{11}^{(1)}} \right ) \; \left ( \mathbb{IE}_2 \overline  { Z_{22}^{(2)}}\right )
\;  \left ( \mathbb{IE}_3 Z_{33}^{(3)}\right )  \;  \left ( \mathbb{IE}_4 { Z_{44}^{(4)}} \right )  \\\nonumber 
\le &&
 \E
\overline{ \left [  \mathbb{IE}_1 \sum_{{\T}_1 \subset \{2, 3, 4\} }  { \ba^{1}}\cdot R^{ (1), ({\T}_1)}  \ba^{1}  \right ] } \times 
\overline{  \left [  \mathbb{IE}_2 \sum_{{\T}_2 \subset \{1, 3, 4\} }  { \ba^{2}}\cdot R^{ (2), ({\T}_2)}  \ba^{2}  \right ] } 
\big [ \cdots R^{ (3), ({\T}_3)} \cdots \big ] \big [  \cdots R^{ (4), ({\T}_4)} \cdots \big ].
\end{eqnarray}
Consider a term consisting of products of factors  with $  \cap_{j= 1, 2, 3, 4} ({\T}_j \cup \{j\} ) \not = \emptyset $. 
Then there is an element $\ell \in \{1, 2, 3, 4\}$ in the common intersection so that  integration w.r.t.  the row $\ba^\ell$
vanishes.  Hence the nonvanishing terms consist of products of term with $  \cap_{j= 1, 2, 3, 4} ({\T}_j \cup \{j\} )  = \emptyset $, i.e., 
$  \cup_{j= 1, 2, 3, 4} ({\T}_j \cup \{j\} )^c = \{1, 2, 3, 4\}$.  Here the notation $^c$ means the complement in $\{1,2,3,4\}$. Thus we have 
\[
\sum_{j=1}^4  (4- |{\T}_j|-1) \ge 4 \Longrightarrow \sum_{j=1}^4  4- |{\T}_j| \ge 8.
\]
Using \eqref{ERK} and Schwarz inequality,   we have  thus proved that 
\[
 \left| \E \; 
 \left ( \mathbb{IE}_1  \overline  { Z_{11}^{(1)}} \right ) \; \left ( \mathbb{IE}_2 \overline  { Z_{22}^{(2)}}\right )
\;  \left ( \mathbb{IE}_3 Z_{33}^{(3)}\right )  \;  \left ( \mathbb{IE}_4 { Z_{44}^{(4)}} \right ) \right| \\
\le C\left((\log N)^{3+2\al}\right)^4 X^{8}.
\]

We now estimate the second  term in r.h.s of \eqref{y10}. 
\begin{align} 
  \E \; 
 \left | \mathbb{IE}_1    { Z_{11}^{(1)}} \right |^2 \, \left ( \mathbb{IE}_2 \overline  { Z_{22}^{(2)}}\right )
\,  \left ( \mathbb{IE}_3 Z_{33}^{(3)}\right )  & = 
 \E
 \left ( \overline {  \mathbb{IE}_1  \sum_{{\T}_0 \subset \{2, 3\} }  { \ba^{1}} \cdot Q^{ (1), (\T_0)} \ba^{1}  }  \right ) 
 \;  \left (   \mathbb{IE}_1  \sum_{{\T}_1 \subset \{2, 3\} }  { \ba^{1}}\cdot Q^{ (1), (\T_1)}  \ba^{1}  \right )  \nonumber   \\
&  \times   
 \left [  \overline{  \mathbb{IE}_2 \sum_{{\T}_2 \subset \{1, 3\} }  { \ba^{2}}\cdot Q^{ (2), (\T_2)}  \ba^{2} } \right ]  \times 
   \left [  \mathbb{IE}_3\!\!\!\sum_{{\T}_3 \subset \{1, 2\} }  { \ba^{3}}\cdot Q^{ (3), (\T_3)} \ba^{3}  \right ]
\end{align}
Consider a term consisting of products of factors  with $[ \cap_{j=  2, 3} ({\T}_j \cup \{j\} ]\cap  {\T}_0 \cap {\T}_1  \not = \emptyset $. 
Then there is an element $\ell \in \{2, 3\}$ in the common intersection and the integration w.r.t.  the row $\ba^\ell$ vanishes. 
Thus the nonvanishing terms consist of products of term with $[ \cap_{j=  2, 3} ({\T}_j \cup \{j\} ]\cap  {\T}_0 \cap {\T}_1  = \emptyset $. In particular, 
$  \{2, 3\} \subset  \cup_{j= 2, 3} ({\T}_j \cup \{j\} )^c \cup [\{2, 3\} \setminus {\T}_0 ]  \cup [\{2, 3\} \setminus {\T}_1 ]$.  Here the notation $^c$ means the complement in $\{1,2,3\}$. Thus we have 
\[
\sum_{j=0}^3  (2- |{\T}_j|)   \ge 2 \Longrightarrow \sum_{j=0}^3  3- |{\T}_j|    \ge 6.
\]
Using \eqref{EQK},  \eqref{XX1N} and  a Schwarz inequality,  we have  
\[
 N^{-1} \left |  \E \; 
 \left | \mathbb{IE}_1   { Z_{11}^{(1)}} \right |^2 \; \left ( \mathbb{IE}_2 \overline  { Z_{22}^{(2)}}\right )
\;  \left ( \mathbb{IE}_3 Z_{33}^{(3)}\right ) \right |  \\
\le  \frac CN \left((\log N)^{3+2\al}\right)^4 X^6  \ll  C\left((\log N)^{3+2\al}\right)^4X^8.
\]
For the other terms in \eqref{y10}, we can just use Schwarz inequality and \eqref{EZK}. We have thus proved the Lemma \ref{motN}
for $p=4$. 

\bigskip

We now prove Lemma \ref{Q1E}. First we prove the properties of $Q$'s.
 Notice that the  decomposition  with $Q$'s in \eqref{KEP}  removes the dependence on rows $2, 3$. 
The starting point is  an expansion of $G^{(1)}_{k\,l}$  
\be\label{y123}
G^{(1)}_{k\,l}=\sum_{{\T}\subset \{2,3\}}Q^{(1),{(\T)}}_{k\,l}
=Q^{(1),\emptyset}_{k\,l}+Q^{(1),(2)}_{k\,l}+Q^{(1),(3)}_{k\,l}+Q^{(1),(2,3)}_{k\,l},
\ee  
where  $Q^{(1),{\T}}_{k\,l}$  is independent of the rows  and columns  in $ {\T} \cup \{ 1 \} $. 
Using the
 notation $(1 \,{\U})$ for $(\{1\}\cup {\U})$, one can check that a 
solution for $Q$ is given by 
\be\label{temp732}
Q^{(1), (\T)}_{k\,l}=  \sum_{{\U} : 
 {\T}\subset {\U}\subset \{2,3\} \backslash\{k, l\}} (-1)^{| {\U}|-|{\T}|}G^{(1  {\U})}_{k\,l}.
\ee 
Thus 
$Q^{(1), (\T)}_{k\,l}=0$ if $k$ or $l\in \T$. 
By definition of $Z_{11}^{(1)}$ 
\eqref{y112} and \eqref{y123}, we have that  the  $Q$'s  satisfy \eqref{KEP}.  
For any  fixed ${\T}$,   $Q^{(1),{\T}}_{k\,l}$ is independent of the rows 
 (column) in $ {\T} \cup \{ 1 \} $. 
Thus  we proved \eqref{PQ0}.

In order to prove \eqref{EQK},  we give another representation of the $Q$'s. 
We begin by   removing the dependence of  the $(kl)$ matrix element of the Green function  
on the index $3$ for $k,l >  3$.  By  \eqref{GiiGjii} or \eqref{GijGkij}, 
we can rewrite  the first term of r.h.s of \eqref{y13} as
\be\label{y16}
G^{(12)}_{k\,l}=G^{(123)}_{k\,l}+\frac{G^{(12)}_{k\,3}G^{(12)}_{3\,l}}{G^{(12)}_{3\,3}},
\,\,\,\,\,\,\,\,\,\, k,l  \notin \{1,2,3 \}.
\ee
This removes the dependence of $G^{(12)}_{k\,l}$ on the index $3$ with the last term as the error term. 
For the last  term on r.h.s of \eqref{y13}, using \eqref{GiiGjii} and \eqref{GijGkij} again, we have 
\be\label{y14}
G^{(1)}_{k\,2}=G^{(13)}_{k\,2}+\frac{G^{(1)}_{k\,3}G^{(1)}_{3\,2}}{G^{(1)}_{3\,3}}, \,\,\,\,\,\,\,\,
G^{(1)}_{2\,l}=G^{(13)}_{2\,l}+\frac{G^{(1)}_{2\,3}G^{(1)}_{3\,l}}{G^{(1)}_{3\,3}},\,\,\,\,\,\,\,\,
 G^{(1)}_{2\,2}=G^{(13)}_{2\,2}+\frac{G^{(1)}_{2\,3}G^{(1)}_{3\,2}}{G^{(1)}_{3\,3}}.
\ee
The last equality  implies
\be\label{y15}
\frac1{G^{(1)}_{2\,2}}=\frac1{G^{(13)}_{2\,2}}-\frac{G^{(1)}_{2\,3}
G^{(1)}_{3\,2}}{G^{(1)}_{3\,3}G^{(1)}_{2\,2}G^{(13)}_{2\,2}}.
\ee
This removes the dependence on the index $3$ of both the Green functions 
and their inverse in the last term in \eqref{y13}. 
Inserting \eqref{y16}--\eqref{y15} into \eqref{y13}, we obtain that if $k, l\notin\{1,2,3\}$
\be
G^{(1)}_{k\,l}=G^{(123)}_{k\,l}
+\frac{G^{(12)}_{k\,3}G^{(12)}_{3\,l}}{G^{(12)}_{3\,3}}
+\left(G^{(13)}_{k\,2}+\frac{G^{(1)}_{k\,3}G^{(1)}_{3\,2}}{G^{(1)}_{3\,3}}\right)
\left(G^{(13)}_{2\,l}+\frac{G^{(1)}_{2\,3}G^{(1)}_{3\,l}}{G^{(1)}_{3\,3}}\right)
\left(\frac1{G^{(13)}_{2\,2}}-\frac{G^{(1)}_{2\,3}G^{(1)}_{3\,2}}{G^{(1)}_{3\,3}G^{(1)}_{2\,2}
G^{(13)}_{2\,2}}\right).
\ee
So for $k,l  \notin \{1,2,3 \}$, we define $Q^{(1), \T}_{kl}$ as follows
\[
Q^{(1), (2, 3)}_{k l} = G^{(123)}_{k\,l}, \qquad 
Q^{(1), (2)}_{k l} = \frac{G^{(12)}_{k\,3}G^{(12)}_{3\,l}}{G^{(12)}_{3\,3}},
\qquad Q^{(1), (3)}_{k l} = \frac{G^{(13)}_{k\,2}G^{(13)}_{2\,l}}{G^{(13)}_{2\,2}},
\]
\begin{eqnarray}\label{ylong}
Q^{(1), \emptyset}_{k l} = &&
\frac{G^{(1)}_{k\,3}G^{(1)}_{3\,2}G^{(13)}_{2\,l}}{G^{(1)}_{3\,3}G^{(13)}_{2\,2}}
+\frac{G^{(13)}_{k\,2}G^{(1)}_{2\,3}G^{(1)}_{3\,l}}{G^{(1)}_{3\,3}G^{(13)}_{2\,2}}
+\frac{G^{(1)}_{k\,3}G^{(1)}_{3\,2}G^{(1)}_{2\,3}G^{(1)}_{3\,l}}{G^{(1)}_{3\,3}
G^{(1)}_{3\,3}G^{(13)}_{2\,2}}
-\frac{G^{(13)}_{k\,2}G^{(1)}_{2\,3}G^{(1)}_{3\,2}G^{(13)}_{2\,l}}{G^{(1)}_{3\,3}
G^{(1)}_{2\,2}G^{(13)}_{2\,2}}\\\nonumber
&&
-\frac{G^{(1)}_{k\,3}G^{(1)}_{3\,2}G^{(1)}_{2\,3}G^{(1)}_{3\,2}G^{(13)}_{2\,l}}{G^{(1)}_{3\,3}G^{(1)}_{2\,2}G^{(13)}_{2\,2}G^{(1)}_{3\,3}}
-\frac{G^{(13)}_{k\,2}G^{(1)}_{2\,3}G^{(1)}_{3\,2}G^{(1)}_{2\,3}G^{(1)}_{3\,l}}{G^{(1)}_{3\,3}G^{(1)}_{2\,2}G^{(13)}_{2\,2}G^{(1)}_{3\,3}}
-\frac{G^{(1)}_{k\,3}G^{(1)}_{3\,2}G^{(1)}_{2\,3}G^{(1)}_{3\,2}G^{(1)}_{2\,3}G^{(1)}_{3\,l}}{G^{(1)}_{3\,3}G^{(1)}_{2\,2}G^{(13)}_{2\,2}G^{(1)}_{3\,3}G^{(1)}_{3\,3}}.
\end{eqnarray}
One can see  that in this case,   $k, l\notin\{1,2,3\}$, \eqref{y123} holds and  $Q^{(1),(\T)}_{k\,l}$'s are independent of 
the rows (column) in $ {\T} \cup \{ 1 \} $. For $k=2, 3$  or $l=2, 3$  the previous  formulas for $Q$ do not make sense. 
But in this case, we do not need to decompose $ G^{(1)}$ in such fine details and we will use 
the simple decomposition 
\begin{align*}
 G^{(1)}_{2\,l} & =G^{(13)}_{2\,l}+\frac{G^{(1)}_{2\,3}G^{(1)}_{3\,l}}{G^{(1)}_{3\,3}},\,\,\,\,\,\,\,\, l\neq 3\,\,\,{\rm and}\,\,\,\, G^{(1)}_{2\,l}=G^{(1)}_{2\,3},\,\,\,\,\,\,\,\, l=3,\,\,\, \\
G^{(1)}_{3\,l} & =G^{(12)}_{3\,l}+\frac{G^{(1)}_{3\,2}G^{(1)}_{2\,l}}{G^{(1)}_{2\,2}},\,\,\,\,\,\,\,\, l\neq 2\,\,\,{\rm and}\,\,\,\, G^{(1)}_{3\,l}=G^{(1)}_{3\,2},\,\,\,\,\,\,\,\, l=2.\,\,\,
 \end{align*}
More precisely, we define  $Q^{(1), (\T)}$ by 
\begin{enumerate}
	\item For $k=2$ and $l\neq 3$, 
	$Q^{(1), (2,3)}_{k\,l}=Q^{(1), (2)}_{k\,l}=0$, $Q^{(1), (3)}_{k\,l}=G^{(13)}_{k\,l}$  
	and $Q^{(1), \emptyset}_{k\,l}=\frac{G^{(1)}_{k\,3}G^{(1)}_{3\,l}}{G^{(1)}_{3\,3}}$.
	
	\item For $k=2$ and $l= 3$, 
	$Q^{(1), (2,3)}_{k\,l}=Q^{(1), (2)}_{k\,l}=Q^{(1), (3)}_{k\,l}=0$ and  $Q^{(1), \emptyset}_{k\,l}=G^{(1)}_{k\,l}$.

\item For $k=3$ and $l\neq 2$, 
	$Q^{(1), (3,2)}_{k\,l}=Q^{(1), (3)}_{k\,l}=0$, $Q^{(1), (2)}_{k\,l}=G^{(12)}_{k\,l}$  
	and $Q^{(1), \emptyset}_{k\,l}=\frac{G^{(1)}_{k\,2}G^{(1)}_{2\,l}}{G^{(1)}_{2\,2}}$.
	
	\item For $k=3$ and $l= 2$, 
	$Q^{(1), (3,2)}_{k\,l}=Q^{(1), (3)}_{k\,l}=Q^{(1), (2)}_{k\,l}=0$ and  $Q^{(1), \emptyset}_{k\,l}=G^{(1)}_{k\,l}$.

\end{enumerate}
\noindent
Similarly, we can  define $Q^{(1), (\T)}$  for  the cases  $l=2$ or $l=3$.
 We now list the properties of  $Q^{(1), (\T)}_{kl}$ for  $k,l>1$ and  $\T\subset\{2,3\}$:
\begin{enumerate}
	\item $Q^{(1),(\T)}_{k\,l}$'s are independent of rows (column) in $ {\T} \cup \{ 1 \} $ and \eqref{y123} holds. 
	\item \be \label{ruleQ0}
	Q^{(1), (\T)}_{k\,l}=0\,\,\quad  \text{if}\,\,\, k\, \text{ or } \,l\in {\T}.
	\ee
	\item If   $k=l$ and ${\T}\cup\{k\}=\{2,3\}$, then 
	\be
	Q^{(1), (\T)}_{k\,l} = G^{(123)}_{k\,l}.
	\ee
	For all other cases,   $Q^{(1), (\T)}_{k\,l}$ is a  finite sum of  terms of the form:
	\be\label{y111}
	\frac{G_oG_o\cdots G_o}{G_dG_d\cdots G_d}
	\ee
	where each $G_o$ ($G_d$ resp.) represents some  off-diagonal (diagonal resp.) matrix element of  $G^{({\U})}$ with $\U$ some finite set.   Furthermore,  for $k\not = l$ or ${\T}\cup\{k\}\neq \{2,3\}$, the number of the off-diagonal elements in the numerator of \eqref{y111}  is 
	strictly bigger  than  $\left|\{2,3\}
	\backslash ({\T}\cup \{k, l\})\right|$.   Using Lemma \ref{boundGij},  in the set  $\Omega^c$ we have 
	\be\label{733}
	|Q^{(1), (\T)}_{k\,l}|\leq C \left(X^{\left|\{2,3\}\backslash({\T}\cup \{k, l\})\right|+1}+{\bf 1}({\T}\cup\{k\} =\{2,3\}, k=l)\right).
	\ee   
\end{enumerate}
Since the probability of the exceptional set $\Omega$ is extremely small, a simple argument which we repeated many times 
shows that it can be neglected in the estimate of the expectation in  \eqref{EQK}. Hence \eqref{EQK} follows from  \eqref{733}.

The proof of \eqref{EQK} shows clearly the approach to remove an element one by one from the Green function. 
 Define $R^{(1), (\T)}_{k\,l}$ as follows (like $Q$'s in \eqref{temp732})
 \be
R^{(1), (\T)}_{k\,l}\equiv  \sum_{{\U} : 
 {\T}\subset {\U}\subset \{2,3,4\} \backslash\{k, l\}} (-1)^{| {\U}|-|{\T}|}G^{(1  {\U})}_{k\,l}.
\ee 
Using the same method we used for $Q$'s, one can  prove the properties of $R$'s in Lemma \ref{Q1E}. The  details will be omitted since 
we will prove the general cases in the  next section.

\bigskip

\section{General case}\label{sec:gen}

The first step  to prove the general cases of  Lemma \ref{motN} is 
to extend the decomposition \eqref{KEP}. For any fixed $i$, $1\leq i\leq N$, 
and a fixed set ${\S}=\{i_1,i_2,\ldots, i_s\}$ such that $i\notin {\S}$, 
$1\leq i_j\leq N$, our goal is to decompose $Z^{(i)}_{ii}$ 
so that the following lemma holds: 

\begin{lemma}\label{decomK}
For $i\notin {\S}$, ${\T}\subset {\S}$ and $\eta\geq 1/N$,  there is a decomposition of 
\be\label{tempd39}
Z^{(i)}_{ii}= \sum_{{\T}\subset {\S}} \mathcal  Z^{(i),{\S},(\T)}, 
\qquad \mathcal Z^{(i),{\S},(\T)}\equiv \sum_{k,l}\overline \ba^{i}_k 
\mathcal G^{(i),{ {\S}},({\T})}_{k\,l}\ba^{i}_l.
\ee
such that   

\smallskip 
\noindent 
(1) $\mathcal G^{(i),{\S},({\T})}_{k\,l}$ is independent of the rows
 or columns of $H$ in $\{i\}\cup {\T}$, i.e., 
\be\label{propmG33}
\frac{\partial \mathcal G^{(i),{\S},({\T})}_{k\,l}}{\partial \ba^a_b}=0,\,\,\,
\frac{\partial \mathcal G^{(i),{\S},({\T})}_{k\,l}}{\partial \overline{\ba^a_b}}=0,
\,\,\,
a\in \{i\}\cup {\T},\,\,\,1\leq b\leq N.
\ee 
(2)  For any positive integer $k$, 
\be\label{temp8.3}
\E\left | \mathcal Z^{(i), {\S}, {(\T)}}  \right | ^k\leq C_{k,s}
 \left((\log N)^{3+2\al}\right)^k(X^{s-t+1})^{k}, \quad s= |\S|,\;\; t = |\T|.
\ee

\end{lemma}

In the applications, $\S$ will be the set of  indices, the
dependencies of which we wish to isolate in $Z^{(i)}_{ii}$. 
For example, for the case $i=1$  and  
$\S = \{2\}$ or  $\S=\{2,3\}$, respectively, if we define 
  \be\label{temp8.5}
\mathcal Z^{(1), \{2\}, {(\T)}}=
\ba^{1}  \cdot  P^{ (1), (\T)} \ba^1, \qquad
 \mathcal Z^{(1), \{2,3\}, {(\T)}}\equiv  \ba^{1} \cdot Q^{ (1), (\T)} 
\ba^1,
\ee
then \eqref{temp8.3} follows from \eqref{y7}, \eqref{y8} and \eqref{EQK}.

To achieve the decomposition \eqref{tempd39}, as in \eqref{temp732}
in  Section \ref{sec:4},
we start with a decomposition  on $G^{(i)}_{k\,l}$.  

\begin{definition}\label{defGiST} As in Lemma \ref{basicIG}, we use 
the notation $(i\, \T)$ for $(\{i\}\cup \T)$. 
For  $1\leq i\leq N$,  $i\notin {\S}=\{i_1,i_2,\ldots, i_s\}$  and  
 ${\T}\subset{\S}$, we define 
\be\label{temp85}
\mathcal G^{(i),{\S},({\T})}_{k\,l}\equiv\sum_{\U: 
 {\T}\subset {\U}\subset {\S} \backslash\{k,\,\, l\}}
 (-1)^{| {\U}|-|{\T}|}G^{ (i  \,  \T)}_{k\,l}.
\ee 
\end{definition}

For example, by \eqref{temp732}, for the case ${\S}=\{2,3\}$ and  $i=1$, we have $\mathcal G^{(1),\{2,3\},({\T})}_{k\,l}
=Q^{(1), ({\T})}_{k\,l}$; for the case ${\S}=\{2\}$ and  $i=1$, from \eqref{defP1} and \eqref{defP2} we have $\mathcal G^{(1),\{2\},({\T})}_{k\,l}=P^{(1), ({\T})}_{k\,l}$. 

From this definition one can easily check that  

\begin{enumerate}
	\item
	 \be\label{propmG2}
\mathcal G^{(i),{\S},({\T})}_{k\,l}=0, \quad  {\rm if } \; k {\,\,\,\rm or\,\,\,}l\in {\T}\cup\{i\}.
\ee\item 
For $k,l\notin  {\T}\cup\{i\}$, 
\be\label{propmG1} 
\mathcal G^{(i),{\S},({\T})}_{k\,l}=\mathcal G^{(i),{\S} \backslash\{k,\,\,l\},({\T})}_{k\,l}.
\ee
\item $\mathcal G^{(i),{\S},({\T})}_{k\,l}$ is independent of the rows or 
columns of $H$ in $\{i\}\cup {\T}$, i.e., 
\be\label{propmG3}
\frac{\partial \mathcal G^{(i),{\S},({\T})}_{k\,l}}{\partial \ba^a_b}=0,\,\,\,
\frac{\partial \mathcal G^{(i),{\S},({\T})}_{k\,l}}{\partial \overline{\ba^a_b}}=0,\,\,\,
a\in \{i\}\cup {\T},\,\,\,1\leq b\leq N.
\ee 
\item All  quantities defined so far depend on the initial matrix $H$,   omitted in our notations. 
If we wish to specify which  matrix is being considered, we will insert the matrix. For example, $\mathcal G^{(i),{\S \backslash \T},{\emptyset}}_{k\,l}(H^{(\T)})$ means it is defined w.r.t. $H^{({\T})}$ which is  the $N-|{\T}|$ by $N-|{\T}|$ minor of $H$ after removing the
 rows and columns in $\T$. Clearly, we have the relation 
\be\label{temp810}
\mathcal G^{(i),{\S},({\T})}_{k\,l}(H)=\mathcal G^{(i),{\S \backslash \T},{\emptyset}}_{k\,l}(H^{(\T)}).
\ee

\end{enumerate}
 
\bigskip

With these definitions, we can decompose $G^{(i)}_{k\,l}$ as follows. 

\begin{lemma}\label{expandG}
For fixed $i$, ${\S}=\{i_1,i_2,\ldots, i_s\}$ such that $i\notin {\S}$, we have the decomposition  
\be\label{resexpandG}
G^{(i)}_{k\,l}=\sum_{{\T}\subset {\S}}\mathcal G^{(i),{\S},{(\T)}}_{k\,l}.
\ee
\end{lemma}
{\it Proof. } Using the definition \eqref{temp85},  we have 
\be
\sum_{{\T}\subset {\S}}\mathcal G^{(i),{\S},{(\T)}}_{k\,l}
=\sum_{{\T}\subset {\S}}
\left(\sum_{\U: \T\subset {\U}\subset {\S} \backslash\{k,\,\,l\}} (-1)^{| {\U}|-|{\T}|}G^{(1\, {\U})}_{k\,l}\right)=\sum_{{\U}\subset {\S} \backslash\{k,\,\,l\}}\left(\sum_{{\T}\subset { {\U}}}(-1)^{|{\U}|-|{\T}|}\right)G^{(1 \,  \U)}_{k\,l}.
\ee
Since $ \sum_{{\T}\subset {{\U}}}(-1)^{| {\U}|-|{\T}|}=0$  unless $ {\U}=\emptyset$, we obtain \eqref{resexpandG} and 
this concludes Lemma \ref{expandG}. 
\qed

\bigskip

For the special case  $i=1$ and $\S= \{2, 3\}$,
 $\mathcal G^{(i),{\S},{(\T)}}_{k\,l}= Q^{(i), (\T)}_{kl}$ 
satisfies the estimate \eqref{733}. We now prove 
a general form of this estimate
on $\mathcal G^{(i),{\S},({\T})}_{k\,l}$.

\begin{lemma}\label{decomG}
 Let  $1\leq i\leq N$ and ${\T}\subset {\S}=\{i_1,i_2,\ldots, i_s\}$  
such that  $i\notin{ \S}$.
Then there exists a constant  $C$,  depending only on $s$, such that 
\be\label{Pmathgst}
\left|\mathcal G^{(i),{\S},({\T})}_{k\,l}\right|\leq C\left({\bf 1}({\T}\cup\{k\} ={\S}, k=l)+
X^{| {\S}\backslash({\T}\cup\{k,\,\,l\})|+1}\right),\,\,\,\,\,\,\,\,{\rm in} \,\,\,\Omega^c, 
\ee 
for sufficiently large $N$  depending only on $s$.
\end{lemma}

\medskip

This lemma is the basic estimate for a power counting argument.
It shows that the off-diagonal elements of $\mathcal G^{(i),{\S},({\T})}_{k\,l}$
are small by a certain power of $X$, which is our small parameter, depending
on the size of the sets $\S$ and $\T$. The diagonal elements,
when not zero by definition, are estimated by 1 (first term in \eqref{Pmathgst}),
but their contribution to the moments of $ \mathcal Z^{(i),{\S},(\T)}$
will be small since $k=l$ reduces the double sum in \eqref{tempd39} 
to a single sum.

\bigskip

\noindent 
{\it Proof of Lemma \ref{decomG}.}   For $k=l$, the estimate \eqref{Pmathgst}
follows directly from \eqref{temp85} and \eqref{cCTGkk}. We can thus assume that 
$k\ne l$ throughout the proof of this lemma.
The argument consists of two
parts. First we prove a representation formula (Lemma \ref{exiF}) that
asserts that $\mathcal G^{(i),{\S},({\T})}_{k\,l}$ is a certain rational
function involving resolvent matrix elements of $H$ and some of its
minors. The denominators in this rational function are products
of diagonal elements of resolvents and  the numerators are products of off-diagonal
matrix elements.
In the second step we will estimate these rational functions,
using that the diagonal elements of the resolvent 
 are typically separated
away from zero and the off-diagonal elements are small by a factor $X$.

For the precise argument,
we start with the cases:
\be\label{spcase}
{\T}=\emptyset , \,\,\,k,l\notin {\S}\,\,\,{\rm and }\,\,\,  {\S}\neq \emptyset.
\ee
The  special case  ${\S}=\{2\}$  can be proved by the representation
 \eqref{defP1} and Lemma \ref{boundGij}.  The case ${\S}=\{2,3\}$ 
was proved in \eqref{733}. 
These examples  show that $\mathcal G^{(i),{\S},{(\T)}}_{k\,l}$ can 
be written as  the finite sum of the terms of the form:
	\be
	\frac{G_o G_o\cdots G_o}{G_dG_d\cdots G_d},
	\ee
	where $G_o$  are off-diagonal elements of some $G^{({\U})}$ and $G_d$ 
 are diagonal elements. Furthermore, in each term, the number of the 
off-diagonal elements in the numerator is strictly greater than  
$s=|\S|$ but less than $4^{s}$. The  number of the diagonal elements 
in the denominator  is also less than $4^{|{\S} |}$.

The Green function 	$G^{(i,{\T})}_{k\,l}$ 
can be viewed as a function from the vector space of matrices.  
This motivates the  following definition. 

\begin{definition}
Denote by $\cX_K$ the space of  $K\times K$ matrices and $\cX=\cup_{K=1}^\infty \cX_K$. Define $\cY$ as the set of functions from 
 $\cX$  to the complex numbers. 
For any ${\S}=\{i_1,i_2,\ldots, i_s\}$ and for any  $i, k , l\notin{ \S}$,
 define the  set of  off-diagonal 
matrix elements considered as  functions of   matrices:  
\be\label{comcondjj'}
\mathcal A^{(i),{\S}}_{k\,l}  \equiv\left\{ f \in \cY:     f(W) = G^{(i\,\U)}_{j j'}(W),    \text{ for some } j\neq j',\,\,\,j,j\,'\in {\S}\cup \{k,l \},\;  {\U}\subset {\S}
\right\}, 
\ee
where $W \in \cX_K$ for some $K$. 
Similarly, we define  the set of diagonal matrix elements:
\be\label{comcondj}
\mathcal B^{(i),{\S}}_{k\,l}  \equiv\left\{f \in \cY:  
 f(W)= G^{(i\,\U)}_{j j}(W): j\in {\S}\cup \{k,l \},\,\,\,\,\,\, {\U}\subset {\S}\right\}.
\ee 
Furthermore  we define $\mathcal C^{(i),{\S}}_{k\,l} $ for all $k,l$ as 
\begin{align}
\label{tempdform}
\mathcal C^{(i),{\S}}_{k\,l} \equiv \Big\{  F\in\cY \;  & \text {is a finite sum of 
functions of the form} \;  \pm\frac{ f_1f_2\cdots f_m}{g_1g_2\cdots g_{m'}}: \nonumber   \\
&   f_\al\in \mathcal A^{(i),{\S}}_{k\,l},    1\leq \al \leq m; \,\,\,g_\beta\in 
\mathcal B^{(i),{\S}}_{k\,l}, \,\,\,
 1\leq \beta \leq m' ; 
\,\,\,s+1\leq m\leq 4^s,\,\,\,0\leq m'\leq 4^s\Big\},
\end{align}
where $s=|\S|$.
 \end{definition}
 
Notice the important condition $m\ge |\S|+1$ in the definition of 
$\mathcal C^{(i),{\S}}_{k\,l}$. Since off-diagonal matrix elements
are typically small, this requirement will guarantee the smallness 
of $\mathcal C^{(i),{\S}}_{k\,l}$ as a certain power of $X$.

With these notations,  the equation  \eqref{ylong} asserts that  for  
$k, l\notin  \{2,3\}$, there is a function
 $F^{(1),\{2,3\}}_{k,l} \in \mathcal C^{(1),\{2,3\}}_{k\,l}$ such that 
\be\label{faketemp820}
\mathcal G^{(1),\{ 2,3\},{\emptyset}}_{k\,l}  =F^{(1),\{2,3\}}_{k,l}.
\ee 
The general case is the following lemma.

\begin{lemma}\label{exiF}
For any ${\S}=\{i_1,i_2,\ldots, i_s\}$ with $s>0$ and  $i, k,l\notin { \S}$, 
there exists a function $F^{(i),{{\S}}}_{k,l}\in \mathcal C^{(i),{\S}}_{k\,l}$ such that 
\be \label{resexiF}
\mathcal G^{(i), {\S},{\emptyset}}_{k\,l} =F^{(i),{\S}}_{k,l}.
\ee
\end{lemma} 
 
 \bigskip

\noindent  
 {\it Proof of Lemma \ref{exiF}: }
  By symmetry, we only need to prove the cases that 
  $$i=1,\,\,\,\,\, {\S}=\{2,3,\ldots,  s+1\}.$$ 
To prove this case, we argue by induction on $s$.  For $s=1$ or $2$, 
Lemma \ref{exiF} was proved in \eqref{defP1} and \eqref{ylong} (cf. \eqref{faketemp820}). 
 Suppose  that Lemma \ref{exiF} is correct for $s=n-1\ge 1$ and 
$ F^{(1),\{2,\ldots,n\}}_{k,l}\in \mathcal C^{(1),\{2,\ldots,n\}}_{k,l}$ 
is the function satisfying \eqref{resexiF}
 for $i=1$ and ${\S}=\{2,3,\ldots,  n\}$
 
Now let $i=1$, ${\S}=\{2,\ldots,n+1\}$ and $k, l\notin \{1,\ldots,n+1\}$. By the induction assumption,  
\be\label{t1}
  \mathcal G^{(1),\{ {2,\ldots, n}\},{\emptyset}}_{k\,l} = F^{(1),\{2,\ldots,n\}}_{k,l}
\ee
with $F^{(1),\{2,\ldots,n\}}_{k,l}$ a finite sum of elements of  the form
\be\label{tempd55}
\pm\left(\prod_{\al=1}^m G^{(1\,{\U}_\al)}_{j_\al j'_\al}\right)
\left(\prod_{\beta=1}^{m'} G^{(1\,{\U}_{m+\beta})}_{j_{m+\beta} j_{m+\beta}}\right)^{-1},
\ee
where ${\U}_\al\subset {\S}$,   $n\leq m\leq 4^{n-1},\,\,\,0\leq m'\leq 4^{n-1}$    and 
$$
j_\al, j\,'_\al, j_{m+\beta}\in \{2,3,\ldots, n\}\cup\{k\}\cup\{l\}.
$$
By definition of  $\mathcal G^{(1),{ {\S}},({\T})}_{k\,l}$ in \eqref{temp85}, we have 

\be\label{932}
\mathcal G^{(1), \{2,...., n+1\}, \emptyset}_{kl} = \mathcal G^{(1), \{ 2,.. n\}, \emptyset}_{kl}
 - \mathcal G^{(1), \{2, ... ,n, n+1\}, (n+1)}_{kl}
\ee
Combining  \eqref{temp810} with \eqref{t1}, 
we have 
\be\label{t2} 
 \mathcal G^{(1),\{ {2,3\ldots, n, n+1}\},{(n+1)}}_{k\,l}(W)=F^{(1),\{2,\ldots,n\}}_{k,l}(W^{(n+1)}), 
 \ee
where $W^{(n+1)}$ is  the minor of $W$ with the  $(n+1)$-th  row and $(n+1)$-th column removed.

We can remove the dependence on the $(n+1)$-th row by the procedure in \eqref{y16}-\eqref{y15}.  Using 
\eqref{GiiGjii}, \eqref{GijGkij} and the notation:
\be
(1\,\U\,n+1)=\left(\{1, n+1\}\cup\U\right),
\ee
 we   have the expansion 
\be\label{tempd34}
G^{(1\,{\U_\al})}_{j_\al\,j'_\al}=G^{(1\,{\U_\al} n+1)}_{j_\al\,j'_\al}
+\frac{G^{(1\,{\U_\al})}_{j_\al\,n+1}G^{(1\,{\U_\al})}_{n+1\,j'_\al}}{G^{(1\,{\U_\al})}_{n+1\,n+1}},
\,\,\,\,\,\,\,\,\,\,\,\,\,1\leq \al\leq m
\ee
and 
\be\label{tempd36}
\frac1{G^{(1\,{\U_\beta})}_{j_\beta\,j_\beta}}=\frac1{G^{(1\,{\U_\beta}\, n+1)}_{j_\beta\,j_\beta}}
-\frac{G^{(i\,{\U_\beta})}_{j_\beta\,n+1}G^{(i\,{\U_\beta})}_{n+1\,j_\beta}}
{G^{(i\,{\U_\beta})}_{j_\beta\,j_\beta}G^{(i\,{\U_\beta}\, n+1)}_{j_\beta\,j_\beta}
G^{(i\,{\U_\beta})}_{n+1\,n+1}},\,\,\,\,\,m+1\leq \beta\leq m+k'.
\ee
We note that the first term on the r.h.s of \eqref{tempd34} is exactly the Green function 
on the  l.h.s of \eqref{tempd34} except that there is an additional superscript $n+1$; 
the  similar comment applies to \eqref{tempd36}.

Inserting \eqref{tempd34} and \eqref{tempd36} into \eqref{tempd55} and expanding it, we
 obtain that \eqref{tempd55} is equal to 
\be\label{temp824}
\pm\left(\prod_{\al=1}^m G^{(1\,{\U}_\al\, n+1)}_{j_\al j'_\al}\right)\left(\prod_{\beta=1}^{m'}
 G^{(1\,{\U}_{m+\beta}\,n+1)}_{j_{m+\beta} \,  j_{m+\beta}}\right)^{-1}+ \,\,\\{   \rm other \; terms} . 
\ee
Here the first term in \eqref{temp824} is the product of the first terms on the right  side of 
 \eqref{tempd34} and \eqref{tempd36} and it is the same as \eqref{tempd55} except that there is
 an additional superscript $n+1$.   One can see that  the other terms in \eqref{temp824} are
 elements in $\mathcal C^{(1),\{2,\ldots,n+1\}}_{k\,l}$, i.e., the number of the off diagonal 
terms in the numerator 
is now at least  $n+1$.  Since this procedure can be applied to each term in 
$ F^{(1),\{2,\ldots,n\}}_{k,l}$, 
we have proved  that there exists an $F\in  \mathcal C^{(1),\{2,\ldots,n+1\}}_{k\,l}$ such that 
\begin{align}\label{temp826}
 \mathcal G^{(1),\{ {2,\ldots, n}\},{\emptyset}}_{k\,l} (W)  =  F^{(1),\{2,\ldots,n\}}_{k,l}
\left(W \right)    & =F^{(1),\{2,\ldots,n\}}_{k,l} \left(W^{(n+1)}\right)+F(W)
\nonumber \\  & = \mathcal G^{(1),\{ {2,\ldots,  n+1}\},{(n+1)}}_{k\,l}(W)+F(W).
 \end{align}
By    \eqref{932} and \eqref{t2}, we can set  $
F^{(i),\{{2,3\ldots, n, n+1}\}}_{k,l} (W)= F(W)$
which is in $\mathcal C^{(1),\{2,3,\ldots,n,n+1\}}_{k\,l}$ and we have thus proved Lemma 
\ref{exiF} by induction. 
\qed

\bigskip 

\bigskip 
\noindent 
Now we start proving the estimates in Lemma \ref{decomG}.
Using \eqref{propmG2}  and \eqref{propmG1}, we only have to prove \eqref{Pmathgst} for 
the case $k,l\notin {\S}\cup\{i\}$.

\medskip 
\noindent 
{\it Case 1, ${\T}={\S}$:}  By  definition, 
\be
\mathcal G^{(i),\,{\S},\,({\T})}_{k\,l}=G^{(i\, {\S})}_{k\,l}.
\ee
Then \eqref{Pmathgst} in this special case follows from  Lemma \ref{boundGij}.

\medskip 
\noindent 
{\it Case 2, ${\T}=\emptyset,\,\,\,\,\,\, k,l\notin{\S}\,\,\,{\rm and}\,\,\,{\S}\neq \emptyset$:}
By Lemma \ref{exiF} and \ref{boundGij}, for any ${\S}
\neq \emptyset$ such that  $i, k,l\notin { \S}$, we have  
\be\label{temp834}
|\mathcal G^{(i),{\S},{\emptyset}}_{k\,l}|\leq C 
\frac{\left(\max_{{\U}\subset {\S}, j\neq j'}|G^{(i\,{\U})}_{jj'}|\right)^{s+1}}
{\left(\min_{{\U}\subset {\S}, j}|G^{(i\,{\U})}_{jj}|\right)^{4^{s}}}
\leq C X^{s+1},
\ee
where $C$ depends on $s=|\S|$.

\medskip 
\noindent 
{\it Case 3, 
${\T}\neq \emptyset, \,\,\,\,\,\,{\T}\subset {\S}, \,\,\,\,\,\,{\T}\neq {\S},\,\,\,\, k,l\notin{\S}\,\,\,{\rm and}\,\,\,{\S}\neq \emptyset$:}  
By  \eqref{temp810} and Lemma \ref{exiF},  there exists a function $F^{(i),{{\S\backslash\T}}}_{k,l}\in \mathcal C^{(i),{{\S\backslash\T}}}_{k\,l}$ (see  \eqref{tempdform}) such that
\be
\mathcal G^{(i),{\S},({\T})}_{k\,l}(H) =\mathcal G^{(i),{{\S\backslash\T}}}_{k\,l}(H^{(\T)})=F^{(i),{{\S\backslash\T}}}_{k,l}(H^{(\T)}),
\ee 
where $H^{({\T})}$ is the $N-|\T|$ by $N-|\T|$ minor of $H$ after removing the
rows and columns in $\T$.  Thus $\mathcal G^{(i),{\S},({\T})}_{k\,l}$ is given by the function
$F^{(i),{{\S\backslash\T}}}_{k,l}$ with all Green functions  $G^{({\U})}_{jj'}$ in the definition of $F^{(i),{{\S\backslash\T}}}_{k,l}$ 
replaced by   $G^{({\U}\cup {\T})}_{jj'}$. From \eqref{temp834} we have 
\be\label{temp8332}
|\mathcal G^{(i),{ {\S}},({\T})}_{k\,l}|\leq 
C \frac{\left(\max_{{\U}\subset {\S\backslash\T}, j\neq j'}|G^{(i\,{\U}\cup {\T})}_{jj'}|\right)^{|{\S\backslash\T}|+1}}{\left(\min_{{\U}\subset {\S\backslash\T}, j}|G^{(i\,{\U}\cup {\T})}_{jj}|\right)^{(4^{|{\S\backslash\T}|})}}.
\ee
Using  Lemma \ref{boundGij}, we have that 
\be\label{temp8342}
|\mathcal G^{(i),{ {\S}},({\T})}_{k\,l}|\leq C X^{|{\S\backslash\T}|+1},\,\,\,\,{\rm in}\,\,\,\Omega^c   
\ee
where $C$ depends on $s$. We have thus proved  \eqref{Pmathgst} for the Case 3 and this  completes the proof of Lemma \ref{decomG}.
\qed

\bigskip

\noindent 
{\it Proof of Lemma \ref{decomK}. } The decomposition \eqref{tempd39} follows from  \eqref{resexpandG} and 
 \eqref{propmG33} is a direct consequence of  \eqref{propmG3}. The estimate \eqref{temp8.3} 
can be proved in the same way 
 as in the proof of  Lemma \ref{y4} using the following three ingredients: 
 (1) The bounds on $\left|\mathcal G^{(i),{ {\S}},({\T})}_{k\,l}\right|$ in \eqref{Pmathgst}.
 (2)  The  large deviation estimate  in Lemma \ref{generalHWT}. 
(3)   The  trivial bound  $|\mathcal G^{(i),{ {\S}},({\T})}_{k\,l}|\leq C/\eta\leq CN$
 where $C$ depends on $|\S|$.
This  concludes  the proof of Lemma \ref{decomK}.    \qed

\bigskip

\noindent 
{\it Proof of Lemma \ref{motN}.}   
We first  introduce the following notations which will be useful  for the expansion of  the $p$-th moment of  $\left|\sum_{i=1}^NZ_i\right|$ in \eqref{52}.

\begin{definition}\label{CondSSV}
\begin{enumerate}
	\item Let ${\bf V}=\langle v_1,v_2,\ldots,v_{p}\rangle $ be a $p$ dimensional vector such that $v_i=0$ or $1$ for $1\leq i\leq {p}$. 
		\item Let ${\bS} =\langle \al_1,\al_2,\ldots,\al_{p}\rangle $ be a  $p$ dimensional vector such that 
		$1\leq \al_i\leq N$ for $1\leq i\leq p$.
		\item Denote by $\S $ the set consisting of elements $\alpha_j$ which is a component of $\bS$.  
\end{enumerate}
We define 
\be\label{tempd49}
A({\bS} , {\bf V})=\E\prod_{j=1}^{p}  \left( {\bf B}^{v_j} Z_{\al_j}\right), \quad   {\bf B}^1 (a+ i b) = a- i b, \quad  {\bf B}^0 (a+ i b) = a+ i b,
\ee
where $\bf B$ is the complex conjugate operator.
\end{definition}
Through the rest of this section,  $\S$ is always the set generated by $\bS$. 
Notice that $|\S| = s \le p$ where $p$ is the number of components in $\bS$. With these notations,  we can estimate   $\E\left|\sum_{i}Z_i\right|^{p}$ by 
\be\label{tempd50}
\E\left|\sum_{i=1}^NZ_i\right|^{p}\leq
\sum_{\bf {\bS}   } \sum_{\bf V}|A({\bS} , {\bf V})|
\leq C_p \sum_{ s\le p } N^{s}\max_{{\bS} , {\bf V}: |{\S } |=s}\left|A({\bf  {\bS} },{\bf V})\right|,
\ee
where   we sum up $ {\bS} $ and $\bf V$ under the conditions in Definition \ref{CondSSV}.
Lemma \ref{motN} is now a simple consequence of the following estimate on  $|A({\bS} , {\bf V})|$.

 \begin{lemma}\label{boundATS} 
Let  $\S$, ${\bS} $ and $\bf V$  satisfy the conditions in Definition \ref{CondSSV}.
With  $A({\bS} ,{\bf V})$ defined in \eqref{tempd49},   there exists a constant $C>0$ depending on $p$ such that 
\be\label{resboundATS}
\left|A({ {\bS} },{\bf V})\right|\leq C\left((\log N)^{3+2\al}\right)^{p} N^{ p-s} X^{2p},
\ee
for sufficiently large $N$ depending only on   $p$.
\end{lemma}

\begin{proof}    Let $ \S_i, 1\leq i\leq p$,  denote the  set $
 \S_i=  {\S }\backslash\{\al_i\}$. 
Using \eqref{tempd39}, we expand $A({\bf  {\bS} },{\bf V})$ as 
\begin{eqnarray}
A({  {\bS} }, {\bf V})&=&
\E\sum_{{\T}_1 \subset {\S}_1 } \ldots \sum_{{\T}_{p} \subset \S_p } A({\T}_1,{\T}_2,\ldots {\T}_{p}, {\bf V} ),\\\nonumber
A({\T}_1,{\T}_2,\ldots {\T}_{p}, {\bf V} )&
\equiv&
\left( {\bf B}^{v_1}\mathbb{IE}_{\al_1}\mathcal Z^{(\al_1),{\S_1}, (\T_1)}\right)
\left( {\bf B}^{v_2}\mathbb{IE}_{\al_2}\mathcal Z^{(\al_2),{\S_2}, (\T_2)}\right)
\cdots
\end{eqnarray}
From the   Schwarz inequality,  \eqref{temp8.3} and $|{\S_i}|=s-1$, we obtain that 
\be\label{tempd58}
\left|\E A({\T}_1,{\T}_2,\ldots {\T}_{p} , {\bf V})\right|\leq C\left((\log N)^{3+2\al}\right)^{p}  X^{\left(p s-\sum_{i=1}^{p} |{\T}_i|\right)},
\ee
where $C$ depends on $p$. 
Suppose that  
\be\label{stisss}
\sum_{i=1}^{p} |{\T}_i|\leq s p -2s.
\ee
Using \eqref{XX1N}, i.e., $ X^2\geq (\log N)^2/M\geq 1/N$, we have 
\be\label{typATTT}
|\E A({\T}_1,{\T}_2,\ldots {\T}_{p}, {\bf V} )|\le  C\left((\log N)^{3+2\al}\right)^{p}
  X^{2s}  \leq C\left((\log N)^{3+2\al}\right)^{p} N^{ p-s} X^{2p}.
\ee 

It remains to  estimate $\E A({\T}_1,{\T}_2,\ldots {\T}_{p} , {\bf V})$ for the cases that 
\be\label{istiss}
\sum_{i=1}^{p} |{\T}_i|\geq  s \, p -2s +1.
\ee
For $\gamma \in \S$, denote $n_\gamma$  to be the number of times that  
$\gamma$ appears in $\{ \al_1\} \cup {\T}_1$, $\{\al_2\} \cup{\T}_2,
\ldots$ and $\{\al_{p}\} 
\cup {\T}_{p}$, i.e., 
\[
n_\gamma = \sum_{k=1}^p {\bf 1}( \gamma \in \{ \al_k\} \cup {\T}_k) .
\]
By definition,  $n_\gamma \ge 1$.  Similarly, we define $m_\gamma$  to be the number of times
 that  $\gamma$ appears in $\langle \al_1, \al_2,
\ldots \al_{p} \rangle $, i.e., 
\[
m_\gamma = \sum_{k=1}^p {\bf 1}( \gamma =  \al_k) .
\]
Let $x= | \{\gamma \in \S: n_\gamma = p \}|$ and  $y= | \{\gamma \in \S:  m_\gamma = 1\}|$. 
Since for each fixed $i$, $\al_i\notin {\T}_i$, then with \eqref{istiss}
 and the definition of $n_\gamma$, 
\be
(p-1) (s-x)+ xp \ge \sum_{\gamma \in {\S }} n_\gamma=\sum_{i=1}^{p}|\{ \al_i\}\cup {\T}_i|
=p+\sum_{i=1}^{p} |{\T}_i|\geq s p -2s+ p+1.
\ee
By definition of $m_\gamma$, we have 
\be
y+ 2(s-y) \le \sum_{\gamma\in {\S }} m_\gamma=p.
\ee
From the last two inequalities, we have $x+ y \ge s+1$ and thus there exists a  $\gamma\in \S $ such that 
 \be\label{exij}
 n_\gamma=p\,\,\,{\rm and}\,\,\,m_\gamma=1.
 \ee
Without loss of generality, we assume that $\gamma=\al_1$. Then using  \eqref{exij}, we know 
\be\label{tempd63}
\gamma\neq \al_k,\,\,\,\gamma\in {{\T}_k},\,\,\,\,\,\,\,{\rm if\,\,\, } k\neq 1 .
\ee
Then with \eqref{tempd63}, the decomposition $ \mathcal Z^{(i),{\S},(\T)}\equiv \sum_{k,l}\overline \ba^{i}_k 
\mathcal G^{(i),{ {\S}},({\T})}_{k\,l}\ba^{i}_l$ \eqref{tempd39} and the property that $\mathcal G^{(i),{\S},({\T})}_{k\,l}$ is independent of the row or columns of $H$ in $\{i\}\cup {\T}$ \eqref{propmG33}, we have that  for $k\neq 1$,  the  $ \mathcal Z^{\al_k,{\S_k},(\T_k)}$  is independent of $\ba^\gamma$.  By the definition of $\mathbb{IE}$, for $k=1$, we also have 
\be
\E_{\ba^\gamma}\mathbb{IE}_{\ba^{\al_1}}\mathcal Z^{(\al_1),{{\S}_{1}},(\T_1)}=\E_{\ba^\gamma}\mathbb{IE}_{\ba^{\gamma}}\mathcal Z^{(\gamma),{{\S}_{1}},({\T}_{1})}=0. \,\,\,\,\,\,\,\,\,\,\,\,\,\,\,\,\,\,\,\,\,\,\,\,\,\,
\ee
Therefore, under the assumption \eqref{istiss} we have 
\begin{eqnarray*}
\E A({\T}_1,{\T}_2,\ldots {\T}_s , {\bf V})
&=&\E  \left({\bf B}^{v_1}\mathbb{IE}_{\ba^{\al_1}} \mathcal Z^{(\al_1),\S _1, (\T_1)}\right)
\left({\bf B}^{v_2}\mathbb{IE}_{\ba^{\al_2}}\mathcal Z^{(\al_2),\S_2, (\T_2)}\right)
\cdots \\
&=&\E\,\, \left(   \E_{\ba^\gamma}{\bf B}^{v_1}\mathbb{IE}_{\ba^{\al_1}}\mathcal Z^{(\al_1),{\S_1},(\T_1)}\right)
\left({\bf B}^{v_2}\mathbb{IE}_{\ba^{\al_2}}\mathcal Z^{(\al_2),{\S_2},(\T_2)}\right)
\cdots
=0.
\end{eqnarray*}
Combining this identity with \eqref{typATTT}, we obtain \eqref{resboundATS} and thus conclude Lemma \ref{boundATS}.  
\end{proof} 

\medskip

{\it Acknowledgement.} The authors thank Terry Tao for his helpful
comments on an earlier version of this manuscript.

\end{document}